\documentclass{IEEEtran}
\pdfoutput=1

\usepackage[latin9]{inputenc}
\usepackage[T1]{fontenc}

\usepackage{amsthm}
\usepackage{amsmath}
\usepackage{amssymb}
\usepackage{commath}              
\usepackage{graphicx}
\usepackage{comment}
\usepackage{color}
\usepackage{url}
\usepackage{algorithmic}
\usepackage{algorithm}
\usepackage[labelformat=simple]{subcaption}
\usepackage{authblk}
\usepackage[normalem]{ulem}
\usepackage{footmisc}

\DeclareCaptionLabelSeparator{periodspace}{.\quad}
\makeatletter                                

\numberwithin{equation}{section}
\numberwithin{figure}{section}

\theoremstyle{plain}
\newtheorem{thm}{\protect\theoremname}[section]
\theoremstyle{definition}
\newtheorem{defn}[thm]{\protect\definitionname}
\theoremstyle{remark}
\newtheorem{claim}[thm]{\protect\claimname}
\theoremstyle{plain}
\newtheorem{lem}[thm]{\protect\lemmaname}
\theoremstyle{remark}

\theoremstyle{plain}
\newtheorem{corollary}[thm]{\protect\corollaryname}
\theoremstyle{plain}
\newtheorem{proposition}[thm]{\protect\propositionname}
\theoremstyle{remark}

\providecommand{\claimname}{Claim}
\providecommand{\definitionname}{Definition}
\providecommand{\lemmaname}{Lemma}
\providecommand{\remarkname}{Remark}
\providecommand{\theoremname}{Theorem}
\providecommand{\corollaryname}{Corollary}
\providecommand{\propositionname}{Proposition}
\providecommand{\examplename}{Example}

\newcommand{\reals}{\mathbb{R}}
\newcommand{\RL}{\mathbb{R}^L}

\newcommand{\E}{\mathbb{E}}
\newcommand{\F}{F}

\newcommand{\DFT}{\operatorname{DFT}}
\newcommand{\SNR}{\operatorname{SNR}}

\DeclareMathOperator*{\argmin}{argmin}

\DeclareMathOperator*{\tr}{tr}
\DeclareMathOperator*{\Cov}{Cov}
\newcommand{\N}{{\mathcal{N}}}
\newcommand{\Z}{{\mathbb{Z}}}
\newcommand{\MSE}{\mathrm{MSE}}
\newcommand{\kk}{{\mathbf{k}}}
\newcommand{\OO}{{O}}
\newcommand{\ones}{{\mathbf{1}}}

\newcommand{\modminus}{-}
\newcommand{\modplus}{+}
\newcommand{\hatM}{\widehat{\hspace{-1pt}M}\phantom{\hspace{-11pt}M}}
\newcommand{\LL}{{\mathcal{L}}}
\newcommand{\Y}{\mathbf{Y}}
\newcommand{\X}{\mathbf{X}}

\def\aa{\mathbf{a}}
\newcommand{\error}{G}
\DeclareMathOperator*{\diag}{diag}
\DeclareMathOperator*{\rank}{rank}

\usepackage{xpatch}
\xpatchcmd{\algorithmic}
  {\arabic{ALC@line}}
  {\theALC@line}
  {}{}
\xpatchcmd{\algorithmic}{1.2em}{1.5em}{}{}

\usepackage{etoolbox}
\AtBeginEnvironment{algorithmic}{\let\textbf\relax}

\begin{document} 

\title{Multireference Alignment is Easier with an Aperiodic Translation Distribution}

\author[1,2]{Emmanuel Abbe}
\author[1]{Tamir Bendory}
\author[1]{William Leeb}
\author[1]{Jo\~ao M. Pereira}
\author[1]{Nir Sharon}
\author[1,3]{Amit Singer\thanks{
	EA was partly supported by the Bell Labs Prize, the NSF CAREER Award CCF--1552131, ARO grant W911NF--16--1--0051, NSF Center for the Science of Information CCF--0939370, and the Google Faculty Research Award. TB, WL, JP, NS, and AS were partially supported by Award Number R01GM090200 from the NIGMS, the Simons Foundation Investigator Award and Simons Collaboration on Algorithms and Geometry, the Moore Foundation Data-Driven Discovery Investigator Award, AFOSR FA9550-17-1-0291 and NSF BIGDATA Award IIS-1837992. WL is now at School of Mathematics, University of Minnesota, Twin Cities, MN, USA. NS is now at School of Mathematical Sciences, Tel Aviv University, Tel Aviv, Israel.
}}
 
 \affil[1]{The Program in Applied and Computational Mathematics,
    Princeton University, Princeton, NJ, USA}
\affil[2]{Electrical Engineering Department, Princeton University, Princeton, NJ, USA}
\affil[3]{Department of Mathematics, Princeton University, Princeton, NJ, USA}

\maketitle

\begin{abstract}
In the multireference alignment model, a signal is observed by the action of a random circular translation and the addition of Gaussian noise. The goal is to recover the signal's orbit by accessing multiple independent observations. Of particular interest is the sample complexity, i.e., the number of observations/samples needed in terms of the signal-to-noise ratio (the signal energy divided by the noise variance) in order to drive the mean-square error (MSE) to zero.
Previous work showed that if the translations are drawn from the uniform distribution, then, in the low SNR regime, the sample complexity of the problem scales as $\omega(1/\SNR^3)$. In this work, using a generalization of the Chapman--Robbins bound for orbits and expansions of the $\chi^2$ divergence at low SNR, we show that in the same regime the sample complexity for any aperiodic translation distribution scales as $\omega(1/\SNR^2)$. This rate is achieved by a simple spectral algorithm. We propose two additional algorithms based on non-convex optimization and expectation-maximization. We also draw a connection between the multireference alignment problem and the spiked covariance model.
\end{abstract}



\section{Introduction}
The problem of multireference alignment (MRA) arises in a variety of engineering and scientific applications, among them structural biology~\cite{diamond1992multiple,theobald2012optimal,park2011stochastic,park2014assembly,scheres2005maximum}, radar~\cite{zwart2003fast,gil2005using}, robotics~\cite{rosen2016certifiably} and image processing~\cite{dryden1998statistical,foroosh2002extension,robinson2009optimal}. In these applications, one aims to estimate a signal from its translated or rotated noisy copies. The problem also serves as a simplified model for more general problems like single-particle reconstruction by cryo--electron microscopy (cryo--EM), in which a three-dimensional density is recovered from two-dimensional projections taken at unknown viewing directions~\cite{bartesaghi20152,sirohi20163,singer2018mathematics}.

In this paper, we focus on the one-dimensional discrete MRA problem on a circle. In this model, we acquire $N$ measurements from the model
\begin{equation}\label{eq:mra}
Y_j = R_{S_j}x + \sigma G_j,\quad j= 1,\dots,N,
\end{equation}
where the $G_j$ are i.i.d and drawn from $\N(0,I_L)$, i.e. $G_j\in \RL$ and its entries are i.i.d standard Gaussian variables. The operator $R_{s}$ translates a signal $x\in\RL$ circularly by $s$ elements, namely, $(R_{s}x)[i]=x[i-s]$, where all indices should be considered as modulo $L$. The translations $S_j$ are i.i.d. and drawn from some unknown distribution $\rho$ on $\Z_L$. Figure~\ref{fig:intro} illustrates the MRA problem in different noise levels.

Previous approaches for estimating $x$ from~\eqref{eq:mra} can be broadly classified into two main categories. The first approach is based on estimating the translations $S_j$, aligning all observations, and averaging them to suppress the noise. However, alignment is too erroneous in low signal--to--noise ratio ($\SNR$)~\cite{aguerrebere2016fundamental,bendory2018toward}, defined here as $\SNR := \|x\|^2 / \sigma^2$. Note that while the translations $S_j$ are unknown, their estimation is not the primary goal of the problem. The translations are referred to as \emph{nuisance variables}.

An alternative approach aims at estimating the signal $x$ directly. Existing methods bypass the need to estimate the translations by employing expectation-maximization (EM) methods or by using features that are invariant under translation~\cite{bendory2017bispectrum}. Section~\ref{sec:related_work} is devoted to a detailed discussion on existing results and algorithms for MRA. In this paper, we take a different route by trying to estimate both the signal and the distribution of translations $\rho$ simultaneously. When $\rho$ is aperiodic, it turns out this is an easier problem than ignoring the fact that $\rho$ is not uniform and estimating $x$ alone.

In this paper we focus on the regime where both the number of observations and the variance of the noise are diverging. More specifically, our goal is to determine the sample complexity of \eqref{eq:mra}, which we define to be the minimal number of measurements, as a function of the SNR, required such that there is a sequence of estimators $\{\hat X_N\}$ of $x$ with mean square error (MSE) converging to $0$ as $N$ diverges. We define the MSE as
\begin{equation} \label{eq:mse_def}
\MSE=\frac1{\|x\|_2^2}\E\left[\min_{s\in\mathbb{Z}_L}\|R_s\widehat X - x\|^2_2\right],
\end{equation}
where the expectation is taken over the estimator $\widehat X$, which is a function of the random observations $Y_j$ with distribution determined by \eqref{eq:mra}. Allowing for a cyclic shift in \eqref{eq:mse_def} is intrinsic to the problem: if we apply a shift $R_s$ to $x$, and its inverse $R_{-s}$ to the right of $\rho$, we will produce exactly the same samples, thus there is no estimator $\hat X$ that is able to distinguish the observations that originate from $x$ and the ones from $R_s x$.

In~\cite{bandeira2017optimal}, it was proven that when $\rho$ is the uniform distribution, then in the low SNR regime, the sample complexity for estimating signals with non-vanishing discrete Fourier transform (DFT) is $\omega(1/\SNR^3)$. In this work, we show that if the translation distribution $\rho$ is aperiodic, meaning there is no $1 \le \ell \le L-1$ where $\rho[k+\ell] = \rho[k]$ for all $0 \le k \le L-1$, the sample complexity for estimating these signals is $\omega(1/\SNR^2)$. This rate is optimal and can be provably achieved by a spectral algorithm based on the first two moments of the data. The main result of this paper is stated as follows:

\vspace{0.3cm}
\textbf{Main Result (informal)}: \emph{Consider the model~\eqref{eq:mra} and suppose that $x\in\RL$ has a non-vanishing $\DFT$. When $\rho$ is aperiodic, the sample complexity of the MRA problem is $\omega(1/\SNR^2)$. This sample complexity is achieved by a spectral algorithm, detailed in Algorithms \ref{alg:spec00} and \ref{alg:SpectralAlg}, based on the first two moments of the data. Conversely, the sample complexity for any periodic distribution, in particular the uniform distribution, scales like $\omega(1/\SNR^3)$.
}

\vspace{0.3cm}

The proposed framework is based on a reliable estimation of the first two moments of the data. Hence, it requires only one pass over the measurements, low storage resources and is computationally efficient. To estimate the signal from the estimated moments, we propose, in addition to the aforementioned spectral algorithm, a non-convex least-squares (LS) algorithm. While the problem is non-convex, it empirically converges to the underlying signal, in the absence of noise, from a random initialization. As an alternative to the method of moments, we also examine an expectation-maximization (EM) algorithm.

The outline of the paper is as follows. Section~\ref{sec:related_work} provides a detailed discussion of existing results and algorithms for MRA.
In Section~\ref{sec:information_limit} we prove that the sample complexity is lower bounded by $\omega(1/\SNR^2)$. We also show that the sample complexity of any periodic distribution of translations with a period of less than $L/2$ scales as $\omega(1/\SNR^3)$. This is an extension of the results of~\cite{bandeira2017optimal}, which considered the uniform distribution case. 
In Section~\ref{sec:spectral_algorithm} we show that if the distribution is aperiodic, or is periodic with period $L/2$, then any signal with non-vanishing $\DFT$ can be estimated from its first and second moments, achieving the optimal estimation rate.
Section~\ref{sec:spike_model} draws the connections between the MRA model and the well-studied spiked covariance model~\cite{johnstone2001distribution, paul2007asymptotics, gavish-donoho-2017, benaych2012singular,dobriban2017optimal}.
Section~\ref{sec:algorithms} discusses and analyzes alternative algorithmic methods based on LS and EM.
Section~\ref{sec:numerics} examines the performance of the proposed algorithms by numerical simulations. 
Section~\ref{sec:conlusion} concludes the paper and proposes potential future extensions.

\begin{figure}     
    \centering    
    \includegraphics[scale=0.535]{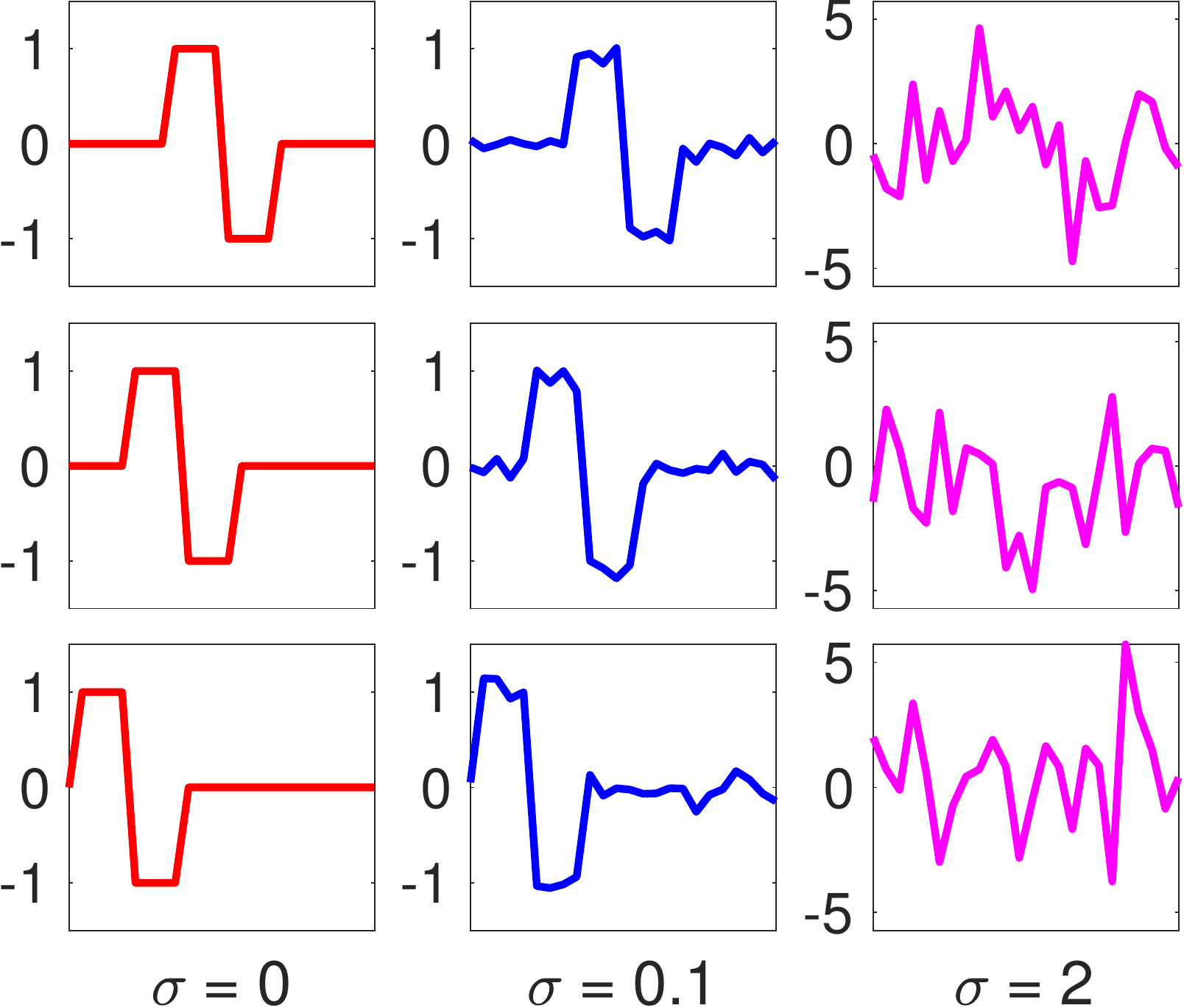}
    \protect    
    \caption{\label{fig:intro} The figures illustrates the MRA 
        measurements according to~\eqref{eq:mra}. The left column presents three measurements with different translations in the absence of noise. In this case, because the solution is defined up to translation, each measurement is a solution. The middle and right columns show measurements with the same translations and low and high noise levels, respectively (note the different scales of the y-axis).}
\end{figure}

Throughout the paper we use the following notation. We will use capital letters for random variables, and lower case letter for instances of this random variables. An estimator of a signal $z\in\RL$ is denoted by $\widehat{Z}$. We assume throughout that all signals are defined cyclically; that is, all indices should be considered modulo $L$. The indices range from $0$ to $L-1$. The DFT of $z$ is defined by $(\F z)[k] = \sum_{i=0}^{L-1}z[i] e^{-2\pi \iota k i /L}$, where $\iota=\sqrt{-1}$. We use $C_z$ for a circulant matrix whose first column is $z$, namely, $C_z[i,j] = z[i-j]$. A diagonal matrix whose diagonal is $z$ is denoted by $D_z$. We reserve $\E, \ast$ and $\odot$ for expectation, convolution and entry-wise product, respectively. The $L$--simplex is denoted by $\Delta^L$. That is to say, $z\in\Delta^L$ implies that $z[i]\geq 0$ for all $i$ and $\sum_{i=0}^{L-1}z[i]=1$.

\section{Related work} \label{sec:related_work}

\subsection{Multireference alignment via synchronization} \label{sec:alignment}

Given the translations $s_j$, the MRA problem~\eqref{eq:mra} is easy. One trivial unbiased estimator of $x$ is given by aligning all measurements and then averaging to suppress the noise, namely,
\begin{equation} \label{eq:knonw_translations}
    \widehat{X} = \frac{1}{N}\sum_{j=1}^NR_{s_j}^{-1}Y_j.
\end{equation} 
The variance of this estimator is $\sigma^2/N$ and therefore the number of measurements $N$ needs to scale like $\sigma^2$ to retain a constant estimation error. In other words, the sample complexity grows like $\omega(1/\SNR)$. One can replace~\eqref{eq:knonw_translations} with other estimators, such as James-Stein shrinkage~\cite{james1961estimation,efron1973stein,efron1975data}, which might improve the numerical performance with finite number of samples, but would not change the asymptotic sample complexity.
In practice, we do not have access to the underlying translations. However, if one can obtain a reliable estimation of the unknown translations ${\hat{s}_j}$, then one can estimate $x$ by the sample mean as in~\eqref{eq:knonw_translations} at sample complexity $\omega(1/\SNR)$. This motivates the design of synchronization methods that aim to estimate the translations $s_j$ from the data $y_j$.

A na\"ive approach for synchronization could be to fix one observation as a template, say $Y_1$, and estimate the relative translation of each $Y_j$, with respect to $Y_1$, by the peak of their cross-correlation:
\begin{equation*}
\hat{S}_j = \arg\max_s \sum_{i=0}^{L-1}Y_1[i]Y_j[i+s].
\end{equation*}
This approach may work in the high $\SNR$ regimes, but fails as the noise level increases (see for instance Figure I.1 in~\cite{bendory2017bispectrum}). 
Many alternative synchronization methods were proposed in the literature. For instance, the angular synchronization method aims at aligning all pairwise observations simultaneously~\cite{singer2011angular,boumal2016nonconvex,perry2018message,chen2018projected,bandeira2014tightness,zhong2018near}. Other methods propose to align through different semidefinite programs (SDPs)~\cite{bandeira2015non,bandeira2014multireference,chen2014near,bandeira2016low}. However, alignment is impossible at the low SNR regime, no matter how many measurements are acquired \cite{bendory2018toward}. For instance, for the continuous counterpart of~\eqref{eq:mra}, it has been shown that the Cr\'amer--Rao lower bound is proportional to $\sigma^2$ and does not depend on $N$. This bound holds even if the sought signal is known~\cite{aguerrebere2016fundamental}.

\subsection{Multireference alignment in low SNR}
\label{subsec:MRA_low_SNR}

This section reviews recent works on MRA in the low SNR regime, in which methods based on alignment fail.
The key idea is to estimate the signal directly, without estimating the translations beforehand. As will be emphasized throughout, previous works did not consider the translation distribution $\rho$, and either assumed or enforced it to be uniform. 

In~\cite{bandeira2017optimal}, it was shown that if the translations are uniformly distributed, namely, $S\sim \text{Uniform}[0,1,\dots,L-1]$, then the number of measurements needs to scale like $\omega(1/\SNR^3)$ for the estimator to converge in $L^2$ to the true signal. A follow-up paper~\cite{perry2017sample} showed that this rate can be achieved by a tensor decomposition algorithm.
The analysis of the uniform distribution is of particular interest since, no matter what $\rho$ is, one can always enforce it to be uniform. This can be done simply by reshuffling all measurements by $z_j=R_{S^\prime_j}y_j$, where $S^\prime_j$ are drawn from the uniform distribution. The new set of measurements $z_j$ obeys the MRA model~\eqref{eq:mra} with uniform translation distribution. However, as will be shown, this is in general a bad strategy, since the uniform distribution has a sample complexity scaling as $\omega(1/\SNR^3)$.

From the algorithmic point--of--view, a recent paper~\cite{bendory2017bispectrum} proposes a method that completely overcomes the need to estimate the translations. The core idea is to estimate features of the underlying signal that are invariant under cyclic translation. Particularly, it was proposed to estimate the mean, power spectrum and bispectrum of the signal from the moments of the data. Since these invariant features are polynomials in the signal with degree at most three, they can be estimated at sample complexity growing like $\omega(1/\SNR^3)$. Using these invariant features, one can recover the signal as $N \to \infty$ using a variety of algorithms~\cite{bendory2017bispectrum}. In~\cite{boumal2018heterogeneous}, it was shown that a similar technique can be used to estimate several signals simultaneously from heterogeneous samples (see also~\cite[Section~5]{perry2017sample}). Since the invariant feature technique requires only one pass over the data, it can be performed in a streaming mode, can be parallelized, requires low storage resources of $\mathcal{O}(L^2)$, and has low computational load. The framework proposed in this paper is also based on estimating moments of the data and therefore enjoys the same advantages; however, since we only require second-order moments, we bring the sample complexity down to $\omega(1/\SNR^2)$. 

Another approach for MRA is to apply an EM algorithm~\cite{dempster1977maximum}. EM is an iterative algorithm that aims to find the marginalized maximum likelihood estimator and is used ubiquitously in many statistical models. 
For the MRA model~\eqref{eq:mra}, and under the assumption that the translations are drawn from the uniform distribution, this algorithm takes a simple form and consists of two steps at each iteration~\cite{bendory2017bispectrum}. Given a current estimation $x_{k-1}$, the first step (called the E-step) computes a set of weights which can be understood as the translation distribution of each measurement $y_j$, if $x_{k-1}$ was the underlying signal. These weights are computed by 
\begin{equation*}
w_k^{\ell,j} = C_k^{j} e^{-\frac{1}{2\sigma^2}\| R_\ell x_{k-1}-y_j \|_2^2 },
\end{equation*}  
where $C_k^{j}$ is a normalization factor so that $\sum_{\ell} w_k^{\ell,j} = 1$. Then, the signal estimation is updated by marginalizing over the distributions and averaging (called the M-step):  
\begin{equation} \label{eq:standard_EM}
x_k = \frac{1}{N} \sum_{j=1}^{N}\sum_{\ell=0}^{L-1} w_k^{\ell,j}R_\ell^{-1}y_j.
\end{equation}
The EM algorithm enjoys an excellent numerical performance; however, its computational load and storage requirements are heavy since it passes through all the data at each iteration. In Section~\ref{sec:NU_EM}, we modify the standard EM algorithm to take the distribution into account.

\section{Information theoretic lower bound} \label{sec:information_limit}

In this section, we provide lower bounds for the MSE of an estimator of the signal in terms of the $\SNR$ and the number of observations $N$. In particular, we show that under mild conditions on the signal the MSE is bounded away from zero if $N=\Omega(1/\SNR^2)$. As described in Section~\ref{sec:spectral_algorithm}, the $\MSE$ of Algorithm~\ref{alg:SpectralAlg} converges to $0$ if the number of measurements grows like $\omega(1/\SNR^2)$. In addition, if the distribution is periodic, the MSE is bounded away from zero if $N=\OO(1/\SNR^3)$. The framework proposed in~\cite{bendory2017bispectrum} and described in Section~\ref{sec:related_work} achieves this sample complexity for any distribution.

Recall that we can estimate the signal only up to cyclic translation. 
We define the best alignment of $\widehat X$ with $x$ by
\begin{equation}\label{eq:Rxdef}
\phi_{x}(\widehat{X}) = \argmin_{z\in \{ R_s\widehat{X}\}_{s\in\Z_L}} \|z-x\| .
\end{equation}
Accordingly, we write \eqref{eq:mse_def} as
\begin{align}
\textrm{MSE}&=\frac{1}{\|x\|^2}\E\left[\|\phi_{x}(\widehat X)-x\|^2_2\right]. \label{eq:MSE+Rx}
\end{align}
Since we are interested in estimators that converge to a cyclic shift of $x$ in $L^2$ as $N$ diverges, we only consider estimators which are consistent, i.e., $\phi_{x}(\widehat X)\rightarrow x$ as $N \to \infty$. However the information lower bounds presented in this paper can be adapted to biased estimators (see Theorem~\ref{thm:ChapmanRobbins}). We now present the main results of this section as follows:
\begin{thm}\label{thm:corsigma} 
	Assume that $x$ is not a constant vector. 
	If $\widehat X$ is a consistent estimator of $x$, then
	\begin{equation}\label{eq:corsigma4}
	\MSE \ge \frac{1}{8N}\frac1{\SNR^2}-\OO\left(\frac1{N\SNR^{1.5}}\right).
	\end{equation}
	Moreover, if $\rho$ is periodic, with a period $\ell<\frac L2$, then
	\begin{equation}\label{eq:corsigma6}
	\MSE \ge \frac{1}{54N}\frac{L-2\ell}{2\ell}\frac1{\SNR^3}-\OO\left(\frac1{N\SNR^{2.5}}\right).
	\end{equation}
		
\end{thm}

Equation \eqref{eq:corsigma4} implies that if the number of measurements $N$ is $\OO(1/\SNR^2)$ and $\rho$ is aperiodic, the MSE is bounded away from $0$, thus the sample complexity is lower bounded by $\omega(1/\SNR^2)$. Similarly, when $\rho$ is periodic with a period $\ell<\frac L2$, \eqref{eq:corsigma6} implies that the sample complexity is lower bounded by $\omega(1/\SNR^3)$. 

Note that previous work~\cite{bandeira2017optimal} derived the sample complexity for the uniform distribution of translations. Theorem~\ref{thm:corsigma} extends it to any distribution. In addition, we extend~\cite{bandeira2017optimal} by providing the constant that multiplies $1/SNR^3$, for the uniform distribution case.

In the rest of this section, we develop the main tools required to prove Theorem~\ref{thm:corsigma}. Specifically, we start by introducing an auxiliary notation and definitions.
Then, in Section~\ref{sec:ChapmanRobbins}, we use an adaptation of the Chapman-Robbins lower bound~\cite{ChapRobb}, which is a generalization of the Cram\'er-Rao bound~\cite{CramerRao}, to derive a lower bound on the MSE in terms of the $\chi^2$ divergence. Then, in Section~\ref{sec:Fisherautocorr}, we express the $\chi^2$ divergence in terms of the Taylor expansion of the posterior probability density and the moment tensors . Finally in section \ref{sec:lowerboundMRA} we combine the results from Section~\ref{sec:ChapmanRobbins} and Section~\ref{sec:Fisherautocorr} to obtain a general lower bound for MRA, which we apply for the case when $\rho$ is aperiodic and periodic, respectively. The final details of the proof of Theorem~\ref{thm:corsigma} are given in Appendix~\ref{proof:corsigma}. 

\subsection{Notation and definitions}

Let $Y^{N}\in \reals^{L\times N}$ be the collection of all measurements as columns in a matrix. Let us denote by $f^{N}_{x,\rho}$ the probability density of the posterior distribution of $Y^{N}$, 
\begin{equation}\label{eq:parindepence}
f^{N}_{x,\rho}(y^N)=\prod_{j=1}^N f_{x,\rho}(y_j),
\end{equation} 
and the expectation of a function $g$ of the measurements under the measure $f^{N}_{x,\rho}$ by
\begin{equation*}
\E_{x,\rho}\left[g\left(Y^{N}\right)\right]:=\int_{\reals^{L\times N}} g\left(y^{N}\right) f^{N}_{x,\rho}\left(y^{N}\right)dy^{N}.
\end{equation*}
For ease of notation, we write $\E\left[g\left(Y^N\right)\right]$ when the signal and distribution are implicit. The bias-variance trade-off of the MSE is given by
\begin{equation}\label{eq:BiasVar}
\MSE=\frac{\tr(\Cov[\phi_x(\widehat X)])}{\|x\|^2}+
\frac{\|\E[\phi_{x}(\widehat X)]-x\|^2}{\|x\|^2},
\end{equation}
with
\begin{equation}\label{eq:Cov}
\Cov[\phi_x(\widehat X)]=\E\left[\phi_x(\widehat X)\phi_x(\widehat X)^T\right]-\E[\phi_{x}(\widehat X)]\E[\phi_{x}(\widehat X)]^T.
\end{equation}
For two symmetric matrices $A$ and $B$, we write $A\succeq B$ if the matrix $A-B$ is positive semidefinite (PSD).

We conclude this part with two definitions. First, we define the moment tensors. For a vector $x\in \RL$, we denote by $x^{\otimes d}$ the $L^{d}$ dimensional tensor where the entry indexed by $\kk=(k_1,\dots,k_d)\in \Z^d_L$ is given by $\prod_{j=1}^d x[k_j]$. The space of $d$-dimensional tensors forms a vector space, with sum and multiplication defined entry-wise. This vector-space has inner product and norm defined by $\left<A,B\right>=\sum_{\kk\in\Z_L^d}A[\kk]B[\kk]$ and $\|A\|^2=\left<A,A\right>$, respectively.
\begin{defn}\label{defn:autocorr}
	The $n$-th order moment of $x$ over $\rho$, is the tensor of order $n$ and dimension $L^{n}$, defined by
	\begin{equation*}\label{eq:autocorrelation}
	M^n_{x,\rho}:=\E\left[(R_S x)^{\otimes n}\right],
	\end{equation*}
	where $S\sim \rho$. 
\end{defn}

We will explore this notion in more detail in section \ref{sec:moments}, in particular we give explicit formulas for the moments when $n=1$ \eqref{eq:first_moment} and $n=2$ \eqref{eq:SecondMomentExpectation}.


Our last definition is of the $\chi^2$ divergence, which gives a measure of how ``far'' two probability distributions are.
\begin{defn}
	The $\chi^2$ divergence between two probability densities $f_A$ and $f_B$, with $f_A$ absolutely continuous with respect to $f_B$, is defined by
	\begin{equation*}
	\chi^2(f_A||f_B):=
	\E\left[\left(\frac{f_A(B)}{f_B(B)}-1\right)^2\right],
	\end{equation*}
	where $B\sim f_B$.
\end{defn}
The following lemma relates the $\chi^2$ divergence between $(x,\rho)$ and $(\tilde x,\tilde \rho)$ for one and $N$ observations.

\begin{lem}\label{lem:nchi2}
\begin{equation}\label{eq:nchi2}
\chi^2(f^N_{\tilde x,\tilde \rho}||f^N_{x,\rho})=(\chi^2(f_{\tilde x,\tilde \rho}||f_{x,\rho})+1)^N-1. 
\end{equation}
\begin{proof}
	See Appendix \ref{proof:eqnchi2}.
\end{proof}
\end{lem}


\subsection{Chapman-Robbins lower bound for an orbit}\label{sec:ChapmanRobbins}
The classical Chapman-Robbins gives a lower bound on an error metric of the form $\E[\|\widehat X -x\|^2]$, i.e., it does not take into consideration a translation-invariant error metric as appears naturally in the MRA problem. 
Hence, we modify the Chapman-Robbins bound to accommodate error of the form~\eqref{eq:MSE+Rx}. 
We point out that $\Cov[\phi_x(\widehat X)]$ is related to the $\MSE$ by~\eqref{eq:BiasVar}.
\begin{thm}[Chapman-Robbins for orbits]\label{thm:ChapmanRobbins}
	For any $\tilde x\in\RL$ such that $\phi_x(\tilde x)\neq x$ and $\tilde \rho\in \Delta^L$, we have
	\begin{equation*}
	\Cov[\phi_x(\widehat X)]\succeq \frac{z z^T}
	{\chi^2(f^N_{\tilde x,\tilde \rho}||f^N_{x,\rho})},
	\end{equation*}
	where $z=\E_{\tilde x,\tilde \rho}[\phi_{x}(\widehat X)]-\E_{x,\rho}[\phi_{x}(\widehat X)]$. 
\end{thm}
\begin{proof}
	See Appendix \ref{proof:ChapmanRobbins}.
\end{proof}

\subsection{Fisher information and moment tensors}
\label{sec:Fisherautocorr}
In this subsection we give a characterization of the $\chi^2$ divergence, which appears in the Chapman-Robbins bound, in terms of the moment tensors. 

Instead of considering the posterior probability density of $Y^N$, we will consider its normalized version $\widetilde Y^N=Y^N/\sigma$. We then have
\begin{equation}\label{eq:tildey}
\widetilde Y_j=\gamma R_{S_j} x + \error_j,
\end{equation}
where $\gamma=1/\sigma$, $S_j\sim \rho$ and $\error_j\sim \N(0,I)$. While this change of variables does not change the $\chi^2$ divergence, we can now take the Taylor expansion of the probability density around $\gamma=0$, that is,
\begin{equation}\label{eq:probdenstaylor}
f_{x,\rho}(y;\gamma)=f_\error(y)\sum_{j=0}^\infty \alpha^j_{x,\rho}(y)\frac{\gamma^j}{j!},
\end{equation}
\smallskip where $f_\error(y)=f_{x,\rho}(y;0)$ is the probability density of $\error_j$ (since when $\gamma=0$, $\widetilde Y_j=\error_j$) and
\begin{equation}\label{eq:alpha}
\alpha^j_{x,\rho}(y):=\frac1{f_\error(y)}\frac{\partial^j f_{x,\rho}}{\partial \gamma^j}(y;0),
\end{equation}
thus $\alpha^0_{x,\rho}(y)=1$. We note $f_{x,\rho}(y;\gamma)$ is infinitely differentiable for all $y\in \RL$, thus $\alpha^j_{x,\rho}(y)$ is always well-defined. We now use ~\eqref{eq:probdenstaylor} to give an expression of the $\chi^2$ divergence in terms of the moment tensors.

\begin{lem}\label{lem:falldominoes} 
	The divergence $\chi^2(f_{\tilde x,\tilde \rho}||f_{x,\rho})$ can be expressed in terms of the data moments as:
	\begin{align}
	\nonumber \chi^2(f_{\tilde x,\tilde \rho}||f_{x,\rho})\hspace{-30pt}&\\
	&=\frac{\sigma^{-2d}}{(d!)^2}\E\left[\left(\alpha_{\tilde x,\tilde \rho}^d(\error)-\alpha_{x,\rho}^d(\error)\right)^2\right]+\OO(\sigma^{-2d-1}), \label{eq:falldominoes}\\
	&=\frac{\sigma^{-2d}}{d!}\|M^d_{\tilde x,\tilde \rho}-M^d_{x,\rho}\|^2+\OO(\sigma^{-2d-1}), \label{eq:chi2autocorr}
	\end{align}
	where $d=\inf\left\{n:\|M^n_{\tilde x,\tilde \rho}-M^n_{x,\rho}\|^2>0\right\}$.
\end{lem}

\begin{proof}
	See Appendix \ref{proof:falldominoes}.
\end{proof}

Equation~\eqref{eq:falldominoes} is not specific to MRA: one can always obtain this expression as long we are considering the low $\SNR$ regime and the observations are independent of the signal in the limit of $\SNR$ tending to 0. The particularization to MRA happens in~\eqref{eq:chi2autocorr}, due to \eqref{eq:tildey} and \eqref{eq:alpha}. 

\subsection{General lower bound for the MRA problem}\label{sec:lowerboundMRA}

The following theorem is obtained from the results presented in the previous sections.
\begin{thm}\label{thm:metalowerbound}
	Consider the estimation problem given by equation \eqref{eq:mra}. For any signal $\tilde x\in \reals^L$ such that $\phi_x(\tilde x)\neq x$ and for any $\tilde \rho\in \Delta^L$, let $K^n_{\tilde x,\tilde \rho}=\frac1{n!}\|M^n_{\tilde x,\tilde \rho}- M^n_{x,\rho}\|^2$, 
	\linebreak $d_{\tilde x,\tilde \rho}=\inf\left\{n:K^n_{\tilde x,\tilde \rho}>0\right\}.$ and
	$$\bar d=\max_{(\tilde x,\tilde \rho):\, \phi_x(\tilde x)\neq x} d_{\tilde x,\tilde \rho}.$$ In other words, $\bar d$ is the smallest positive integer such that the moments $\{M^n_{x,\rho}\}_{n\le\bar d}$ define $x$ and $\rho$ unequivocally.
	Finally let
	$$\lambda^{m}_N=N/\sigma^{2m},\quad m\in \Z_+.$$
	We have
	\begin{equation}\label{eq:momentssanov}
	\MSE\ge \sup_{(\tilde x,\tilde \rho):\, d_{\tilde x,\tilde \rho}=\bar d}\left\{\frac{\|\phi_{x}(\tilde x)-x\|^2/\|x\|^2}
	{\exp\left(\lambda^{\bar d}_N K^{\bar d}_{\tilde x,\tilde \rho}\right)-1+\OO\left(\lambda^{\bar d}_N \sigma^{-1}\right)}\right\},
	\end{equation}
	thus the MSE is bounded away from zero if $\lambda^{\bar d}_N$ is bounded from above, or equivalently $N=\OO(1/\SNR^{\bar d})$.
\end{thm}

\begin{proof}
	We first note that $\bar d\le L$. Given all the tensor moments up to order $L$, we can write the polynomial $\prod_{i=1}^L (\alpha-x[i])$ in terms of entries of the moments, and find all the unordered entries of $x$ by taking the roots of the polynomial. To find the right order one can look at the value of other entries of the moments. Since $\bar d\le L$, the maximum is well defined.
	
	 By Theorem \ref{thm:ChapmanRobbins}, Lemma \ref{lem:falldominoes}, equations \eqref{eq:Cov} and \eqref{eq:nchi2} we obtain
	\begin{equation}\label{eq:limitexplanation}
	\MSE\ge \frac{\|z\|^2/\|x\|^2}{\left(1+\sigma^{-2d} K^{d}_{\tilde x,\tilde \rho}+\OO\left(\sigma^{-2d-1}\right)\right)^N-1}.
	\end{equation}
	with $z=\E_{\tilde x,\tilde \rho}[\phi_{x}(\widehat X)]-\E_{x,\rho}[\phi_{x}(\widehat X)]$. Since $\hat X$ is consistent, $\|z\|^2\rightarrow \|\phi_{x}(\tilde x)-x\|^2$ as $N$ diverges. On the other hand we have
	\begin{multline*}
	\left(1+\sigma^{-2d} K^{d}_{\tilde x,\tilde \rho}+\OO(\sigma^{-2d-1})\right)^N=\\\exp\left(\lambda^{d}_N K^{d}_{\tilde x,\tilde \rho}\right)+\OO\left(\lambda^{d}_N \sigma^{-1}\right)
	\end{multline*}
	and \eqref{eq:momentssanov} now follows from taking the supremum over $\tilde x$ and $\tilde \rho$.
\end{proof}
From Theorem \ref{thm:metalowerbound} we can obtain \eqref{eq:corsigma4} by providing $\tilde x$ and $\tilde \rho$ which have $M^1_{\tilde x,\tilde \rho}=M^1_{x,\rho}$, this implies $\bar d\ge 2$ and the $\MSE$ is bounded away from 0 if $N=\OO(1/\SNR^2)$. Moreover, to obtain \eqref{eq:corsigma6} when $\rho$ is periodic we can provide $\tilde x$ and $\tilde \rho$ which have $M^d_{\tilde x,\tilde \rho}=M^d_{x,\rho}$ for $d=1,2$, similarly to Proposition~\ref{prop:counterexample}, this implies $\bar d\ge 3$ and the $\MSE$ is bounded away from 0 if $N=\OO(1/\SNR^3)$.

However, when $N=\omega(1/\SNR^{\bar d})$ the supremum in \eqref{eq:momentssanov} is going to be achieved in the limit $(\tilde x, \tilde \rho)\rightarrow (x,\rho)$. Thus, to prove Theorem \ref{thm:corsigma}, we use intermediate results which explore the limit $(\tilde x, \tilde \rho)\rightarrow (x,\rho)$, and thus provide tighter bounds. However, since considering the limit introduces some technical details, we leave its analysis to Appendix~\ref{sec:derivatives}. The final details of the proof of Theorem \ref{thm:corsigma} are presented in Appendix~\ref{proof:corsigma}.

As a final remark on the results in this section, Theorem \ref{thm:metalowerbound} is not particular to MRA and can be easily generalized to a broader class of problems, which consider the actions of different groups and include the problem of single-particle reconstruction using cryo-EM \cite{abbe2018estimation}. A recent paper \cite{bandeira2017optimal} obtains similar results for minimax lower bounds, and employs techniques from computational algebra to estimate $\bar d$, and consequently the sample complexity, for a variety of models, such as cryo-EM and heterogeneous MRA.

\section{Provable algorithm based on the first two moments} 
\label{sec:spectral_algorithm}

In this section, we provide a spectral algorithm to estimate the signal, up to cyclic translation, from the first and second moments of the data, provided that the translation distribution is aperiodic. We prove that this algorithm estimates the signal exactly with high probability in the limit of SNR tending to $0$ with a growing number of samples; we will describe the asymptotic model more precisely in Section \ref{sec:spectral_method}. Because the method relies on only second-order information, its sample complexity in this case only grows like $\omega(1/\SNR^2)$, compared to sample complexity growing as $\omega(1/\SNR^3)$ if the translation distribution is periodic (with period smaller than $L/2$; see Section~\ref{sec:nonunique}). As we proved in Section~\ref{sec:information_limit}, $\omega(1/\SNR^2)$ is indeed the sample complexity for aperiodic distributions.
 
\subsection{Moments of $R_S x$} \label{sec:moments}

Before describing the algorithm, we will review a few basic properties of the moments of the random vectors $R_S x$, defined in Definition \ref{defn:autocorr}, and conclude with a theoretical result about the sufficient information they hold. 

We will first consider the first moment of the translated signal, $M^1= \E [R_S x]$, where $S \sim \rho$. This is equal to the convolution of $x$ with $\rho$; that is,
\begin{align} \label{eq:first_moment}
M^1 = x \ast \rho = C_x\rho = C_\rho x,
\end{align}
where $C_x$ is the circulant matrix with $x$ as its first column (and similarly for $C_\rho$). In this case, the convolution theorem implies
\begin{equation} \label{eqn:first_moment_fourier}
\F M^1 = \F{x}\odot\F{\rho},
\end{equation}
where $\odot$ and $\F$ denote entry-wise product and Fourier transform, respectively. We can estimate the first moment from the noisy observations~\eqref{eq:mra} by 
\begin{align} \label{eq:first_moment_est}
\hatM^1 = \frac{1}{N}\sum_{i=1}^N Y_i.
\end{align}
Note that if $L$ and $\sigma$ are fixed, then $\hatM^1$ is a consistent estimator of $M^1$ as $N \to \infty$.

The second moment of $R_S x$ is defined as 
\[ M^2 = \E \left[ (R_S x)(R_S x)^T \right] , \] 
where $S \sim \rho$. It can be verified that 
\begin{equation} \label{eq:SecondMomentExpectation}
M^2 = C_x D_\rho C_x^T,
\end{equation} 
where $D_\rho$ is a diagonal matrix of $\rho$. The unbiased second moment of $R_S x$ is then estimated from the observations $Y_j$ by:
\begin{align} \label{eq:second_moment_est}
\hatM^2 &= \frac{1}{N} \sum_{i=1}^N Y_i Y_i^T - \sigma^2 I,
\end{align}
where $I$ denotes the $L\times L$ identity matrix. As with the first moment, when $L$ and $\sigma$ are fixed then $\hatM^2$ is a consistent estimator of $M^2$ as $N \to \infty$.

We conclude this section with the following result, showing conditions which guarantee that there exists only one pair of signal and distribution (up to translation) that exactly agrees with the second moment data. Recall that a distribution $\rho$ is periodic if and only if there exists a period $1\le \ell<L$ such that
\begin{equation*}
\rho[k] = \rho[k+\ell] , \quad k=0,\ldots,L-1 . 
\end{equation*}
If no period exists we simply call $\rho$ aperiodic distribution.

\begin{thm} \label{thm:uniqueness}
Assume that $\rho_1$ is an aperiodic distribution, and that $x_1$ is a signal with non-vanishing DFT. Let $x_2$ and $\rho_2$ be any other signal and distribution with the same first two moments as $x_1$ and $\rho_1$. Then $x_2$ and $\rho_2$ are equal to $x_1$ and $\rho_1$, respectively, up to a shift. More precisely, there is $s \in \{0,\dots,L-1\}$ with $x_2 = R_s x_1$ and $\rho_2 = R_{-s} \rho_1$.
\end{thm}

The proof is given in Appendix~\ref{sec:proof_thm_uniqueness}. Next, we show a constructive method to recover $x$ and $\rho$ from their first two moments $M^1$ and $M^2$.
    
\subsection{Moment inversion when $\rho$ has a unique entry} 
\label{sec:moment-inversion}

The key observation driving the algorithm we will describe is that when $\rho$ has at least one distinct entry, and if $x$ has non-zero DFT, then $x$ can be recovered exactly from the first two moments $M^1$ and $M^2$.

We first note that the power spectrum of the signal, $P_x[k]:= \vert (\F x)[k]\vert^2$, is the Fourier transform of the signal's auto-correlation and thus can be derived directly from the second moment. 
Next, recall the factorization $M^2 = C_x D_\rho C_x^T$ from equation~\eqref{eq:SecondMomentExpectation}. The circulant matrix $C_x$ is diagonalized by the Fourier matrix $F$ as follows:
\[ C_x = F^{-1} D_{Fx} F, \]
thus we have
\begin{equation}\label{eq:powerspectrum}
F M^2 F^{-1} = \frac{1}{L} D_{Fx} C_{F \rho} D_{\overline{Fx}}.
\end{equation}
The $k$-th element of the diagonal of \eqref{eq:powerspectrum} is given by $$\frac1L(F \rho)[0]~\vert (\F x)[k]\vert^2=\frac1LP_x[k],$$
where $(F\rho)[0]=\sum_i \rho[i]=1$, since $\rho$ is a distribution. Consequently, we can obtain the power spectrum of $x$ from $M_2$ by
\begin{equation} \label{eq:psfrommoment}
	P_x=L\diag(F M^2 F^{-1}). 
\end{equation}
Now if we conjugate $M^2$ by the matrix $F^{-1} D_{1/|P_x|^{1/2}} F$, we obtain the matrix $\widetilde{M}^2 = C_{\tilde x} D_\rho C_{\tilde x}^T$, where $\tilde x$ is the vector with the normalized Fourier transform 
\begin{align} \label{eq:tilde_x}
    (F \tilde{x})[k] = \frac{(F x)[k]} {|(Fx)[k]|}.
\end{align}
Therefore, the matrix $C_{\tilde x}$ is both circulant and real orthonormal, i.e., $C_{\tilde x}^{-1} = C_{\tilde x}^T$. Consequently, the decomposition $\widetilde{M}^2 = C_{\tilde x} D_\rho C_{\tilde x}^T$ is an eigendecomposition of $\widetilde{M}^2$, and the eigenvectors are translations of $\tilde{x}$.

If $\rho$ has at least one distinct entry, then the associated eigenvector $v$ will be a translation of $\tilde{x}$, with arbitrary scaling; that is, $v = \alpha \cdot R_s \tilde{x}$ for some number $\alpha$ and shift $s$. Since the Fourier coefficients are still normalized, we multiply $\abs{P_x}^{1/2}$ and $\F v$ coordinate-wise to get
\[ \tilde{v} = \alpha \cdot \F^{-1} \left( \F(R_s \tilde{x}) \odot \abs{P_x}^{1/2} \right) 
= \alpha \cdot R_s x . \]
Letting $\text{Sum}(x)$ denote the sum of all elements in $x$, we have $\alpha = \text{Sum}( \tilde{v} ) / \text{Sum}(x)$. To uncover $\alpha$, note that the zeroth Fourier coefficient of $M^1 = x \ast \rho$ is $(F M^1)[0] = (Fx)[0] \cdot (F\rho)[0]$. But since $\rho$ is a probability vector, $(F \rho)[0] = 1$, and so $\text{Sum}(M^1) = (F M^1)[0] = (Fx)[0] = \text{Sum}(x)$. Consequently, $\alpha = \text{Sum}(\tilde{v} ) / \text{Sum}(M^1)$, and $R_s x = \tilde{v} / \alpha$.

Note that once we have determined $x$, we can also determine $\rho$ from $M^1 = x \ast \rho$ by deconvolution; indeed, since $M^1 = C_x \rho$, we have $\rho = C_x^{-1} M^1$. The algorithm is summarized in Algorithm~\ref{alg:spec00}.

\begin{algorithm}[ht]
    \caption{Exact recovery from the first two moments}
    \label{alg:spec00}
    \begin{algorithmic}[1]
        \REQUIRE Moments $M^1$ and $M^2$.
        \ENSURE The signal $x$ and distribution $\rho$.
        
        \COMMENT{ {\bf Normalize $Fx$}} \setcounter{ALC@line}{0}  \renewcommand{\theALC@line}{1.\arabic{ALC@line}}
        
        \STATE $P_x  \gets  L\diag(F M^2 F^{-1}) $ 
        \STATE $p  \gets  (P_x)^{-1/2} $ 
        \STATE $Q \gets \F^{-1} D_{p}\F$ 
        \STATE $\widetilde{M}^2 \gets Q M^2 Q^\ast$

        \COMMENT{ {\bf Extract eigenvector and rescale}} \setcounter{ALC@line}{0}  \renewcommand{\theALC@line}{2.\arabic{ALC@line}}

        \STATE $v \gets \operatorname{UniqEig}(\widetilde{M}^2)$ 
        \STATE $\tilde{v} \gets \F^{-1}\left(  (P_x)^{1/2} \odot \F v  \right)$ 
        \STATE $x \gets \left(  \operatorname{Sum}(M^1)/  \operatorname{Sum}(\tilde{v}) \right) \tilde{v}$
        \STATE $\rho \gets C_{x}^{-1}M^1 $
        \RETURN $x$ and $\rho$
    \end{algorithmic}
\end{algorithm}

We have proved the following result:
\begin{proposition}
Suppose $x$ has non-vanishing $\DFT$ and $\rho$ has at least one distinct entry. Let $M^1 = \E[R_S x]$ and $M^2 = \E[(R_S x)(R_S x)^T]$ be the first two moments. Then, Algorithm \ref{alg:spec00} returns the signal $x$ and the distribution $\rho$ exactly (up to cyclic translation).
\end{proposition}

\subsection{Estimating $x$ in low SNR}
\label{sec:spectral_method}

Section \ref{sec:moment-inversion} shows that Algorithm \ref{alg:spec00} recovers $x$ exactly from the exact values of $M^1$ and $M^2$, as long as the DFT of $x$ is non-vanishing and $\rho$ has at least one distinct entry. In this section we show that under the same conditions, Algorithm \ref{alg:spec00} is stable under small perturbations of the moments. We also show that if $N=\omega(\sigma^4)$, or equivalently $N=\omega(1/\SNR^2)$, the MSE of the estimate given by Algorithm \ref{alg:spec00} converges to $0$ as $N$ diverges.

We first observe that whenever $\rho$ is aperiodic, we can modify the observations to assume that $\rho$ in fact has all distinct entries. Indeed, we generate a new set of measurements $z_j = R_{S'_j}y_j$, where $S'_j$ are drawn from a new, known distribution $\theta$. In this case, the translations are distributed according to $\rho \ast \theta$. The following lemma shows that by choosing $\theta$ as a random probability distribution on the simplex, we can ensure that all entries of $\rho \ast \theta$ are distinct with probability $1$. Note that if the DFT of $\theta$ is non-vanishing (which holds with probability 1 for random $\theta$), then one can recover fully $\rho$ from $\rho \ast \theta$.

\begin{lem} \label{prop:reshuffling}
Let $\rho$ be an aperiodic vector on the simplex and let $\theta$ be a random probability density function on the simplex. Then, all entries of $\rho\ast\theta$ are distinct with probability 1.
\begin{proof}
See Appendix~\ref{sec:proof_prop_reshuffling}.
\end{proof}
\end{lem}
Using this lemma, we will assume from now on that all entries of $\rho$ are distinct. The following corollary states that Algorithm \ref{alg:spec00} is stable to perturbations of the moments and power spectrum:

\begin{corollary}
	\label{cor:stability}
	Suppose $x$ has non-vanishing $\DFT$ and denote by $\hatM^1$ and $\hatM^2$ the sample moments defined by equations \eqref{eq:first_moment_est} and \eqref{eq:second_moment_est}. Suppose that $\| \hatM^1 - M^1\|_{\text{\emph{F}}} \le \varepsilon$ and $\| \hatM^2 - M^2 \|_{\text{\emph{F}}} \le \varepsilon$, for sufficiently small $\varepsilon > 0$. Then {Algorithm~\ref{alg:spec00}}, with input data $\hatM^1$ and $\hatM^2$, returns an estimate $\widehat{X}_{\text{Spectral}}$ of $x$ with error at most $C \varepsilon$, where $C$ is a finite and positive constant which depends only on $x$ and $\rho$.
	
\end{corollary}

\begin{proof}
	See Appendix~\ref{sec:proof_cor_stability}.
\end{proof}

The following theorem shows that if $N$ grows like $\omega(\sigma^4)$, the MSE of the estimator converges to $0$ as $N$ diverges.

\begin{thm}
	\label{thm:exp_decay}
	
	If $N=\omega(\sigma^4)$, the MSE of $\widehat{X}_{\text{Spectral}}$, defined in Corollary \ref{cor:stability}, converges to $0$ as $N$ diverges.

\end{thm}

\begin{proof}
	See Appendix~\ref{sec:proof_spectral l2 convergence}.
\end{proof}

Algorithm \ref{alg:SpectralAlg} describes the entire pipeline for estimating $x$ from the noisy measurements~\eqref{eq:mra}, including randomly shifting the observations, estimating the moments, and using Algorithm \ref{alg:spec00} to estimate $x$ from the estimated moments.

\begin{algorithm}[ht]
    \caption{Estimating $x$ and $\rho$ from noisy data}
    \label{alg:SpectralAlg}
    \begin{algorithmic}[1]
        \REQUIRE  $y_j$, $j=1,\ldots,N$ of~\eqref{eq:mra} and noise variance $\sigma^2$.
        \ENSURE An estimated signal $\hat{x}$ and estimated distribution $\hat{\rho}$.
        
        \COMMENT{{\bf Reshuffling observations (optional)}} \setcounter{ALC@line}{0}  \renewcommand{\theALC@line}{1.\arabic{ALC@line}}
        \STATE draw a random distribution $\theta \in \Delta^L $
        \STATE for each $j=1,\dots,N$: $y_j \gets R_{S^\prime_j}y_j $ for $S^\prime_j \sim\theta$

        \COMMENT{ {\bf Moment estimation}} \setcounter{ALC@line}{0}  \renewcommand{\theALC@line}{2.\arabic{ALC@line}}
        \STATE $\hatM^1 \gets \frac{1}{N} \sum_{j=1}^N y_j$
        \STATE $\hatM^2 \gets \frac{1}{N} \sum_{j=1}^N y_j y_j^T -\sigma^2 I $ \label{line:ps1}
                
        \COMMENT{ {\bf Eigendecomposition and normalization}} \setcounter{ALC@line}{0}  \renewcommand{\theALC@line}{3.\arabic{ALC@line}}
        
        \STATE obtain $\hat{x}$ and $\hat{\rho}^\prime$ from Algorithm \ref{alg:spec00} with $\hatM^1$ and $\hatM^2$.
        \RETURN $\hat{x}$ and $\hat{\rho} = C_\theta^{-1}\hat{\rho}'.$
    \end{algorithmic}
\end{algorithm}

\subsection{Non-uniqueness for periodic $\rho$}
\label{sec:nonunique}

We have shown that the first and the second moments suffice to determine the signal if the distribution is aperiodic. In this section, we provide a complementary result, showing that if the distribution is periodic, then having the first two moments is not enough to uniquely determine a signal with non-vanishing $\DFT$. In particular, given a distribution $\rho$ with period $\ell$, a signal $x_2$ has the same first two moments as $x_1$ if it satisfies: 
\begin{equation} \label{eq:construction_example} 
(\F{x_2})[k] = \left\{
\begin{array}{ll}
(\F{x_1})[k],  & k=t\frac{L}{\ell}, \quad t=0,\ldots,\ell-1 , \\
-(\F{x_1})[k], & \mbox{  otherwise. }
\end{array}
\right.
\end{equation}
This construction is demonstrated in Figure~\ref{fig: example_construction}. 

\begin{figure*}[ht]
    \centering
    \begin{subfigure}[t]{0.22\textwidth}
        \includegraphics[scale=.2]{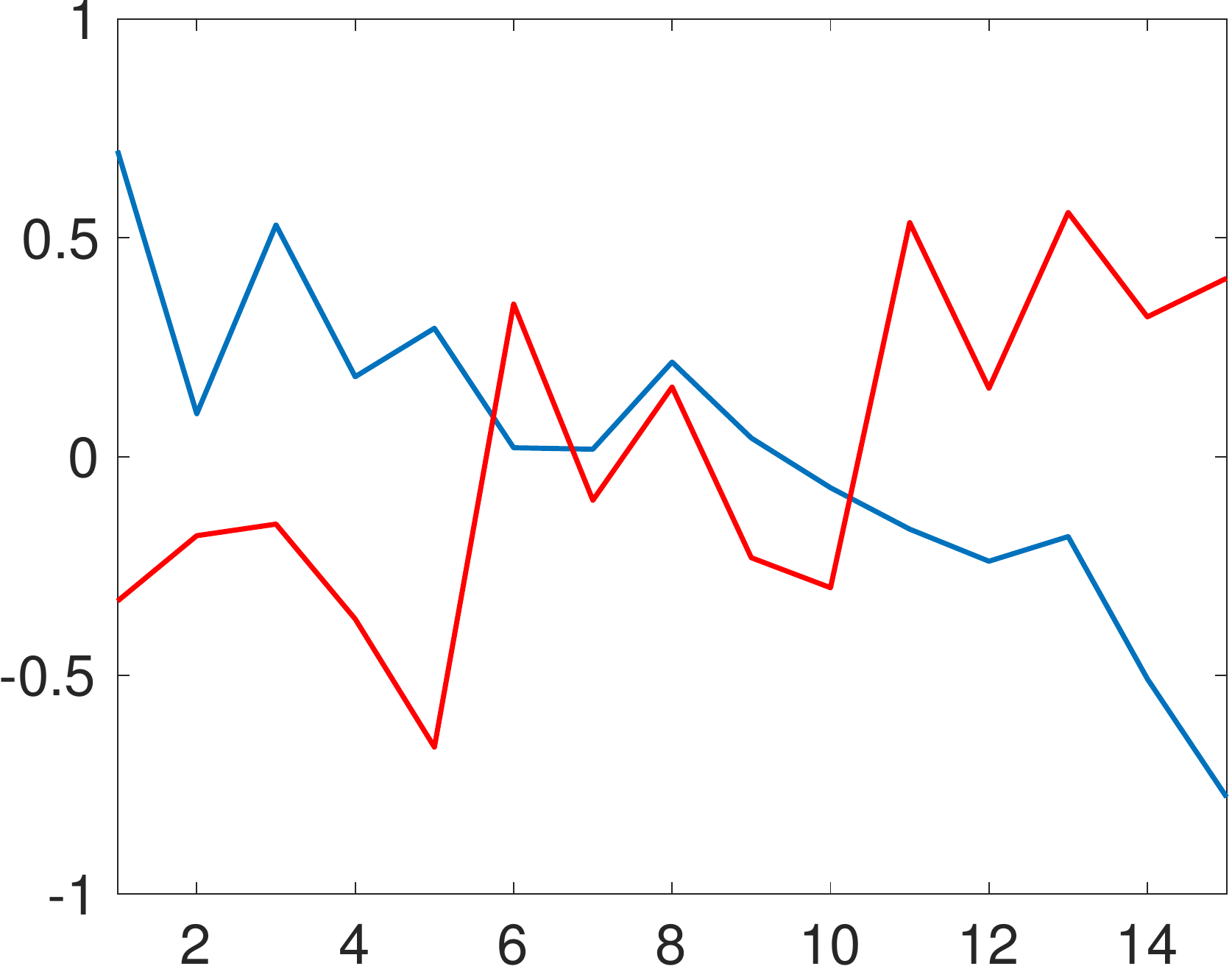}
        \caption{The two different real signals $x_1$ and $x_2$ of length $15$}    \label{fig:two_signals}
    \end{subfigure}   \quad
    \begin{subfigure}[t]{0.22\textwidth}
        \includegraphics[scale=0.2]{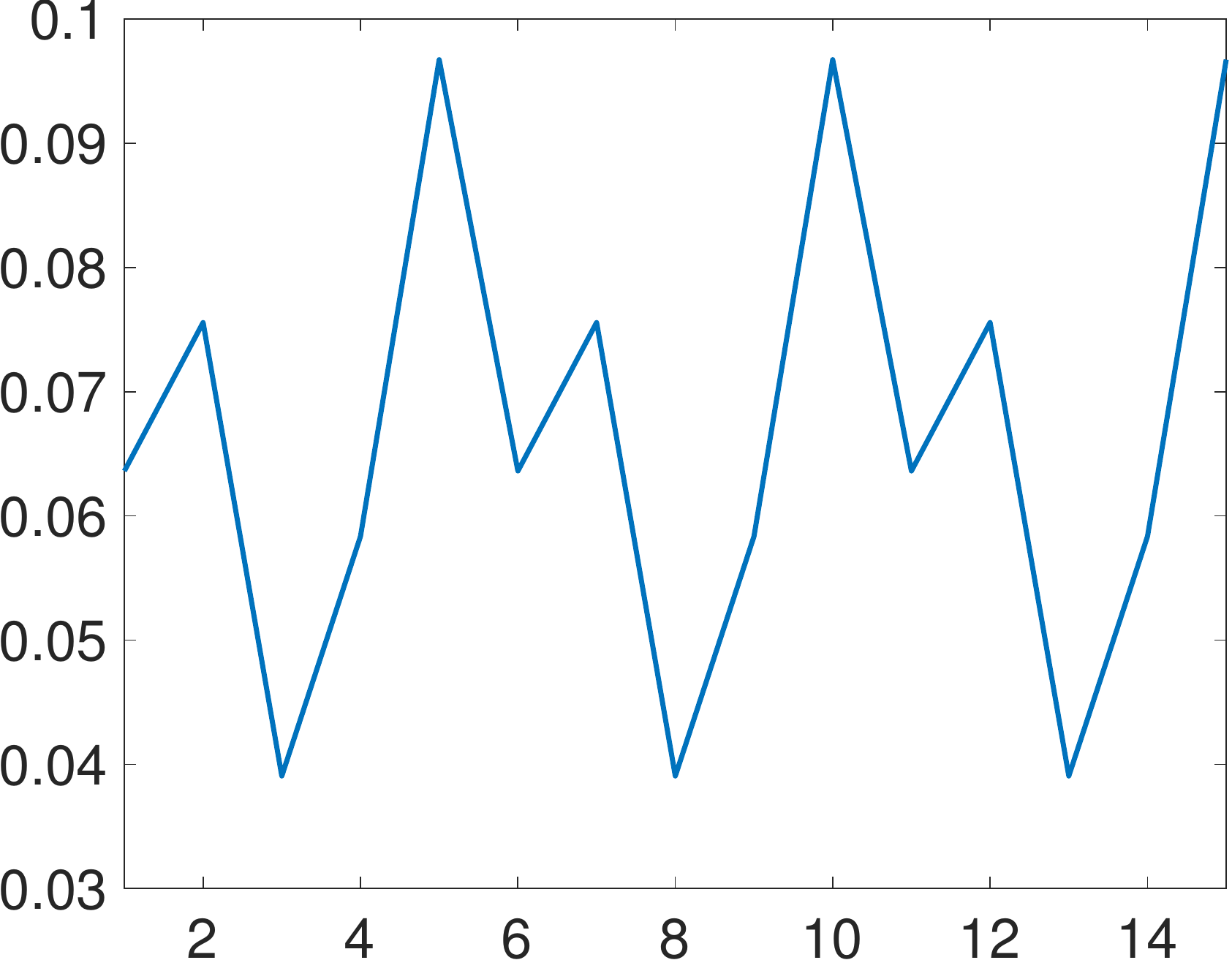}
        \caption{The 5-periodic distribution} \label{fig:periodic_rho}
    \end{subfigure}      \quad
    \begin{subfigure}[t]{0.22\textwidth}
        \includegraphics[scale=0.2]{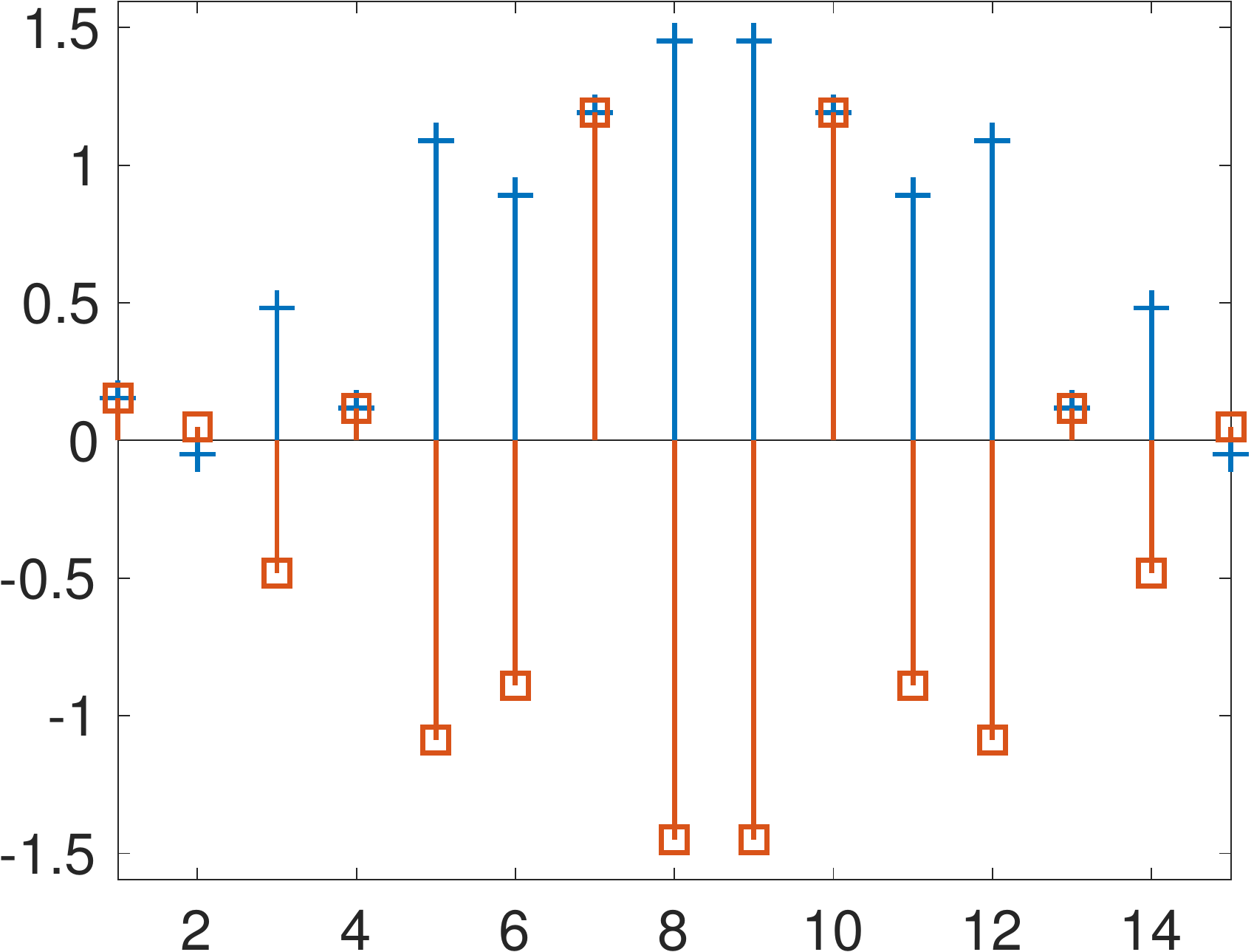}
        \caption{The real parts of $\F x_1$ and $\F x_2$}    \label{fig:real_part}
    \end{subfigure}  \quad
    \begin{subfigure}[t]{0.22\textwidth}
        \includegraphics[scale=0.2]{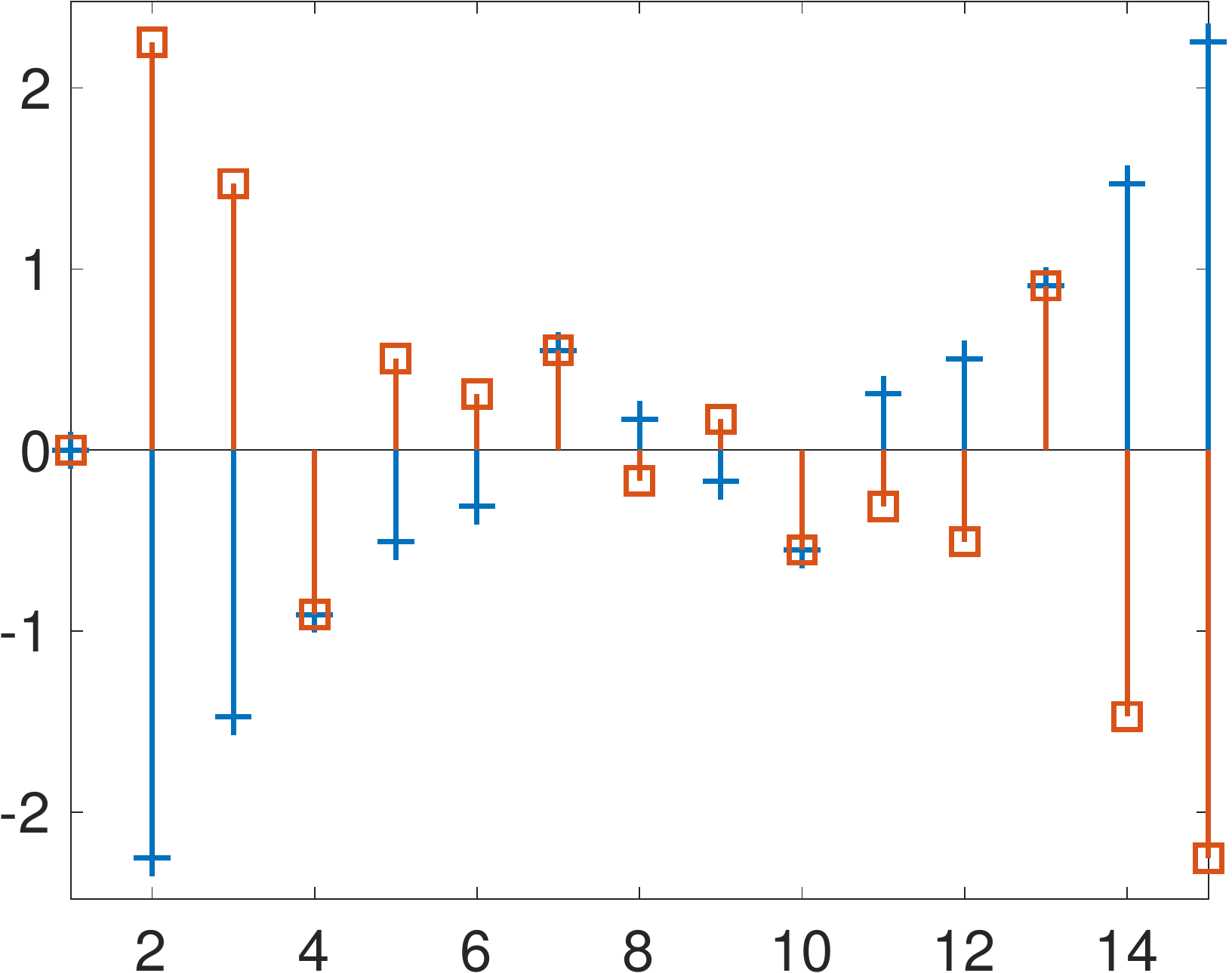}
        \caption{The imaginary parts of $\F x_1$ and $\F x_2$}    \label{fig:image_part}
    \end{subfigure}  
    \caption{\label{fig: example_construction}This example demonstrates the constriction of~\eqref{eq:construction_example} and Proposition~\ref{prop:counterexample}. The figures present two different real signals of length $15$ and a 5-periodic distribution. The Fourier transforms of the signals obey~\eqref{eq:construction_example}. 
        The two signals have the same first two moments under the periodic distribution. } 
\end{figure*}

\begin{proposition} \label{prop:counterexample}
    Let $\ell<L/2$ be a divisor of $L>1$. Suppose that $\rho$ is periodic, with period $\ell$, and let $x_1$ be a given signal with non-vanishing $\DFT$. Then the signal $x_2$ defined by \eqref{eq:construction_example} is not a translation of $x_1$, and has the same first and second moments as $x_1$. Therefore, if the distribution is periodic, then any signal with non-vanishing $\DFT$ is not uniquely determined from its first two moments.
\end{proposition}
    \begin{proof}
        See Appendix~\ref{sec:proof_prop_counterexample}.
    \end{proof}
In Section~\ref{sec:information_limit} we established this result from an information-theoretic perspective by showing that the sample complexity for periodic distribution grows like $\omega(1/\SNR^3)$, and extending~\cite{bandeira2017optimal} that considered only the uniform distribution. Indeed, 
the uniform distribution is merely a special case of periodic distributions with minimal period $\ell = 1$. When $\ell>1$, one can interpret the periodicity as having a uniform distribution over the different cosets of $\mathbb{Z}_L$ with respect to the subgroup generated by a translation in $\ell$ coordinates. These cosets are exactly the analogue of the sparsity pattern of ${\F{\rho}}$ attained by jumps of $L/\ell$. This also explains why uniformity is the only pathological case for a prime $L$. Therefore, if one can choose how to sample the signal, a prime number of samples should be considered. 

As it turns out, there is one special case where the first two moments are enough to determine $x$ uniquely, up to cyclic translation, even when $\rho$ is periodic. This special case occurs when $L$ is even and $\rho$ is $L/2$-periodic. Note that in this case the information theoretic lower bound presented in section \ref{sec:information_limit} is also $\omega(1/\SNR^2)$. This result is formulated in the following claim:
\begin{claim} \label{claim:L_2_period}
    Suppose that $x$ has non-vanishing $\DFT$, $L$ is even and $\rho$ is $L/2$-periodic. Then, $x$ is uniquely determined from its first two moments, up to global translation.
    \begin{proof}
        See Appendix~\ref{sec:proof_L_2_period}.
    \end{proof} 
\end{claim}

\section{Connection with the spiked covariance model } \label{sec:spike_model}

In this section, we point out a connection between the spectral algorithm presented in Section~\ref{sec:spectral_algorithm}, and the spiked covariance model well-known in statistics~\cite{johnstone2001distribution, paul2007asymptotics, gavish-donoho-2017, benaych2012singular,dobriban2017optimal}. Though somewhat informal, this analysis will provide insight into how the complexity of recovering $x$ depends on the dimension $L$ when the distribution $\rho$ has a fixed support size.

In the spiked model, we observe a matrix
\begin{align}
    \Y = \X + \mathbf G \in \mathbb{R}^{L \times N},
\end{align}
where $\X$ is a rank $r$ matrix and 
\[  \mathbf G=(G_{ij}) , \quad   G_{ij} \overset{iid}{\sim} \mathcal{N}(0,\sigma^2). \] 
This model is typically studied in the \textit{high-dimensional regime}, in which $L$ grows proportionally to $N$; that is, $L = L(N)$ and $L / N \to \gamma > 0$ as $N \to \infty$. In this setting, there is a precise understanding of the limiting behavior of the data matrix $\Y$ and the low-rank matrix $\X = [X_1,\dots,X_N]$.

In \cite{benaych2012singular} (see also \cite{paul2007asymptotics}), it is shown that when the low-rank matrix $\X$ is random (for instance, its columns may be drawn from a suitable low-rank, mean-zero distribution), then the limiting cosine $c$ of the angles between the top eigenvector of $\X \X^T$ and the top eigenvector of $\Y \Y^T$ is given by the formula:
\begin{align}
\label{cos_limit}
c^2 = 
\begin{cases}
\frac{1 - \sigma^4 \gamma/\lambda^2}{1 + \sigma^2 \gamma/ \lambda}
& \text{ if } \lambda > \sigma^2 \sqrt{\gamma}, \\
0 
& \text{ otherwise,}
\end{cases}
\end{align}
where $\lambda$ is the top eigenvalue of $\X\X^T / N$.

The key phenomenon is the phase transition at 
\begin{equation} \label{eq:phase_transition}
    \lambda_{critical} = \sigma^2 \sqrt{\gamma}.
\end{equation}
It is only when $\lambda$ is greater than this critical value that we are guaranteed a non-trivial correlation between the top eigenvector of the observed matrix $\Y\Y^T/N$ and the top eigenvector of $\X\X^T / N$.

We can view the observation model in the one-dimensional MRA model~\eqref{eq:mra} as a special instance of the spiked model, by taking the $i$th column of $\X$ to be $X_i = R_{S_i} x$. As $N\to\infty$, we can write 
\begin{align} \label{eq:spike_model_second_moment}
\frac{1}{N} \X \X^T = C_x D_{{\rho}} C_x^T.
\end{align}
Consequently, under the assumption that the $\DFT$ of $x$ does not vanish, the rank of $\X$ is the size of the support of ${\rho}$. When the support size of $\rho$ is fixed at $r$, the MRA problem is an instance of the spiked model.

Let us assume that the $\vert (\F x)[k]\vert =1$ for all $k$. This can be done by estimating the power spectrum first and then normalizing all Fourier coefficients. In this case, $C_x$ is an orthogonal matrix. In other words, 
 $x \perp R_\ell x$ for every $\ell \ne 0$; consequently, the $R_\ell x$ are precisely the top $r$ eigenvectors of $\X \X^T / N$, with corresponding eigenvalues $\|x\|^2 \rho[\ell]$. Then,~\eqref{cos_limit} tells us exactly how well we expect the spectral algorithm to perform in recovering $x$; indeed, the theory predicts a non-zero angle between $x$ and the top eigenvector of $\Y \Y^T / N$ whenever:
\begin{align} \label{eq:sample_complexity_spike_model}
N \ge \frac{L \sigma^4}{\|x\|^4 (\max\rho)^2} = \frac{L }{ (\max\rho)^2}\frac{1}{\SNR^2}.
\end{align}
Below this threshold, the output will be essentially random. We see that if the distribution is well-localized, then $\max\rho = \Omega(1)$ (with respect to the growing value of $L$) and then the sample complexity grows like $\frac{L}{\SNR^2}$. On the other hand, if the distribution is almost uniform, then $\max\rho= \mathcal{O}(1/L)$ as $L \to \infty$, and thus the sample complexity will be proportional to $L^3/\SNR^2$. 

To illustrate the relationship between the spiked model and MRA, we ran the following experiment. We generated a signal $x\in\mathbb{R}^{400}$ with i.i.d.\ normal entries and normalized it so that $\|x\|_2=10$. For noise levels $\sigma$ between 0.1 and 10, we drew $N$ samples of $x$ with noise at level $\sigma$, where $N$ is chosen at 100 plus the critical threshold given by~\eqref{eq:sample_complexity_spike_model} for $\sigma = \lambda^{1/2}\gamma^{-1/4}= 5.5313$ according to~\eqref{eq:phase_transition}. For $\sigma$ large enough, $N$ will not be large enough for the spectral method to produce an estimate better than random. The distribution of translations $\rho$ was taken to be $\rho[i] \propto i^2 $, for $i=1,\dots,5$, and zero elsewhere. 
Each experiment was repeated 200 times. The plots in Figure~\ref{fig-compare-spiked} display the average values over these 200 runs.

For each draw, we compute the top eigenvalue of the clean data matrix~\eqref{eq:spike_model_second_moment}, denoted by $\lambda$, and the associated eigenvector, which is a translated copy of $x$. We also compute the top eigenvector of the data matrix $\Y \Y^T / N$. The angle between the two eigenvectors is predicted by~\eqref{cos_limit}. In Figure~\ref{fig:spike_cosine}, we plot the predicted cosine against the true cosine. Clearly, we never attain the predicted value of zero in finite samples, but we see a precipitous decline when the noise level $\sigma$ exceeds its threshold value (the vertical dashed line).

We also measure the relative mean squared error defined by equation~\eqref{eq:mse_def}, where $\widehat{X}$ is the top eigenvector multiplied by $\|x\|$. In Figure~\ref{fig:spike_mse}, we plot this error as a function of $\sigma$. For reference, we also plot the ordinary error predicted by the spiked model (as derived from the predicted cosine between the vectors), without minimizing over shifts. Of course, minimizing over shifts will decrease the error; however, we still see the same qualitative behavior predicted from the spiked model, namely an increase in error as $\sigma$ grows, until the critical threshold of $\sigma$ is reached, after which the error plateaus.

%
%
\begin{figure*}
\centering 
    \begin{subfigure}{0.43\textwidth} 
            \begin{center}
             \includegraphics[scale=0.5]{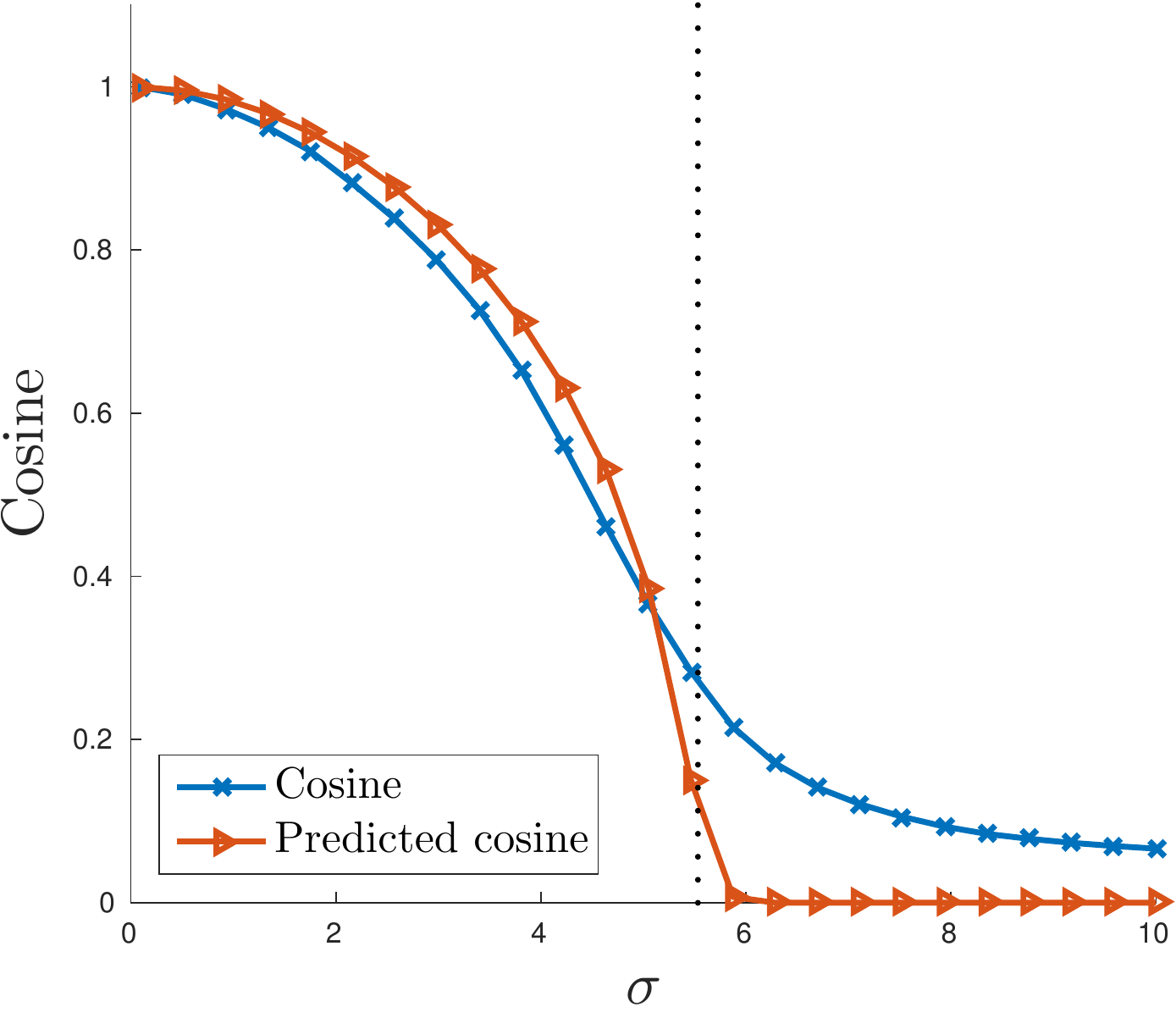}
\caption{ \label{fig:spike_cosine} Empirical cosines between the top eigenvectors of the matrices $\frac{1}{N} \X \X^T$ and $\frac{1}{N} \Y \Y^T$ as a function of the noise level compared to asymptotic cosines predicted by spiked model; see~\eqref{cos_limit}.}            
\end{center}
    \end{subfigure} \quad
        \begin{subfigure}{0.43\textwidth} 
        \begin{center}
            \includegraphics[scale=0.5]{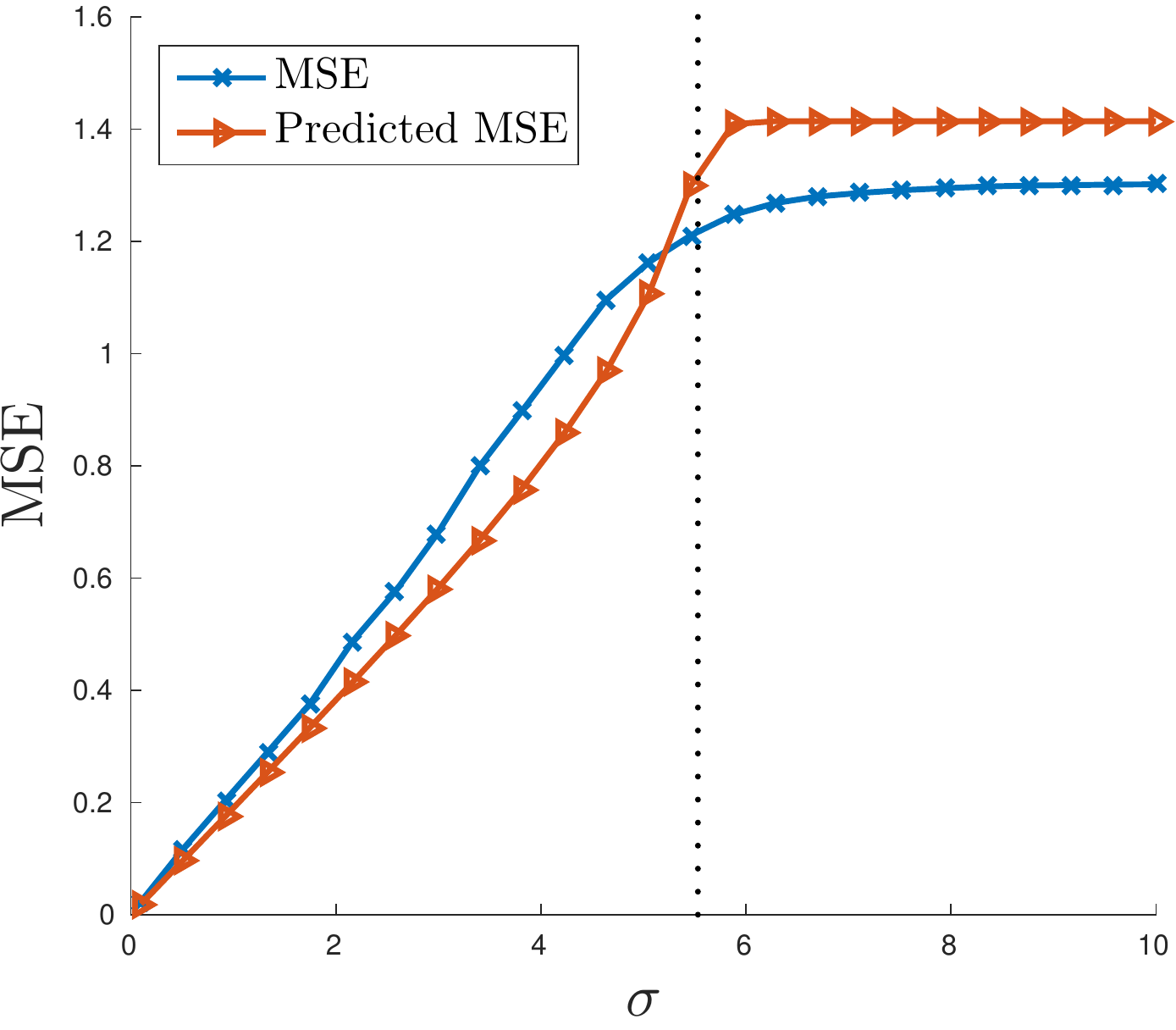}
            \caption{\label{fig:spike_mse} Empirical MSE compared to the asymptotic MSE predicted by spiked model. The MSE is defined in~\eqref{eq:mse_def}.} 
        \end{center}
    \end{subfigure} 
    \caption{    \label{fig-compare-spiked} Experiments related to the connection between the spike model and the MRA problem as discussed in Section~\ref{sec:spike_model}. The dashed line is the predicted threshold value of $\sigma = \lambda^{1/2}\gamma^{-1/4}= 5.5313$.}
\end{figure*}


\section{Additional Algorithms} \label{sec:algorithms}

While the spectral algorithm (Algorithm \ref{alg:SpectralAlg}) is asymptotically optimal as $\sigma$, $N \to \infty$ and for signals with non-vanishing $\DFT$, it may not perform well in small sample size or low $\DFT$ values. Therefore, in this section, we present two additional algorithms based on non-convex LS minimization and a modification of the EM algorithm presented in Section~\ref{sec:related_work} that takes the distribution into account. In Appendix~\ref{subsec:SDPmethod}, we also describe and analyze a convex relaxation approach based on semidefinite programming.

\subsection{Non-convex least-squares minimization} \label{subsec:LS_method}
 
The following method aims to find a signal in $\RL$ and a distribution in $\Delta^L$ that fit the observed data as well as possible in the LS sense. We formulate the problem
 as a smooth, non-convex, optimization problem with the constraint that the distribution lies on a simplex. Given estimators $\hatM^1$ and $\hatM^2$ of the first two moments $M^1$ and $M^2$, the problem reads
\begin{equation} \label{eq:LS}
    \min_{\tilde x\in\RL,\tilde \rho\in\Delta^L} 
        \| \hatM^2 - C_{\tilde x}D_{\tilde\rho} C_{\tilde x}^T\|_{\textrm{F}}^2 
        + \lambda\| \hatM^1 - C_{\tilde x}\tilde \rho\|_2^2,
\end{equation}
where $\lambda > 0$ is a predefined parameter. It can be verified that, by omitting signal-dependent terms, the variance of the elements of the first moment estimator is proportional to $\sigma^2$. It can be also shown that the variance of the elements of the second moment is proportional to $3L\sigma^4$ and $L\sigma^2$ in the low and the high $\SNR$ regimes, respectively (again, by omitting signal-dependent terms)~\cite{boumal2018heterogeneous}. Therefore, we set $\lambda =\frac{1}{L(1+3\sigma^2)}$ in our implementation.

 \subsection{An expectation-maximization algorithm for estimating $x$ and $\rho$ simultaneously} \label{sec:NU_EM}

In Section~\ref{sec:related_work}, we reviewed the EM algorithm for MRA from \cite{bendory2017bispectrum}, which is invariant to the distribution of translations. In this section, we modify the algorithm to take the distribution into account. A similar approach was introduced for the application of cryo-EM in~\cite{dvornek2015subspaceem}.
 
If we denote $s:=\{s_j\}_{1\le j\le N}$, the forward model of the MRA model~\eqref{eq:mra} reads:
\begin{align*}
    f_{x,\rho}(y,s) &= f_{x,\rho}(y|s) \prod_{j=1}^{N} \rho[s_j]\\
    &= \prod_{j=1}^{N} \rho[s_j]\frac{1}{(2 \pi \sigma^2)^{L/2}} 
    e^{- \frac{1}{2\sigma^2}\left\|R_{s_j}x - y_j\right\|^2 }
    .
\end{align*}
The log-likelihood function is then given, up to a constant, by
\begin{align*}
    \log \LL(y, s | x,\rho) 
    = \sum_{j=1}^N \left\{ \log \rho[s_j] 
                      - \frac{1}{2\sigma^2}\left\|R_{s_j}x - y_j\right\|^2 \right\} .
\end{align*}

The goal of the EM algorithm is to compute the maximum in $(x,\rho)$ of the marginal likelihood $\LL(y | x,\rho) = \sum_{s} \LL(y,s | x,\rho)$. The algorithm proceeds as follows. Start with some initial guesses $x_0$ and $\rho_0$ for the signal and distribution. Given $x_k$ and $\rho_k$, the next guess is given as follows:
\begin{align*}
(x_{k+1},\rho_{k+1}) = \operatorname*{arg\max}_{x,\rho\ } Q(x,\rho | x_k,\rho_{k}),
\end{align*}
where
\begin{align}
Q(x,\rho | x_k,\rho_k) 
&:= \E\left[\log \LL\left(y, S^{k} | x,\rho\right)\right].    \label{eq:Qdef} 
%
%
\end{align}
Here the distribution $S^{k}$ depends on $x_k$ and $\rho_{k}$ through
\begin{align*}
w_{k}^{\ell,j} := \mathbb{P}[S^{k}_j = \ell] 
=C_k^{j} e^{-\frac{1}{2\sigma^2}\|R_\ell x_k - y_j\|^2}  \rho_k[\ell], 
\end{align*}
where $C_k^{j}$ is a normalization term so that $\sum_{\ell} w_k^{\ell,j} = 1$. We can explicitly write \eqref{eq:Qdef} (omitting a constant term) as
\begin{align*}
Q(x,\rho | x_k,\rho_k) \hspace{-30pt} &\\
&= \sum_{j=1}^N \E\left[ 
\log \rho[S^k_j] - \frac{1}{2\sigma^2}\|R_{S^k_j}x - y_j\|^2\right] 
\\ 
&=\sum_{j=1}^N \sum_{\ell=0}^{L-1} w_{k}^{\ell,j} \left\{
\log \rho[\ell] - \frac{1}{2\sigma^2}\|R_{\ell}x - y_j\|^2  \right\},
\end{align*}

Maximizing $Q$ over $x$ and $\rho$ is simple, since the first term depends only on $\rho$ and the second term depends only on $x$. Specifically, it is easy to see that the maximum over $x$ is given by a weighted average of the translated observations:
\begin{align} \label{eq:EM_signal_update}
x_{k+1} = \frac{1}{N} \sum_{j=1}^N \sum_{\ell=0}^{L-1}  w_{k}^{\ell,j} R^{-1}_{\ell} y_j.
\end{align}
This step is almost identical (up to the values of the weights) to the standard EM update step~\eqref{eq:standard_EM}.

The maximimizing value of $\rho$ also has a closed formula. First, observe that we can write:
\begin{align*}
\rho_{k+1} = \operatorname*{arg\max}_{\rho \in \Delta^L} 
\sum_{\ell=0}^{L-1} W_k[\ell] \log (\rho[\ell]), 
\end{align*}
where $W_k[\ell] = \sum_{j=1}^N w_{k}^{\ell,j}$. To maximize a positive weighted combination of logarithms over the simplex, we use the following lemma:

\begin{lem} \label{lem:lemma_EM}
    If $w[\ell] > 0$ are positive weights, then the maximizer of $\sum_{\ell} w[\ell] \log(q[\ell])$ over all $q \in \Delta^L$ is 
    \begin{equation*}
q^*[\ell] = w[\ell] / \sum_{\ell^\prime} w[\ell^\prime].    
    \end{equation*}
    \end{lem}
\begin{proof}
    See Appendix~\ref{sec:proof_lemma_EM}.
\end{proof}

From this lemma, the maximizing $\rho$ is given by the formula:
\begin{align} \label{eq:EM_dist_update}
\rho_{k+1}[\ell]= \frac{W_k[\ell]} {\sum_{\ell^\prime=0}^{L-1} W_k[\ell^\prime]}.
\end{align}

To conclude, the modified EM updates the signal and the distribution estimations by~\eqref{eq:EM_signal_update} and~\eqref{eq:EM_dist_update}, respectively. However, compared to the methods which are based on moments estimation like Algorithm~\ref{alg:SpectralAlg} or the LS, it passes through the data at each iteration. Therefore, for large sample size, its 
computational cost may be substantially heavier.  

\section{Numerical experiments} \label{sec:numerics}

In this section, we present numerical results for the algorithms described in Section~\ref{sec:algorithms} and Algorithm~\ref{alg:SpectralAlg}. To measure the accuracy of an estimator $\widehat{X}$, we define the recovery relative error as 
\begin{equation}\label{eq:relative_error}
\textrm{relative error} = \min_{s\in\mathbb{Z}_L}\frac{\|R_s\widehat X - x\|_2}{\|x\|_2}.
\end{equation}
%
 The code of this section, including Matlab implementations and examples, is publicly available online~\footnote{\url{https://github.com/nirsharon/aperiodicMRA}}.

\subsection{Influence of the number of samples}

In the first example, we use a Haar-like signal of length $L=20$, depicted in Figure~\ref{fig:clean_signal}. Next, we generate its noisy, translated copies according to the MRA model~\eqref{eq:mra}, with noise variance of $\sigma = .25$. One example of a data sample corrupted with such noise is illustrated in Figure~\ref{fig:corrupted_sig_1}.

We use the EM algorithm of Section~\ref{sec:NU_EM} to estimate the signal. This process is repeated three times for different number of samples, $N=10^3$, $N=10^5$, and $N=10^7$. The estimates are presented in Figure~\ref{fig:estimation_N_1000}--\ref{fig:estimation_N_10000000}. As expected, the quality of the estimation improves significantly as $N$ grows.

%
%
\begin{figure*}
    \begin{center}
        \begin{subfigure}[t]{.19\textwidth} 
        \begin{center}
                    \includegraphics[width=\textwidth]{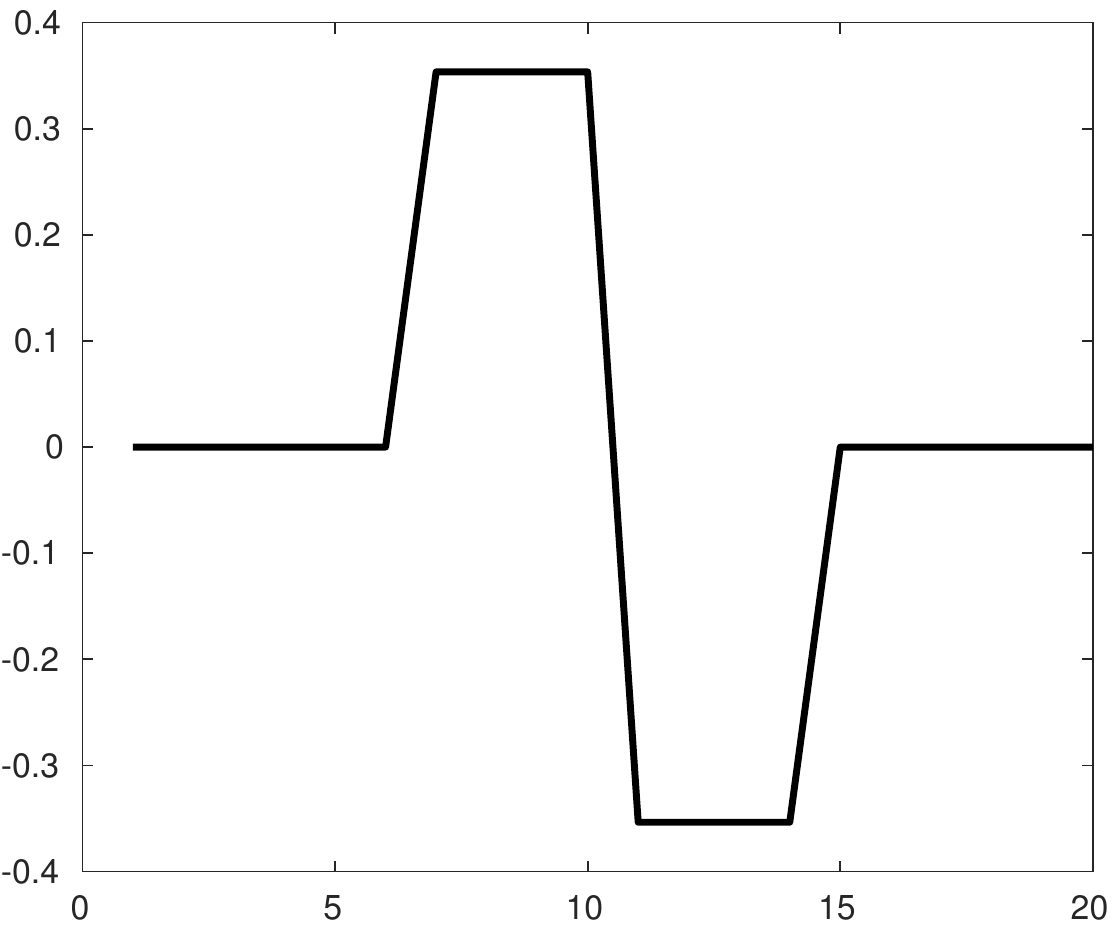}
                            \end{center}
            \caption{The clean signal}
            \label{fig:clean_signal}
        \end{subfigure}     \hfill
        \begin{subfigure}[t]{.19\textwidth}
        \begin{center}
                    \includegraphics[width=.98\textwidth]{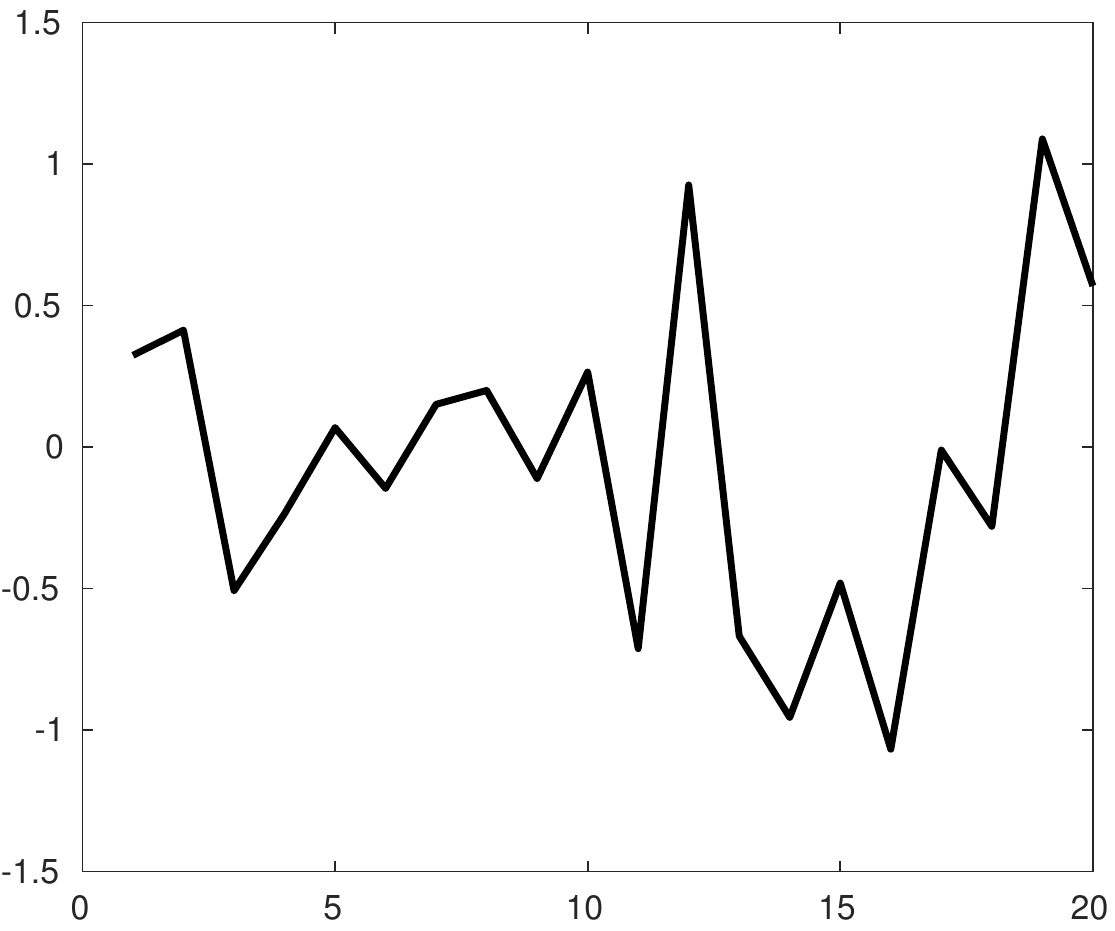}
        \end{center}
            \caption{An example of a data sample (one measurement)}
            \label{fig:corrupted_sig_1}
        \end{subfigure}   \hfill
        \begin{subfigure}[t]{.19\textwidth}
        \begin{center}
                    \includegraphics[width=\textwidth]{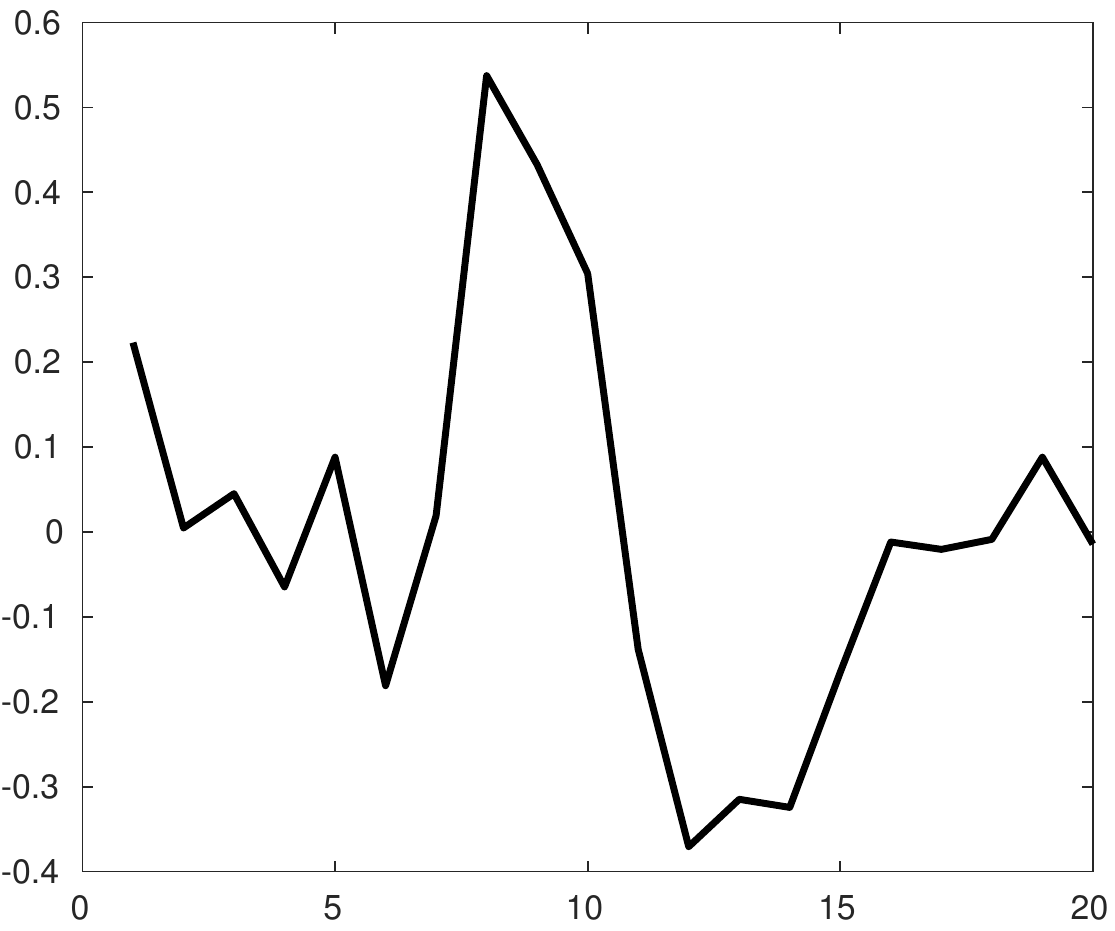}
        \end{center}
            \caption{An estimate with \\ $N=1,000$}
            \label{fig:estimation_N_1000}
        \end{subfigure}  \hfill
        \begin{subfigure}[t]{.19\textwidth}
        \begin{center}
            \includegraphics[width=\textwidth]{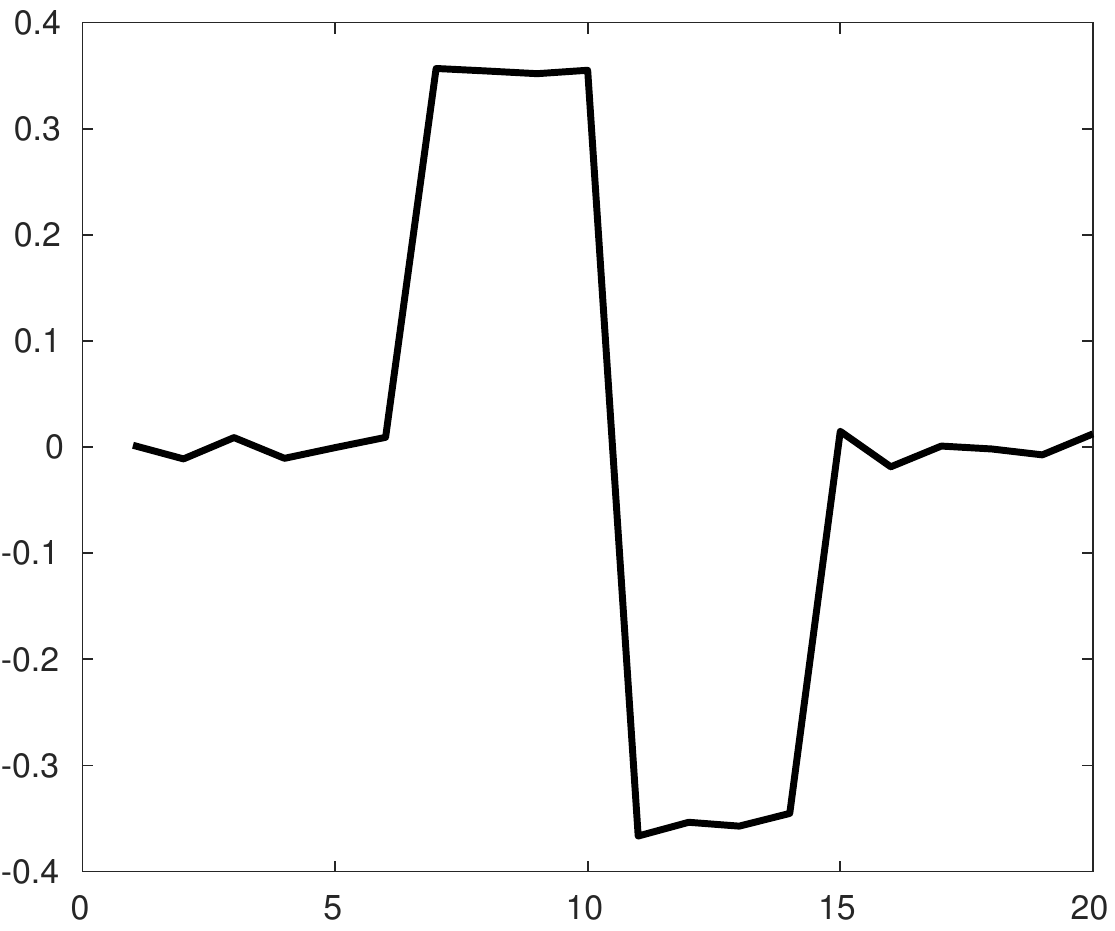}
        \end{center}
            \caption{An estimate with $N=100,000$}
            \label{fig:estimation_N_100000}
        \end{subfigure}   \hfill
        \begin{subfigure}[t]{.19\textwidth}
        \begin{center}
            \includegraphics[width=\textwidth]{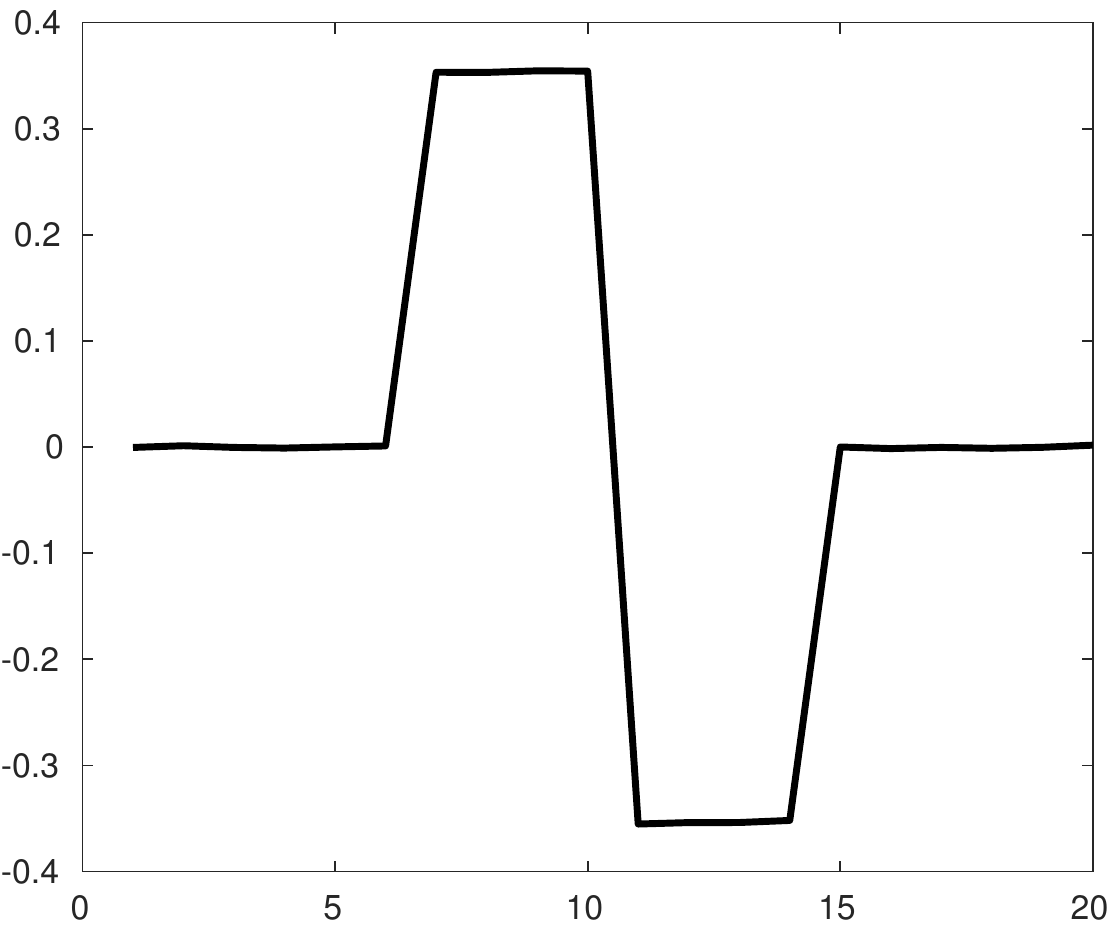}
        \end{center}
            \caption{An estimate with \\ $N=10,000,000$}
            \label{fig:estimation_N_10000000}
        \end{subfigure}
        \caption{An example of the estimation quality of a Haar-like signal with different number of samples ($N$), using the LS method. In these tests, $\sigma=1$.}
        \label{fig:visualEstimating}
    \end{center}
\end{figure*}

\subsection{Comparison of EM algorithms}

In~\cite{bendory2017bispectrum}, it is shown that in most cases, an EM method as described in Section~\ref{subsec:MRA_low_SNR}, achieves the smallest estimation error compared to the competitor algorithms.
The EM algorithm described in that paper is invariant to the distribution $\rho$. In particular, it treats the data as if it were drawn from the uniform distribution, which requires sample complexity that grows like $\omega(1/\SNR^3)$ rather than $\omega(1/\SNR^2)$. By contrast, the EM algorithm we propose in Section~\ref{sec:NU_EM} also estimates the distribution $\rho$ at each iteration. The updated estimation of the distribution is then used to update the signal's estimation. 


To demonstrate the importance of including the distribution into the model of the estimator, we consider a family of distributions 
\begin{equation} \label{eq:distribution}
\rho[t] \propto \exp(-t^2 / s^2)
\end{equation}
where the parameter $s > 0$ controls the concentration of $\rho$, or alternatively its uniformity: the larger $s$ is, the more uniform $\rho$ is. In general, we expect our algorithms to provide better estimations when $s$ is smaller, i.e., when $\rho$ is more concentrated; see Section~\ref{sec:spike_model}. 

We compared the standard EM with the EM algorithm described in Section~\ref{sec:NU_EM}. The experiments were conducted as follows. We fixed a random signal of length $L=25$ with i.i.d.\ normal entries and unit norm, and a series of distributions of the form~\eqref{eq:distribution} with the parameter $s$ varying between $3$ and $9$. Then, for each distribution we generated $N=2,000$ samples drawn with a fixed level of noise $\sigma = 1 $. We repeated the experiment independently $20$ times and averaged the errors. In Figure~\ref{fig:EM_comparison}, we plot the relative errors of the methods as a function of the uniformity parameter $s$. As expected, the standard EM is invariant to $s$. On the other hand, the adapted version of the EM exploits the varying distribution and performs better under more concentrated distributions. As the distribution becomes more uniform, the two methods exhibit similar error rates.

%
%
\begin{figure}[ht]
    \centerline{
        \includegraphics[width=.95\columnwidth]{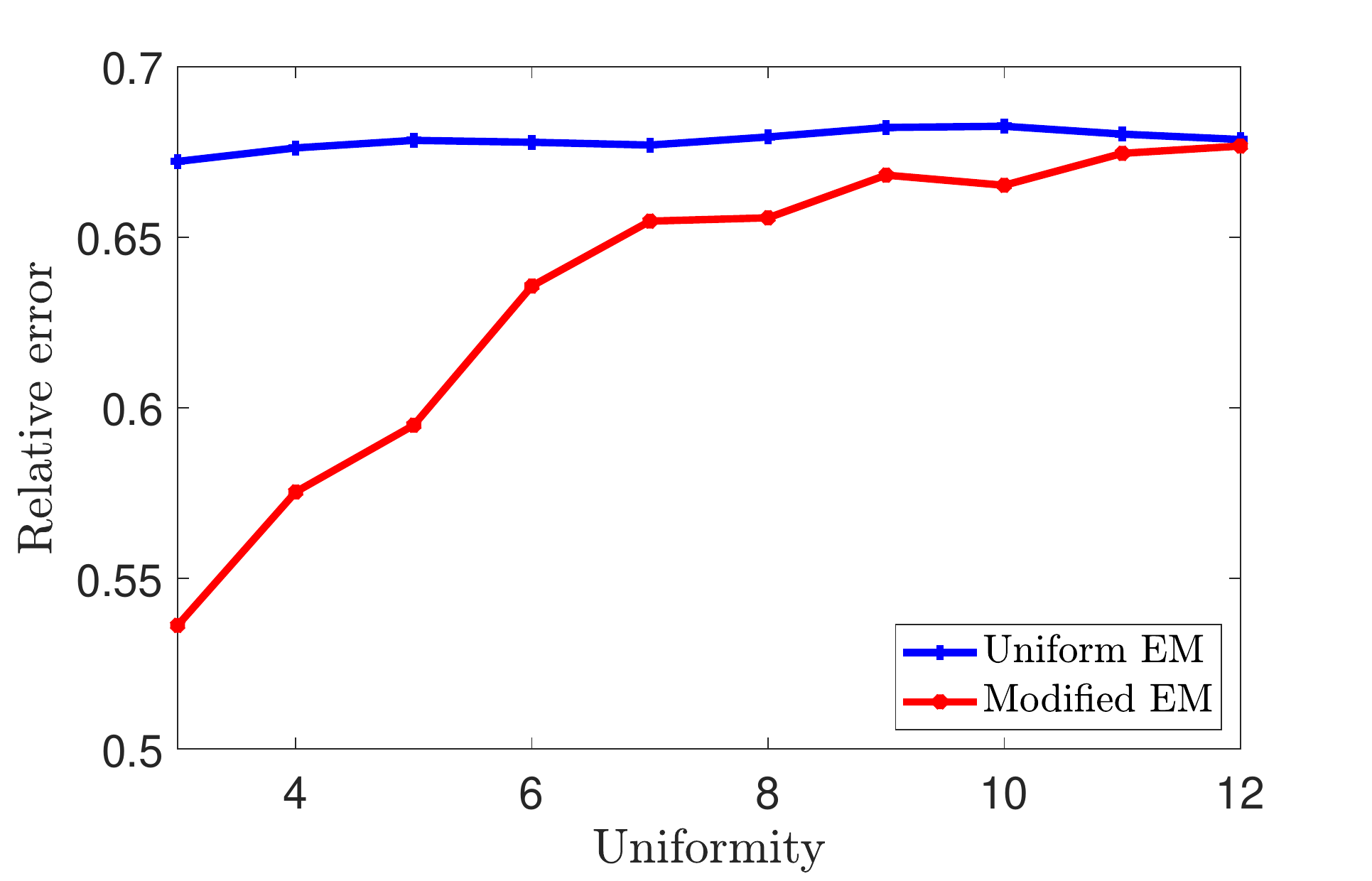}
    }
    \caption{EM comparison: the standard EM described in Section~\ref{sec:related_work} (uniform EM) versus the EM that includes distribution estimation (modified EM) described in Section~\ref{sec:NU_EM}. The algorithms were compared with different distributions of the form~\ref{eq:distribution} as a function of the parameter $s$. }
    \label{fig:EM_comparison}
\end{figure}

\subsection{Comparison of the different methods}

This paper presents three alternative techniques for solving the MRA: the spectral method described in Algorithm~\ref{alg:SpectralAlg}, the LS optimization of Section~\ref{subsec:LS_method}, and the EM of Section~\ref{sec:NU_EM}. In this comparison, we examined the estimation error of these three methods with different noise levels. We use a random signal of length $L=15$ with i.i.d. normal entries and rescaled to unit norm. The distribution $\rho$ is obtained by normalizing a vector with i.i.d. entries, distributed uniformly in $[0,1]$. We fix the number of samples to be $N=100,000$. Then, we sample the level of noise $\sigma$ at $20$ points ranging from $0.01$ to $10$. In Figure~\ref{fig:overall_comparison} we plot the average error for each of the sampled points over 40 different values of $x$ and $\rho$. As can be seen, the LS and EM methods are more robust to noise than the spectral method. In addition, the gap between these two methods becomes small as the SNR decreases.

%
%
\begin{figure}[ht]
    \centerline{
        \includegraphics[width=.95\columnwidth]{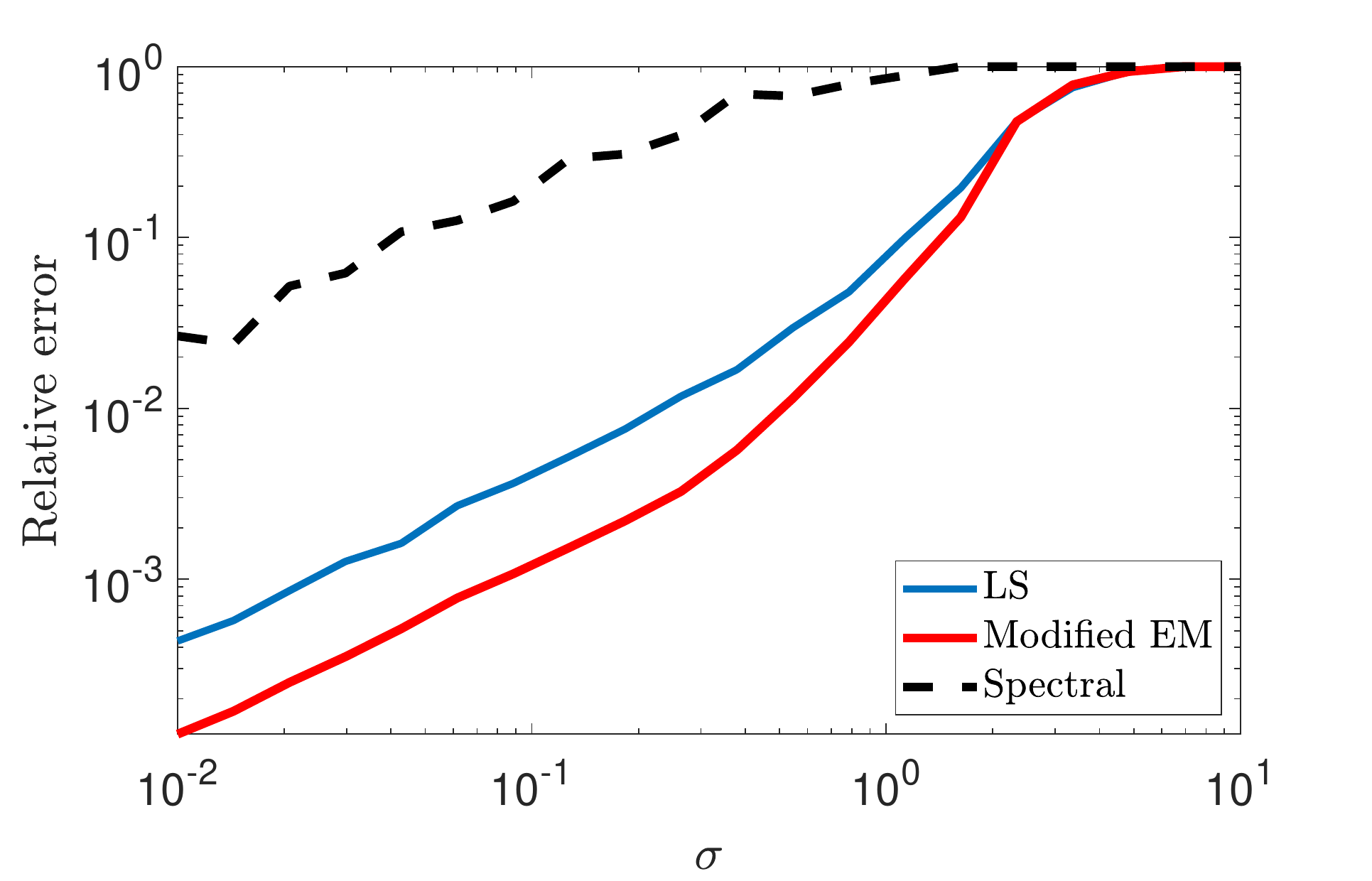}
    }
    \caption{A comparison of three methods: least squares (LS), expectation maximization (EM), and the spectral method, under varying level of noise.}
    \label{fig:overall_comparison}
\end{figure}

\subsection{Numerical error rates for the EM algorithm}

When the distribution $\rho$ is aperiodic, the optimal MSE for recovering $x$ in the low SNR, and large $N$ regime is of size $O(\sigma^4 / N)$. Since the relative error scales as $\sqrt \MSE$, by \eqref{eq:relative_error}, if the $\log$-error is viewed as a function of $\log(\sigma)$, the slope is expected to be no smaller than $2$ when $\sigma$ is large.

In Figure~\ref{fig-err} we plot the average log-error of the EM algorithm over $300$ trials as a function of $\log(\sigma)$. In each trial, we used the EM algorithm to estimate a randomly generated signal, with translations drawn from a randomly generated probability distribution. When $\sigma$ is large, the curve is indeed a line with slope close to $2$, which is the expected rate. However, when $\sigma$ is small, the curve is a line with slope close to $1$; namely, the error behaves approximately like $O(\sigma/\sqrt{N})$, rather than $O(\sigma^2/\sqrt{N})$. The moderate slope for high SNR suggests that in this regime the recovery problem is easier; for example, we know that alignment is possible in high SNR, as described in Section~\ref{sec:alignment}.

In Figure~\ref{fig-err2} we plot the average log-error (again over $300$ experiments) as a function of $\log(\sigma)$, but in this case each experiment used the uniform distribution of translations. In this regime, we know from~\cite{bandeira2017optimal} that the optimal slope is 3, not 2; and indeed, when $\sigma$ is large the curve has slope close to $3$. As in the other plot, when $\sigma$ is small the curve has slope close to $1$. Taken together, these two experiments suggest that the EM algorithm exhibits near-optimal behavior for both periodic and aperiodic distributions.

%
%
\begin{figure}[ht]
    \centerline{
        \includegraphics[scale = 0.48]{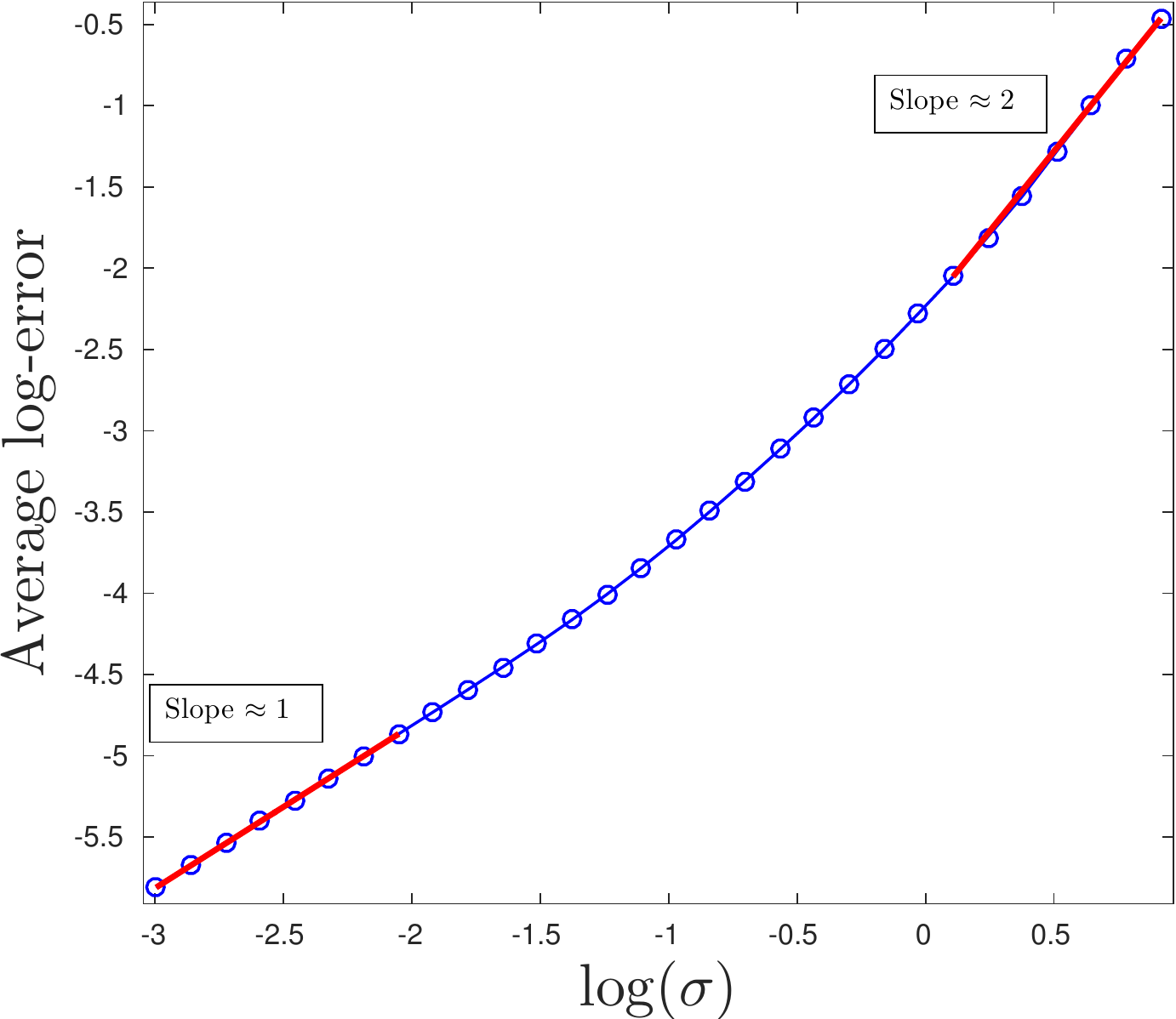}
    }
    \caption{Log-log plot of the error of the EM method versus $\sigma$, with random distributions.}
    \label{fig-err}
\end{figure}

%
%
\begin{figure}[ht]
    \centerline{
        \includegraphics[scale = 0.48]{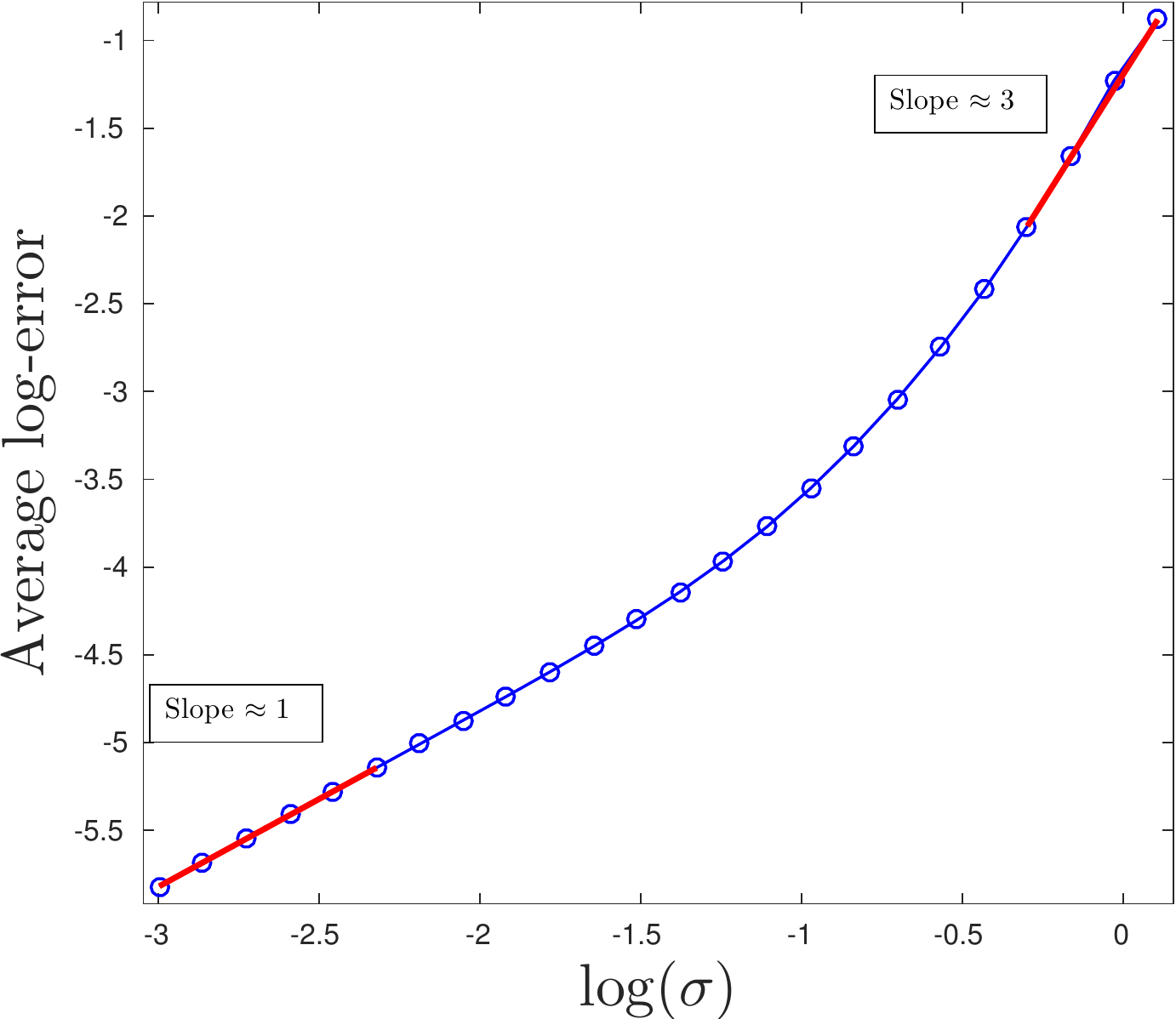}
    }
    \caption{Log-log plot of the error of the EM method versus $\sigma$, with uniform distribution}
    \label{fig-err2}
\end{figure}

\section{Discussion} \label{sec:conlusion}

In this paper, we have shown that the sample complexity for MRA with an aperiodic distribution of translations grows like $\omega(1/\SNR^2)$. This sample complexity can be achieved by a simple spectral algorithm. We also examined empirically the LS and EM algorithms.
Additionally, we extended previous works by showing that the sample complexity for any periodic distribution scales as $\omega(1/\SNR^3)$. 

We drew connections between the MRA problem and the spiked covariance model. This connection implies that the sample complexity is inversely proportional to the square of the maximal value of the distribution. Therefore, the more uniform the distribution is, the higher the sample complexity of the problem. 

One of the motivations for considering the MRA model arises from the imaging technique called single particle cryo--electron microscopy (cryo--EM), allowing to visualize molecules at near-atomic resolution~\cite{bartesaghi20152,sirohi20163}. In cryo--EM, noisy two-dimensional tomographic projections of the three-dimensional underlying molecule, taken at unknown viewing direction, are collected. The distribution of viewing directions in cryo--EM is typically non-uniform, as many molecules exhibit some preferred orientation \cite{frank1987tilt}.

The MRA model~\eqref{eq:mra} can be thought of as a simplified model for the cryo--EM problem, where cyclic translations replace actions of elements of the group $SO(3)$ \cite{singer2018mathematics}. The tomographic projection does not appear in~\eqref{eq:mra}. Our technique for MRA, based on the low-order moments of the data, is similar to the framework proposed by Zvi Kam in~\cite{kam1980reconstruction,levin20183d} for cryo--EM. In particular, Kam suggested a method to estimate a molecule directly from the statistics of the projections, rather than estimating the viewing directions. Our work is one step towards understanding the sample complexity of Kam's method in particular, and the cryo--EM problem in general.

\section*{Acknowledgments}
We would like to thank Afonso Bandeira, Nicolas Boumal, Joseph Kileel, Roy Lederman and Zhizhen Zhao for many insightful discussions.

\bibliographystyle{ieeetr}
\bibliography{refs_mra}

\begin{thebibliography}{10}

\bibitem{diamond1992multiple}
R.~Diamond, ``On the multiple simultaneous superposition of molecular
  structures by rigid body transformations,'' {\em Protein Science}, vol.~1,
  no.~10, pp.~1279--1287, 1992.

\bibitem{theobald2012optimal}
D.~L. Theobald and P.~A. Steindel, ``Optimal simultaneous superpositioning of
  multiple structures with missing data,'' {\em Bioinformatics}, vol.~28,
  no.~15, pp.~1972--1979, 2012.

\bibitem{park2011stochastic}
W.~Park, C.~R. Midgett, D.~R. Madden, and G.~S. Chirikjian, ``A stochastic
  kinematic model of class averaging in single-particle electron microscopy,''
  {\em The International journal of robotics research}, vol.~30, no.~6,
  pp.~730--754, 2011.

\bibitem{park2014assembly}
W.~Park and G.~S. Chirikjian, ``An assembly automation approach to alignment of
  noncircular projections in electron microscopy,'' {\em IEEE Transactions on
  Automation Science and Engineering}, vol.~11, no.~3, pp.~668--679, 2014.

\bibitem{scheres2005maximum}
S.~H. Scheres, M.~Valle, R.~Nu{\~n}ez, C.~O. Sorzano, R.~Marabini, G.~T.
  Herman, and J.-M. Carazo, ``Maximum-likelihood multi-reference refinement for
  electron microscopy images,'' {\em Journal of molecular biology}, vol.~348,
  no.~1, pp.~139--149, 2005.

\bibitem{zwart2003fast}
J.~P. Zwart, R.~van~der Heiden, S.~Gelsema, and F.~Groen, ``Fast translation
  invariant classification of {HRR} range profiles in a zero phase
  representation,'' {\em IEE Proceedings-Radar, Sonar and Navigation},
  vol.~150, no.~6, pp.~411--418, 2003.

\bibitem{gil2005using}
R.~Gil-Pita, M.~Rosa-Zurera, P.~Jarabo-Amores, and F.~L{\'o}pez-Ferreras,
  ``Using multilayer perceptrons to align high range resolution radar
  signals,'' in {\em International Conference on Artificial Neural Networks},
  pp.~911--916, Springer, 2005.

\bibitem{rosen2016certifiably}
D.~M. Rosen, L.~Carlone, A.~S. Bandeira, and J.~J. Leonard, ``A certifiably
  correct algorithm for synchronization over the special euclidean group,''
  {\em arXiv preprint arXiv:1611.00128}, 2016.

\bibitem{dryden1998statistical}
I.~L. Dryden and K.~V. Mardia, {\em Statistical shape analysis}, vol.~4.
\newblock J. Wiley Chichester, 1998.

\bibitem{foroosh2002extension}
H.~Foroosh, J.~B. Zerubia, and M.~Berthod, ``Extension of phase correlation to
  subpixel registration,'' {\em IEEE transactions on image processing},
  vol.~11, no.~3, pp.~188--200, 2002.

\bibitem{robinson2009optimal}
D.~Robinson, S.~Farsiu, and P.~Milanfar, ``Optimal registration of aliased
  images using variable projection with applications to super-resolution,''
  {\em The Computer Journal}, vol.~52, no.~1, pp.~31--42, 2009.

\bibitem{bartesaghi20152}
A.~Bartesaghi, A.~Merk, S.~Banerjee, D.~Matthies, X.~Wu, J.~L. Milne, and
  S.~Subramaniam, ``2.2 {{\AA}} resolution cryo-{EM} structure of
  $\beta$-galactosidase in complex with a cell-permeant inhibitor,'' {\em
  Science}, vol.~348, no.~6239, pp.~1147--1151, 2015.

\bibitem{sirohi20163}
D.~Sirohi, Z.~Chen, L.~Sun, T.~Klose, T.~C. Pierson, M.~G. Rossmann, and R.~J.
  Kuhn, ``The 3.8 {{\AA}} resolution cryo-{EM} structure of {Z}ika virus,''
  {\em Science}, vol.~352, no.~6284, pp.~467--470, 2016.

\bibitem{singer2018mathematics}
A.~Singer, ``Mathematics for cryo-electron microscopy,'' {\em to appear in the
  Proceedings of the International Congress of Mathematicians 2018}, 2018.

\bibitem{aguerrebere2016fundamental}
C.~Aguerrebere, M.~Delbracio, A.~Bartesaghi, and G.~Sapiro, ``Fundamental
  limits in multi-image alignment,'' {\em IEEE Transactions on Signal
  Processing}, vol.~64, no.~21, pp.~5707--5722, 2016.

\bibitem{bendory2018toward}
T.~Bendory, N.~Boumal, W.~Leeb, E.~Levin, and A.~Singer, ``Toward single
  particle reconstruction without particle picking: Breaking the detection
  limit,'' {\em arXiv preprint arXiv:1810.00226}, 2018.

\bibitem{bendory2017bispectrum}
T.~Bendory, N.~Boumal, C.~Ma, Z.~Zhao, and A.~Singer, ``Bispectrum inversion
  with application to multireference alignment,'' {\em IEEE Transactions on
  Signal Processing}, vol.~66, pp.~1037--1050, Feb 2018.

\bibitem{bandeira2017optimal}
A.~Bandeira, P.~Rigollet, and J.~Weed, ``Optimal rates of estimation for
  multi-reference alignment,'' {\em arXiv preprint arXiv:1702.08546}, 2017.

\bibitem{johnstone2001distribution}
I.~M. Johnstone, ``On the distribution of the largest eigenvalue in principal
  components analysis,'' {\em Annals of Statistics}, vol.~29, no.~2,
  pp.~295--327, 2001.

\bibitem{paul2007asymptotics}
D.~Paul, ``Asymptotics of sample eigenstructure for a large dimensional spiked
  covariance model,'' {\em Statistica Sinica}, vol.~17, no.~4, pp.~1617--1642,
  2007.

\bibitem{gavish-donoho-2017}
M.~Gavish and D.~L. Donoho, ``Optimal shrinkage of singular values,'' {\em IEEE
  Transactions on Information Theory}, vol.~63, no.~4, pp.~2137--2152, 2017.

\bibitem{benaych2012singular}
F.~Benaych-Georges and R.~R. Nadakuditi, ``The singular values and vectors of
  low rank perturbations of large rectangular random matrices,'' {\em Journal
  of Multivariate Analysis}, vol.~111, pp.~120--135, 2012.

\bibitem{dobriban2017optimal}
E.~Dobriban, W.~Leeb, and A.~Singer, ``Optimal prediction in the linearly
  transformed spiked model,'' {\em arXiv preprint arXiv:1709.03393}, 2017.

\bibitem{james1961estimation}
W.~James and C.~Stein, ``Estimation with quadratic loss,'' in {\em Proceedings
  of the Fourth Berkeley Symposium on Mathematical Statistics and Probability},
  pp.~361--379, 1961.

\bibitem{efron1973stein}
B.~Efron and C.~Morris, ``Stein's estimation rule and its competitors-an
  empirical {B}ayes approach,'' {\em Journal of the American Statistical
  Association}, vol.~68, no.~341, pp.~117--130, 1973.

\bibitem{efron1975data}
B.~Efron and C.~Morris, ``Data analysis using {S}tein's estimator and its
  generalizations,'' {\em Journal of the American Statistical Association},
  vol.~70, no.~350, pp.~311--319, 1975.

\bibitem{singer2011angular}
A.~Singer, ``Angular synchronization by eigenvectors and semidefinite
  programming,'' {\em Applied and computational harmonic analysis}, vol.~30,
  no.~1, pp.~20--36, 2011.

\bibitem{boumal2016nonconvex}
N.~Boumal, ``Nonconvex phase synchronization,'' {\em SIAM Journal on
  Optimization}, vol.~26, no.~4, pp.~2355--2377, 2016.

\bibitem{perry2018message}
A.~Perry, A.~S. Wein, A.~S. Bandeira, and A.~Moitra, ``Message-passing
  algorithms for synchronization problems over compact groups,'' {\em
  Communications on Pure and Applied Mathematics}, vol.~71, no.~11,
  pp.~2275--2322, 2018.

\bibitem{chen2018projected}
Y.~Chen and E.~J. Cand{\`e}s, ``The projected power method: An efficient
  algorithm for joint alignment from pairwise differences,'' {\em
  Communications on Pure and Applied Mathematics}, vol.~71, no.~8,
  pp.~1648--1714, 2018.

\bibitem{bandeira2014tightness}
A.~S. Bandeira, N.~Boumal, and A.~Singer, ``Tightness of the maximum likelihood
  semidefinite relaxation for angular synchronization,'' {\em Mathematical
  Programming}, vol.~163, no.~1, pp.~145--167, 2017.

\bibitem{zhong2018near}
Y.~Zhong and N.~Boumal, ``Near-optimal bounds for phase synchronization,'' {\em
  SIAM Journal on Optimization}, vol.~28, no.~2, pp.~989--1016, 2018.

\bibitem{bandeira2015non}
A.~S. Bandeira, Y.~Chen, and A.~Singer, ``Non-unique games over compact groups
  and orientation estimation in cryo-{EM},'' {\em arXiv preprint
  arXiv:1505.03840}, 2015.

\bibitem{bandeira2014multireference}
A.~S. Bandeira, M.~Charikar, A.~Singer, and A.~Zhu, ``Multireference alignment
  using semidefinite programming,'' in {\em Proceedings of the 5th conference
  on Innovations in theoretical computer science}, pp.~459--470, ACM, 2014.

\bibitem{chen2014near}
Y.~Chen, L.~Guibas, and Q.~Huang, ``Near-optimal joint object matching via
  convex relaxation,'' in {\em Proceedings of the 31st International Conference
  on Machine Learning (ICML-14)}, pp.~100--108, 2014.

\bibitem{bandeira2016low}
A.~S. Bandeira, N.~Boumal, and V.~Voroninski, ``On the low-rank approach for
  semidefinite programs arising in synchronization and community detection,''
  in {\em Conference on Learning Theory}, pp.~361--382, 2016.

\bibitem{perry2017sample}
A.~Perry, J.~Weed, A.~Bandeira, P.~Rigollet, and A.~Singer, ``The sample
  complexity of multi-reference alignment,'' {\em arXiv preprint at
  arXiv:1707.00943}, 2017.

\bibitem{boumal2018heterogeneous}
N.~Boumal, T.~Bendory, R.~R. Lederman, and A.~Singer, ``Heterogeneous
  multireference alignment: A single pass approach,'' in {\em Information
  Sciences and Systems (CISS), 2018 52nd Annual Conference on}, pp.~1--6, IEEE,
  2018.

\bibitem{dempster1977maximum}
A.~P. Dempster, N.~M. Laird, and D.~B. Rubin, ``Maximum likelihood from
  incomplete data via the {EM} algorithm,'' {\em Journal of the royal
  statistical society. Series B (methodological)}, pp.~1--38, 1977.

\bibitem{ChapRobb}
D.~G. Chapman and H.~Robbins, ``Minimum variance estimation without regularity
  assumptions,'' {\em Ann. Math. Statist.}, vol.~22, pp.~581--586, 12 1951.

\bibitem{CramerRao}
H.~Cram{\'e}r, {\em Mathematical Methods of Statistics (PMS-9)}, vol.~9.
\newblock Princeton university press, 2016.

\bibitem{abbe2018estimation}
E.~Abbe, J.~M. Pereira, and A.~Singer, ``Estimation in the group action
  channel,'' in {\em 2018 IEEE International Symposium on Information Theory
  (ISIT)}, pp.~561--565, June 2018.

\bibitem{dvornek2015subspaceem}
N.~C. Dvornek, F.~J. Sigworth, and H.~D. Tagare, ``Subspace{EM}: A fast
  maximum-a-posteriori algorithm for cryo-{EM} single particle
  reconstruction,'' {\em Journal of structural biology}, vol.~190, no.~2,
  pp.~200--214, 2015.

\bibitem{frank1987tilt}
M.~Radermacher, T.~Wagenknecht, A.~Verschoor, and J.~Frank, ``Three-dimensional
  reconstruction from a single-exposure, random conical tilt series applied to
  the 50{S} ribosomal subunit of {E}scherichia coli,'' {\em Journal of
  Microscopy}, vol.~146, no.~2, pp.~113--136, 1987.

\bibitem{kam1980reconstruction}
Z.~Kam, ``The reconstruction of structure from electron micrographs of randomly
  oriented particles,'' {\em Journal of Theoretical Biology}, vol.~82, no.~1,
  pp.~15--39, 1980.

\bibitem{levin20183d}
E.~Levin, T.~Bendory, N.~Boumal, J.~Kileel, and A.~Singer, ``3{D} ab initio
  modeling in cryo-{EM} by autocorrelation analysis,'' in {\em Biomedical
  Imaging (ISBI 2018), 2018 IEEE 15th International Symposium on},
  pp.~1569--1573, IEEE, 2018.

\bibitem{yu2015useful}
Y.~Yu and R.~J. Samworth, ``A useful variant of the {D}avis-{K}ahan theorem for
  statisticians,'' {\em Biometrika}, vol.~102, no.~2, pp.~315--323, 2015.

\bibitem{vershynin2010intro}
R.~Vershynin, ``Introduction to the non-asymptotic analysis of random
  matrices,'' {\em arXiv preprint arXiv:1011.3027}, 2010.

\bibitem{goemans1995improved}
M.~X. Goemans and D.~P. Williamson, ``Improved approximation algorithms for
  maximum cut and satisfiability problems using semidefinite programming,''
  {\em Journal of the ACM (JACM)}, vol.~42, no.~6, pp.~1115--1145, 1995.

\bibitem{grant2008cvx}
M.~Grant, S.~Boyd, and Y.~Ye, ``{CVX}: {M}atlab software for disciplined convex
  programming,'' 2008.

\end{thebibliography}

\appendix

\numberwithin{equation}{subsection}
\numberwithin{thm}{subsection} 

\subsection{Proof of Lemma \ref{lem:nchi2}}
\label{proof:eqnchi2}

We have 
\begin{align*}
\chi^2(f^N_{\tilde x,\tilde \rho}||f^N_{x,\rho}) \hspace{-40pt}\\
&=\int_{\reals^{L\times N}} \left(\frac{f^N_{\tilde x,\tilde \rho}(z^N)}{f^N_{x,\rho}(z^N)}-1\right)^2 f^N_{x,\rho}(z^N)\, dz^N,\\
&=\int_{\reals^{L\times N}} \frac{f^N_{\tilde x,\tilde \rho}(z^N)^2}{f^N_{x,\rho}(z^N)}-2 f^N_{\tilde x,\tilde \rho}(z^N)+f^N_{x,\rho}(z^N)\,dz^N,\\
&=\int_{\reals^{L\times N}} \frac{f^N_{\tilde x,\tilde \rho}(z^N)^2}{f^N_{x,\rho}(z^N)}\,dz^N -1,\\
&=\left(\int_{\reals^{L}} \frac{f_{\tilde x,\tilde \rho}(z)^2}{f_{x,\rho}(z)}\,dz\right)^N -1,\\
&=(\chi^2(f_{\tilde x,\tilde \rho}||f_{x,\rho})+1)^N-1,\\
\end{align*}
where the third line follows from $f^N_{\tilde x,\tilde \rho}$ and $f^N_{x,\rho}$ being probability distributions, and the fourth line follows from \eqref{eq:parindepence}.

\subsection{Proof of Theorem \ref{thm:ChapmanRobbins}}
\label{proof:ChapmanRobbins}
The proof mimics the one of the classical Chapman and Robbins bound. Recalling equation~\eqref{eq:Cov} and the definition of positive semidefinite matrices, the statement is equivalent to 
\begin{multline}\label{eq:chapmanrobbins}
\E_{x,\rho}\left[\left(w^T(\phi_x(\widehat X)-\E_{x,\rho}[\phi_{x}(\widehat X)])\right)^2\right]\\\ge \frac{\left[w^T\left(\E_{\tilde x,\tilde \rho}[\phi_{x}(\widehat X)]-\E_{x,\rho}[\phi_{x}(\widehat X)]\right)\right]^2}{\chi^2(f^N_{\tilde x,\tilde \rho}||f^N_{x,\rho})},
\end{multline}
for all $w,\tilde x\in\RL$ and $\tilde \rho\in \Delta^L$. Define $$Z=\frac{f_{\tilde x,\tilde \rho}(Y)}{f_{x,\rho}(Y)}.$$
and note that
\begin{itemize}
	\item $\E_{x,\rho}[g(Y)Z]=\E_{\tilde x,\tilde \rho}[g(Y)],$ 
	\item $\E_{x,\rho}[Z-1]=0,$
	\item $\E_{x,\rho}[(Z-1)^2]=\chi^2(f^N_{\tilde x,\tilde \rho}||f^N_{x,\rho}).$
\end{itemize}
We have
\begin{align*}
w^T\left(\E_{\tilde x,\tilde \rho}[\phi_{x}(\widehat X)]-\E_{x,\rho}[\phi_{x}(\widehat X)]\right)\hspace{-98pt}\\
&=
\E_{\tilde x,\tilde \rho}[w^T \phi_x(\widehat X)]-\E_{x,\rho}[w^T \phi_x(\widehat X)]\\
&=\E_{x,\rho}[w^T \phi_x(\widehat X)(Z-1)]\\
&=\E_{x,\rho}\left[w^T\left(\phi_x(\widehat X)-\E_{x,\rho}[\phi_{x}(\widehat X)]\right)(Z-1)\right],
\end{align*}
and by Cauchy-Schwarz
\begin{multline*}
\left[w^T\left(\E_{\tilde x,\tilde \rho}[\phi_{x}(\widehat X)]-\E_{x,\rho}[\phi_{x}(\widehat X)]\right)\right]^2\\
\le \E_{x,\rho}[(w^T(\phi_x(\widehat X)-\E_{x,\rho}[\phi_{x}(\widehat X)]))^2]\chi^2(f^N_{\tilde x,\tilde \rho}||f^N_{x,\rho}).
\end{multline*}

\subsection{Proof of Lemma \ref{lem:falldominoes}}
\label{proof:falldominoes}
Equation \eqref{eq:falldominoes} follows from some algebraic manipulations:
\begin{align*}
\chi^2(f_{\tilde x,\tilde \rho}||f_{x,\rho})\hspace{-50pt} &\\
&=\int_{\reals^L} \left(\frac{f_{\tilde x,\tilde \rho}(y;\gamma)}{f_{x,\rho}(y;\gamma)}-1\right)^2 f_{x,\rho}(y;\gamma)\,dy\\
&=\int_{\reals^L} \frac{\left(\sum\limits_{i=0}^\infty (\alpha^i_{\tilde x,\tilde \rho}(y)-\alpha^i_{x,\rho}(y))\frac{\gamma^i}{i!}\right)^2}
{\sum\limits_{i=0}^\infty \alpha^i_{x,\rho}(y)\frac{\gamma^i}{i!}} f_G(y)\,dy\\
&=\int_{\reals^L} \frac{\left(\sum\limits_{i=d}^\infty (\alpha^i_{\tilde x,\tilde \rho}(y)-\alpha^i_{x,\rho}(y))\frac{\gamma^{i}}{i!}\right)^2}{1+\sum\limits_{i=1}^\infty \alpha^i_{x,\rho}(y)\frac{\gamma^i}{i!}} f_G(y)\,dy\\
&=\frac{\gamma^{2d}}{(d!)^2}\int_{\reals^L} \left(\alpha^d_{\tilde x,\tilde \rho}(y)-\alpha^d_{x,\rho}(y)\right)^2 f_G(y)\,dy+\OO(\gamma^{2d+1})\\
&=\frac{\gamma^{2d}}{(d!)^2} \E\left[\left(\alpha_{\tilde x,\tilde \rho}^d(G)-\alpha_{x,\rho}^d(G)\right)^2\right]+\OO(\gamma^{2d+1}),
\end{align*}
where the third equation follows from the definition of $d$, i.e. $\alpha_{\tilde x,\tilde \rho}^n(z)=\alpha_{x,\rho}^n(z)$ almost surely for all $n< d$. Equation~\eqref{eq:falldominoes} now follows from $\gamma=1/\sigma$.

We now prove \eqref{eq:chi2autocorr}. It is enough to show that 
$$\E\left[\alpha_{\tilde x,\tilde \rho}^d(\error)\alpha_{x,\rho}^d(\error)\right]=d!\left<M^d_{\tilde x,\tilde \rho},M^d_{x,\rho}\right>,$$
Let $S$ and $\tilde S$ be two independent random variables such that $S\sim \rho$ and $\tilde S\sim \tilde \rho$. We have
\begin{align}
\nonumber \left<M^d_{\tilde x,\tilde \rho},M^d_{x,\rho}\right>
&=\left<\E[(R_{\tilde S} \tilde x)^{\otimes d}],\E[(R_S x)^{\otimes d}]\right>\\
\nonumber &=\E\left[\left<(R_{\tilde S} \tilde x)^{\otimes d},(R_S x)^{\otimes d}\right>\right]\\
&=\E\left[\left<R_{\tilde S} \tilde x,R_S x\right>^d\right]. \label{eq:StildeS}
\end{align}
On the other hand, we can write $f_{x,\rho}$ explicitly by
\begin{align*}
f_{x,\rho}(y;\gamma)
&=\frac1{\sqrt{2\pi}^L}\sum_{\ell=0}^{L-1}
\rho[\ell] \exp\left(-\frac{\|y- \gamma R_\ell x\|^2}{2}\right)\\
&=\E[f_\error(y- \gamma R_S x)],
\end{align*}
where $S\sim \rho$, thus by equation (\ref{eq:alpha})
\begin{align*}
&\hspace{-10pt}\E\left[\alpha_{\tilde x,\tilde \rho}^d(\error)\alpha_{x,\rho}^d(\error)\right]\\
&=\E\left[\frac{\partial^d}{\partial \tilde \gamma^d}\left(\frac{f_{\tilde x,\tilde \rho}(\error;\tilde \gamma)}{f_\error(\error)}\right|_{\tilde \gamma=0}\frac{\partial^d}{\partial \gamma^d}\left(\frac{f_{x,\rho}(\error;\gamma)}{f_\error(\error)}\right|_{\gamma=0}\right]\\
&=\frac{\partial^{2d}}{\partial \tilde \gamma^d\partial \gamma^d}
\E\left[\frac{f_{\tilde x,\tilde \rho}(\error;\tilde \gamma)}{f_\error(\error)}\frac{f_{x,\rho}(\error;\gamma)}{f_\error(\error)}\right]_{\tilde \gamma,\gamma=0}\\
&=\frac{\partial^{2d}}{\partial \tilde \gamma^d\partial \gamma^d}
\E\left[\frac{f_\error(\error- \tilde \gamma R_{\tilde S} \tilde x)}{f_\error(\error)}\frac{f_\error(\error- \gamma R_S x)}{f_\error(\error)}\right]_{\tilde \gamma,\gamma=0},
\end{align*}
where $S$ and $\tilde S$ are defined as in \eqref{eq:StildeS}. We have
\begin{align*}
&\hspace{-2pt}\E\left[\frac{f_\error(\error- \tilde \gamma R_{\tilde S} \tilde x)}{f_\error(\error)}\frac{f_\error(\error- \gamma R_S x)}{f_\error(\error)}\bigg|\tilde S,S\right]\\
&\text{\footnotesize $\displaystyle= \frac1{\sqrt{2\pi}^L}\int_{\RL} \exp\left(-\frac{\|z- \tilde \gamma R_{\tilde S} \tilde x\|^2+\|z- \gamma R_{S} x\|^2-\|z\|^2}{2}\right) \,dz$}\\
&\text{\footnotesize $\displaystyle=\frac1{\sqrt{2\pi}^L}  \int_{\RL} \exp\left(-\frac{\|z- \tilde \gamma R_{\tilde S} \tilde x- \gamma R_{S} x\|^2}{2}+\gamma\tilde \gamma \left<R_{\tilde S} \tilde x,R_S x\right> \right) \,dz$}\\
&\text{\footnotesize $=$} \exp\left(\gamma\tilde \gamma \left<R_{\tilde S} \tilde x,R_S x\right> \right).
\end{align*}
The proof of \eqref{eq:chi2autocorr} finally follows from equation \eqref{eq:StildeS} and
\begin{align*}
\E\left[\alpha_{\tilde x,\tilde \rho}^d(\error)\alpha_{x,\rho}^d(\error)\right]\hspace{-30pt}&\\
&=\E\left[\frac{\partial^{2d}}{\partial \tilde \gamma^d\partial \gamma^d}\exp\left(\gamma\tilde \gamma \left<R_{\tilde S} \tilde x,R_S x\right> \right)\right]_{\tilde \gamma,\gamma=0}\\
&=d!\,\E\left[\left<R_{\tilde S} \tilde x,R_S x\right>^d\right].
\end{align*}

\subsection{Analog results for derivatives}
\label{sec:derivatives}

This section provides analog results to the ones presented in section \ref{sec:information_limit}, but involving the limit $(\tilde x,\tilde \rho) \rightarrow (x,\rho)$. More specifically, we will take $(\tilde x,\tilde \rho)=(x+hz,\rho+h\theta)$, and study the limit $h\rightarrow 0$. For the rest of the section, identify $v=\nobreak(z,\theta)\in \reals^{2L}$. Since $\rho+h\theta$ has to be a probability distribution, we require that $\ones^T \theta=0$ and $\theta[i]\ge 0$ whenever $\rho[i]= 0$.

In comparison with section \ref{sec:information_limit}, where we used the $\chi^2$ divergence and the moment tensors, in this section we use the Fisher information matrix and directional derivatives of the moment tensors, respectively. We define the Fisher information matrix as the $2L\times 2L$ matrix such that
\begin{equation*}
\Gamma^N_{x,\rho}:=\Cov[\nabla \log f^N_{x,\rho}].
\end{equation*}
Here $\nabla \log f^N_{x,\rho}\in \reals^{2L}$, since there is a component that depends on $x$ and one that depends on $\rho$.
The Fisher information matrix is also the Hessian of the $\chi^2$ divergence, i.e.,
\begin{equation}\label{eq:chifisherhessian}
\lim_{h\rightarrow 0}\frac{\chi^2(f^N_{x+hz,\rho+h\theta}||f^N_{x,\rho})}{h^2}= v^T \Gamma^N_{x,\rho} v.
\end{equation}
The Fisher information matrix of $N$ observations is related to the one observation version by
\begin{equation}\label{eq:NFisher}
\Gamma^N_{x,\rho}= N \Gamma_{x,\rho}.
\end{equation}
We define the Jacobian $J_{x,\rho}$ as the $L\times 2L$ matrix such that
\begin{equation}\label{eq:jacobiandef}
J_{x,\rho}v=\lim_{h\rightarrow 0}\frac{\E_{x+hz,\rho+h\theta}[\phi_{x}(\widehat X)]-\E_{x,\rho}[\phi_{x}(\widehat X)]}{h}.
\end{equation}
We also define the directional derivative of $M^d_{x,\rho}$ along $v=(z,\theta)$ as the $d$-dimensional tensor
\begin{equation*}\label{eq:dirautocorr}
\nabla_v M^d_{x,\rho}:=\lim_{h\rightarrow 0}\frac{M^d_{x+hz,\rho+h\theta}-M^d_{x,\rho}}{h}.
\end{equation*}
This derivative always exists, an explicit formula for $\nabla_v M^d_{x,\rho}$ is given in Lemma \ref{lem:fourierautocorr}. The next corollary is an analog of the Cram\'er-Rao bound for estimation of an orbit in MRA.
\begin{corollary}\label{cor:CramerRao} 
	For any $v=(z,\theta)\in \reals^{2L}$, such that $\ones^T \theta=0$ and $\theta[i]\ge 0$, whenever $\rho[i]= 0$, we have
	\begin{equation*}
	\Cov[\phi_x(\widehat X)] \succeq \frac{J_{x,\rho}vv^TJ_{x,\rho}^T}{N v^T\Gamma_{x,\rho}v}.
	\end{equation*}
\end{corollary}

\begin{proof}
	\label{proof:CramerRao}
	
	If $\theta$ is under the hypothesis of the theorem, then there exists $h_0>0$ such that for all $0\le h\le h_0$, $\rho+h\theta\in \Delta^L$. Letting $(\tilde x,\tilde \rho)=hv+(x,\rho)$ in Theorem~\ref{thm:ChapmanRobbins} we obtain for any $w\in \RL$
	\begin{align*}
	w^T\Cov[\phi_x(\widehat X)]w\hspace{-30pt}\\
	&\ge \lim_{h\rightarrow 0}\frac{(w^T(\E_{x+hz,\rho+h\theta}[\phi_{x}(\widehat X)]-\E_{x,\rho}[\phi_{x}(\widehat X)]))^2}
	{\chi^2_N(f_{x+hz,\rho+h\theta}||f_{x,\rho})}\\
	&= \frac{(w^TJ_{x,\rho} v)^2}{N v^T \Gamma_{x,\rho} v},
	\end{align*}
	by equations \eqref{eq:chifisherhessian}, \eqref{eq:NFisher} and \eqref{eq:jacobiandef}, and the corollary follows.
\end{proof}

We now use (\ref{eq:probdenstaylor}) to give an expression of the Fisher information in terms of the directional derivative of the tensor moments.

\begin{lem}\label{lem:falldominoesfisher}
	For any $v=(z,\theta)\in \reals^{2L}$,
	\begin{align}\label{eq:Fisheralpha}
	v^T \Gamma_{x,\rho} v&=\frac{\sigma^{-2d}}{(d!)^2}  \E\left[\left(v^T \nabla \alpha_{x,\rho}^d(\error)\right)^2\right]+\OO(\sigma^{-2d-1}),\\
	&=\frac{\sigma^{-2d}}{d!}  \|\nabla_v M^d_{x,\rho} \|^2+\OO(\sigma^{-2d-1}), \label{eq:fisherautocorr}
	\end{align}
	where $d=\inf\left\{n: \|\nabla_v M^n_{x,\rho} \|^2>0\right\}$.
\end{lem}

\begin{proof}
	In this case we cannot just take the limit $h\rightarrow 0$ in \eqref{eq:falldominoes}, since the term contained in $\OO(\sigma^{-2d-1})$ might blow up. Instead we proceed by doing similar algebraic manipulations. Recall that $\nabla f_{x,\rho}(y;\gamma)$ and $\nabla \alpha^i_{x,\rho}(y)$ are in $\reals^{2L}$, with $v^T\nabla f_{x,\rho}(y;\gamma)$ being the directional derivative of $f_{x,\rho}(y;\gamma)$ in the direction $v=(z,\theta)$. We have
	\begin{align*}
	v^T \Gamma_{x,\rho} v &=v^T \Cov[\nabla \log f_{x,\rho}(Y;\gamma)] v\\
	&=\E_{x,\rho}\left[\left(\frac{v^T \nabla f_{x,\rho}(Y;\gamma)}{f_{x,\rho}(Y;\gamma)}\right)^2\right] \\
	&=\int_{\RL} \frac{\left(\sum\limits_{i=0}^\infty v^T \nabla \alpha^i_{x,\rho}(y)\frac{\gamma^i}{i!}\right)^2}
	{\sum\limits_{i=0}^\infty \alpha^i_{x,\rho}(y)\frac{\gamma^i}{i!}} f_\error(y)\,dy
	\end{align*}
	where the second line follows from
	\[ \E_{x,\rho}\left[\frac{\nabla f_{x,\rho}(Y;\gamma)}{f_{x,\rho}(Y;\gamma)}\right]=0 \]
	By the definition of $d$ and \eqref{eq:fisherautocorr}, we have $v^T \nabla \alpha_{x,\rho}^n(z)=0$ almost surely for $n< d$, thus
	\begin{align*}
	v^T \Gamma_{x,\rho} v &=\int_{\RL}\frac{\left(\sum\limits_{i=d}^\infty v^T \nabla \alpha^i_{x,\rho}(y)\frac{\gamma^i}{i!}\right)^2}
	{1+\sum\limits_{i=i}^\infty \alpha^i_{x,\rho}(y)\frac{\gamma^i}{i!}} f_\error(y)\,dy\\
	&= \frac{\gamma^{2d}}{(d!)^2} \int_{\RL} \left(v^T \nabla \alpha_{x,\rho}^d(y)\right)^2 f_\error(y)\,dy+\OO(\gamma^{2d+1})\\
	&=\frac{\gamma^{2d}}{(d!)^2}  \E\left[\left(v^T \nabla \alpha_{x,\rho}^d(\error)\right)^2\right]+\OO(\gamma^{2d+1}),
	\end{align*}
	Equation \eqref{eq:Fisheralpha} now follows since $\gamma=1/\sigma$.

	We now prove \eqref{eq:fisherautocorr} We let $(\tilde x,\tilde \rho)=(x,\rho)+hv$ in \eqref{eq:chi2autocorr} and take the limit $h\rightarrow 0$ to get
	\begin{align*}
	\E\left[\left(v^T \nabla \alpha_{x,\rho}^d(\error)\right)^2\right]\hspace{-1cm}\\
	&=\lim_{h\rightarrow 0}
	\frac{\E\left[\left(\alpha_{x+hz,\rho+h\theta}^d(\error)-\alpha_{x,\rho}^d(\error)\right)^2\right]}{h^2}\\
	&=d!\lim_{h\rightarrow 0}
	\frac{ \|M^d_{x+hz,\rho+h\theta}-M^d_{x,\rho}\|^2}{h^2}\\
	&=d! \|\nabla_v M^d_{x,\rho} \|^2.
	\end{align*}
\end{proof}

Finally, from Corollary \ref{cor:CramerRao} and Lemma~\ref{lem:falldominoesfisher}, we obtain a result analog to Theorem \ref{thm:metalowerbound}.

\begin{corollary}\label{lem:lbmetatheorem}
	For any $v=(z,\theta)\in \reals^{2L}$, such that $\ones^T \theta=0$ and $\theta[i]\ge 0$ whenever $\rho[i]= 0$,	let $Q^n_{v}= \frac1{n!}\|\nabla_v M^d_{x,\rho}\|^2$, $q_{v}=\inf\left\{n:Q^n_{v}>0\right\}$ and $\bar q=\max q_{v}$. 
	Then
	\begin{equation}\label{eq:momentscramerrao}
	\MSE\ge \sup_{v: q_v=\bar q}\left\{\frac{\|z\|^2}{\lambda^{\bar q}_N Q^{\bar q}_{v}
		+\OO\left(\lambda^{\bar q}_N \sigma^{-1}\right)}\right\}.
	\end{equation}
\end{corollary}

\subsection{Proof of Theorem \ref{thm:corsigma}}
\label{proof:corsigma}

Before proving Theorem~\ref{thm:corsigma}, we need the following lemma.

\begin{lem}\label{lem:fourierautocorr}
	
	The entries with index $\kk=(k_1,k_2,\dots,k_d)\in \Z^d_L$ of $M^d_{x,\rho}$ and $\nabla_v M^d_{x,\rho}$ can be explicitly written as 
	\begin{equation}\label{eq:autocorrelation2}
	M^d_{x,\rho}[\kk]:=\sum _{\ell=0}^L \rho[\ell]\prod_{i=1}^d x[k_i\modminus \ell],
	\end{equation}
	and
	\begin{equation}\label{eq:dirautocorr2}
	(\nabla_v M^d_{x,\rho})[\kk]=\sum _{\ell=0}^{L-1}\left(\rho[\ell]\sum_{i=1}^d \frac{z[k_i\modminus \ell]}{x[k_i\modminus \ell]}+\theta[\ell]\right)\prod_{i=1}^d x[k_i\modminus \ell],
	\end{equation}
	where we use the convention $x[k_i\modminus \ell]/x[k_i\modminus \ell]=1$ when $\\x[k_i\modminus \ell]=0$. Moreover, denote the $d$-dimensional Fourier Transform by $\F_d$. For any $\aa=(a_1,a_2,\dots,a_d)\in \Z^d_L$ we have
	\begin{equation}\label{eq:FourierTransformA}
	\F_d M^d_{x,\rho}[\aa]=\F \rho\left[\sum_{j=1}^d a_j\right]\prod_{j=1}^d \F x[a_j],
	\end{equation}
	and
	\begin{align}\label{eq:FourierTransformvA}
	F_d (\nabla_v M^d_{x,\rho})[\aa]= \hspace{-30pt} &\\
	& \left(\sum_{j=1}^d \frac{Fz[a_j]}{Fx[a_j]}+\frac{F\theta\left[\sum_{j=1}^d a_j\right]}{F\rho\left[\sum_{j=1}^d a_j\right]}\right)\F_d M^d_{x,\rho}[\aa],
	\end{align}
	again using the convention $Fx[a_j]/Fx[a_j]=1$ when $Fx[a_j]=0$.
	The denote the $d$-dimensional Fourier Transform preserves the $L_2$ norm of the tensors, i.e.
	\begin{equation}\label{eq:fourierautounit}
	\|M^d_{\tilde x,\tilde \rho}-M^d_{x,\rho}\|^2=\frac{1}{L^d}\|\F_d M^d_{\tilde x,\tilde \rho}-\F_d M^d_{x,\rho}\|^2,
	\end{equation}
	and
	\begin{equation}\label{eq:fourierautounitdir}
	\|\nabla_v M^d_{x,\rho} \|^2=\frac{1}{L^d}\|\F_d\nabla_v M^d_{x,\rho} \|^2.
	\end{equation}
	Also, 
	
\end{lem}

\begin{proof}
	We first prove equation (\ref{eq:autocorrelation2}). By equation (\ref{eq:autocorrelation}), we have 
	\begin{align*}
	M^d_{x,\rho}[\kk]&=\E\left[\prod_{i=1}^d (R_Sx)[k_i]\right]\\
	&=\E\left[\prod_{i=1}^d x[k_i\modminus S]\right]\\
	&=\sum _{\ell=0}^L \rho[\ell]\prod_{i=1}^d x[k_i\modminus \ell].
	\end{align*}
	Equation (\ref{eq:dirautocorr2}) follows from the formula of the derivative of the product:
	\begin{multline*}
	\left(\rho[\ell]\prod_{i=1}^d x[k_i\modminus \ell]\right)'\\
	=\left(\rho'[\ell]+\rho[\ell]\sum_{i=1}^d \frac {x'[k_i\modminus \ell]}{x[k_i\modminus \ell]}\right)\prod_{i=1}^d x[k_i\modminus \ell].
	\end{multline*}
	We finally prove~\eqref{eq:FourierTransformA}; the proof of (\ref{eq:FourierTransformvA}) is analogous.
	\begin{align*}
	\F_dM^d_{x,\rho}[\aa]&=\sum_{\kk\in\Z_L^d}
	M^d_{x,\rho}[\kk]\exp\left(-\frac{2\pi \iota}L \left<\kk,\aa\right> \right)\\
	&= \sum_{\kk\in\Z_L^d}
	\sum _{\ell=0}^{L-1} \rho[\ell]\prod_{j=1}^d x[k_j\modminus \ell]\exp\left(-\frac{2\pi \iota}L k_ja_j \right)\\
	&=\sum_{\kk\in\Z_L^d}
	\sum _{\ell=0}^{L-1} \rho[\ell]\prod_{j=1}^d x[k_j]\exp\left(-\frac{2\pi \iota}L a_j(k_j\modplus \ell)\right)\\
	&= \hspace{-6pt} \sum_{\kk\in\Z_L^d}
	\sum _{\ell=0}^{L-1} \rho[\ell]\prod_{j=1}^d x[k_j]\exp\left(-\frac{2\pi \iota}L (k_ja_j+\ell a_j )\right)\\
	&=\sum _{\ell=0}^{L-1} \rho[\ell]\exp\left(-\frac{2\pi \iota}L \left(\ell\sum_{j=1}^d a_j\right)\right) \prod_{j=1}^d \F x[a_j]\\
	&=\F \rho\left[\sum_{j=1}^d a_j\right]\prod_{j=1}^d \F x[a_j].
	\end{align*}
\end{proof}

We are now ready to prove Theorem~\ref{thm:corsigma}, starting by~\eqref{eq:corsigma4}. Since $\widehat X$ is consistent, $\E_{x,\rho}[\phi_{x}(\widehat X)]\rightarrow x$ and $J_{x,\rho}\rightarrow [I_L ~0_{L\times L}]$ as $N\rightarrow \infty$. By~\eqref{eq:BiasVar} and Corollary~\ref{cor:CramerRao} we have
\begin{align}
\nonumber \lim_{N\rightarrow \infty}N\cdot \MSE \hspace{-20pt}\\
\label{eq:unbiased1}&\ge\lim_{N\rightarrow \infty}\frac{N\tr(\Cov[\phi_x(\widehat X)])}{\|x\|^2}\\
\nonumber &\ge \lim_{N\rightarrow \infty} \frac{\sigma^{2d}}{d!}\frac{\|J_{x,\rho}v\|^2}{\|x\|^2}\frac1{\|\nabla_v M^d_{x,\rho} \|^2} - \OO\left(\sigma^{2d-1}\right)\\
&= \frac{\sigma^{2d}}{d!}\frac{\|z\|^2}{\|x\|^2}\frac1{\|\nabla_v M^d_{x,\rho} \|^2} - \OO\left(\sigma^{2d-1}\right).
\end{align}
We will choose $z=x-\frac{\ones^Tx}L \ones$, and $\theta=\frac1L\ones-\rho$. This choice of $\theta$ is under the theorem assumptions, since $\ones^T \theta=0$ and $\theta[i]=\frac1L\ge 0$ whenever $\rho[i]= 0$. By the linearity of the Fourier transform, this definition is equivalent to $\F z=\F x-\F x[0]\delta_0$ and $\F \theta=\F \rho [0]\delta_0-\F \rho=\delta_0-\F \rho$. Since the $d$-dimensional Fourier Transform is unitary, we can write using Lemma~\ref{lem:fourierautocorr} 
\begin{equation}\label{eq:corfouriereq}
\|\nabla_v M^d_{x,\rho} \|^2=
\frac1{L^d}\sum_{\aa\in\Z_L^d}|\F_d\nabla_v M^d_{x,\rho}[\aa]|^2.
\end{equation}
For $d=1,2$ we have
\begin{equation*}
\F_1\nabla_v M^1_{x,\rho}[a]=\F\rho[a]\F z[a]+\F\theta[a] \F x[a],
\end{equation*}
and 
\begin{align}
\F_2\nabla_v M^2_{x,\rho}[a_1,a_2]\
\nonumber =&\F \rho\left[a_1+a_2\right]\F z[a_1]\F x[a_2]\\
\nonumber &+\F \rho\left[a_1+a_2\right]\F x[a_1]\F z[a_2]\\
\label{eq:cor4acorr2} &+\F \theta\left[a_1+a_2\right] \F x[a_1]\F x[a_2].
\end{align}
Now by our choice of $z$ and $\theta$ we have $\F\rho[a]\F z[a]=-\F\theta[a] \F x[a]$ for all $a\in \Z_L$, so $\|\nabla_v M^1_{x,\rho} \|=0$. On the other hand, by some algebra manipulation of~\eqref{eq:cor4acorr2} we obtain
\begin{align*}
\|\nabla_v M^2_{x,\rho} \|^2\hspace{-30pt}\\
&=\frac{1}{L^2}\left(3\|\F z\|_4^4+\sum_{\aa\in\Z_L^2}|\F\rho[a_1+a_2]\F z[a_1]\F z[a_2]|^2\right)\\
&\le \frac{4}{L^2} \|\F z\|^4\\
&\le 4\|z\|^2\|x\|^2,
\end{align*}
where we used $|\F\rho[a_1+a_2]|\le 1$ and $\|\F z\|_4\le \|\F z\|\le \|\F x\|$, and \eqref{eq:corsigma4} follows.

We now proceed to prove~\eqref{eq:corsigma6}. Suppose that $\rho$ is periodic with period $\ell<\frac L 2$, and let $b=\frac L \ell$, so that $b>2$. Then $\F \rho[k]= 0$ if $b$ does not divide $k$. For a positive integer $i\le \lceil\frac{b-2}2\rceil$, define $z_i\in \RL$ such that 
$$\F z_i[k]=\left\{\begin{array}{rl}
\medskip  \F x[k]\iota&\text{if }b|k-i,\\
\medskip -\F x[k]\iota&\text{if }b|k+i,\\
0&\text{otherwise,}\end{array}\right.$$
where $b|k$ means that $b$ divides $k$.
Assume $z_i\neq 0$, let $\theta_i=0_L$ and $v_i=(z_i,\theta_i)$. Since $\widehat X$ is consistent and $\{z_i\}_{1\le i\le \lceil\frac{b-2}2\rceil}$ is a set of orthogonal vectors, we have by~\eqref{eq:unbiased1} and Corollary~\ref{cor:CramerRao}:
\begin{align*}
\lim_{N\rightarrow \infty}N\cdot \MSE \hspace{-20pt}\\
&\ge\lim_{N\rightarrow \infty}\frac{N\tr(\Cov[\phi_x(\widehat X)])}{\|x\|^2}\\
&\ge \lim_{N\rightarrow \infty}\frac1{\|x\|^2} \sum_{i=1}^{\lceil\frac{b-2}2\rceil} \frac{N z_i^T \Cov[\phi_x(\widehat X)]z_i}{\|z_i\|^2} \\
&\ge \frac{1}{\|x\|^2}\sum_{i=1}^{\lceil\frac{b-2}2\rceil}  \frac{\sigma^{2d_i}}{d_i!}\frac{\|z_i\|^2}{\|\nabla_{v_i} M^{d_i}_{x,\rho} \|^2} - \OO\left(\sigma^{2d_i-1}\right),
\end{align*}
where $d_i=\inf\left\{n:\|\nabla_{v_i} M^n_{x,\rho} \|^2>0\right\}$. Recalling equation~\eqref{eq:corfouriereq} and since $\theta_i=0$, we have now for $d=1,2,3$,
\begin{equation}\label{eq:cor6acorr1}
\F_1\nabla_{v_i} M^1_{x,\rho}[a]=\F\rho[a]\F z_i[a],
\end{equation}
\begin{multline}\label{eq:cor6acorr2}
\F_2\nabla_{v_i} M^2_{x,\rho}[a_1,a_2]=\F \rho\left[a_1+a_2\right](\F z_i[a_1]\F x[a_2]\\
+\F x[a_1]\F z_i[a_2]),
\end{multline}
and
\begin{align}
\nonumber \F_3\nabla_{v_i} M^3_{x,\rho}[a_1,a_2,a_3]=&\F \rho\left[a_1+a_2+a_3\right]\\
\nonumber &\quad (\F z_i[a_1]\F x[a_2]\F x[a_3]\\
\nonumber &\quad +\F x[a_1]\F z_i[a_2]\F x[a_3]\\
&\quad +\F x[a_1]\F x[a_2]\F z_i[a_3]). \label{eq:cor6acorr3}
\end{align}
Since $\F\rho[a]\neq 0\Rightarrow b|a \Rightarrow \F z_i[a]=0$, $(\ref{eq:cor6acorr1})=0 ~\forall a\in\Z_L$ . Also $\F\rho[a_1+a_2]\neq 0$ implies $b|a_1+a_2$. Let $\tilde a_j=\mod(a_j,b)$ for $j=1$ and $2$. Since $b|a_1+a_2$, $\tilde a_1+\tilde a_2=b$, so assume with out loss of generality that $\tilde a_1\le \frac b2$. If $\tilde a_1\neq i$, then $\F z[a_1]=\F z[a_2]=0$. On the other hand, if $\tilde a_1= i$, then
\begin{align*}
\F z[a_1]\F x[a_2]+\F x[a_1]\F z[a_2]\hspace{-25pt}\\
&=\iota\F x[a_1]\F x[a_2]-\iota\F x[a_1]\F x[a_2]\\
&=0,
\end{align*}
so~\eqref{eq:cor6acorr2}$=0~\forall \aa\in\Z^2_L$. Finally since $|\F\rho[\cdot]|\le 1$ we have
\begin{align*}
\|\nabla_{v_i} M^3_{x,\rho} \|^2&\le \frac9{L^3}\sum_{\aa\in\Z_L^3}|\F z_i[a_1]\F x[a_2]\F x[a_3]|^2\\
&=9\|z_i\|^2\|x\|^4,
\end{align*}
and the result follows. Finally, if $z_i=0$, we can alternatively choose 
$$\F \tilde z_i[k]=\left\{\begin{array}{rl}
\medskip  \iota&\text{if }b|k-i,\\
\medskip -\iota&\text{if }b|k+i,\\
0&\text{otherwise.}\end{array}\right.$$ 
We still have $\eqref{eq:cor6acorr1}=0 ~\forall a\in\Z_L$ and $\eqref{eq:cor6acorr2}=0 $ for all $\aa\in\Z^2_L$ except if $\tilde a_1=i$. But $z_i=0$ implies $\F x[a]=0$ if $\mod(a,b)=\pm i$, so $\eqref{eq:cor6acorr2}=0 $ also if $\tilde a_1=i$.

\subsection{Proof of Theorem~\ref{thm:uniqueness}} \label{sec:proof_thm_uniqueness}

We show that if the first two moments of two pairs, signal and distribution, are equal then the pairs are identical up to a translation. Specifically, suppose that $x_1$ and $\rho_1$ have the same first two moments as $x_2$ and $\rho_2$. Equality of the first moments means that $x_1 \ast \rho_1=  x_2 \ast \rho_2$, and therefore:
\begin{align*}
    (\F x_1)[k] \cdot (\F\rho_1)[k] = (\F x_2)[k] \cdot (\F \rho_2)[k].
\end{align*}
Since $\F x_1$ is non-vanishing, we define the ratio 
\begin{align*}
    r[k] = \frac{(F x_2) [k]}{(F x_1) [k]}.
\end{align*}
Then,
\begin{equation} \label{eqn:first_moment_relation}
    (\F \rho_1) [k] = (\F \rho_2) [k] \cdot r[k].
\end{equation}

Furthermore, from the equality of second moments $C_{x_1} D_{\rho_1} C_{x_1}^T = C_{x_2} D_{\rho_2} C_{x_2}^T $, or equivalently (after taking Fourier transforms) $D_{\F x_1} C_{\F \rho_1} D_{\F x_1}^* = D_{\F x_2} C_{\F \rho_2} D_{\F x_2}^*$. Consequently, for $k$, $p =0,\ldots ,L-1$:
\begin{align*} 
    (\F x_1)[k] \cdot  (\F\rho_1) & [k-p] \cdot (\F x_1)[p]^\ast   \nonumber \\
                           &= (\F x_2)[k] \cdot (\F \rho_2)[k-p]\cdot  (\F x_2)[p]^\ast , 
\end{align*}
or equivalently,
\begin{equation} \label{eqn:from_second_moment}
(\F\rho_1)[k-p] = (\F \rho_2)[k-p] \cdot r[k] \cdot r[p]^\ast.
\end{equation}
Because $\rho_1$ and $\rho_2$ are probability distributions, $(\F\rho_1)[0] = (\F\rho_2)[0]=1$. Therefore, taking $k=p$ in \eqref{eqn:from_second_moment} implies $\abs{r[k]}=1$. By ~\eqref{eqn:first_moment_relation}, $r[0] = 1$, and $\F \rho_1$ and $\F \rho_2$ have the same support.

We will denote by $ \operatorname{GCD}\left(a_1,\dots,a_\ell \right)$ the greatest common divisor of the positive numbers $a_1,\dots,a_\ell$.

\begin{lem} \label{lemma:aperiodicCond}
If a distribution $\rho$ is aperiodic then
\begin{align*}
    \operatorname{GCD}\left( \{ k \mid 1\le k \le L , \ (\F \rho)[k]\neq 0 \} \right) = 1 .
\end{align*}
\end{lem}
\begin{proof}[Proof of Lemma~\ref{lemma:aperiodicCond}]
A necessary and sufficient condition for a distribution $\rho$ to have period $\ell$ is that $(\F \rho)[m] \neq 0$ only for $m$ of the form $ k (L/\ell)$, $k=0,1,\ldots,\ell-1$. Therefore, the aperiodicity of a distribution $\rho$ means that the shared greatest common divisor of all the indices of nonzero entries in $\F\rho$ (which includes $L$, since $\rho[0]= \rho[L]=1$) is $1$. In fact, if the GCD were equal to some $d>1$, then the distribution would be periodic with a period of $ L/d $ as all nonzero entries would be of the form $kd$, $k\in\{0,1,\ldots L/d\}$.
\end{proof}

Let $m_1,\dots,m_\ell$ be the indices of the support of $\F \rho_1$ (and $\F \rho_2$). Because the greatest common divisor $ \operatorname{GCD}$ is associative -- that is, $ \operatorname{GCD}(a,b,c) = \operatorname{GCD}(\operatorname{GCD}(a,b),c)$ -- by Lemma~\ref{lemma:aperiodicCond} there exist integers $a_1,\ldots,a_\ell$ such that
\begin{equation} \label{eqn:sum_of_nonzeros_indices}
    \sum_{j=1}^n a_j m_j = 1  \mod L .
\end{equation}

Taking $k-p = m_j$ in \eqref{eqn:from_second_moment}, we obtain:
\begin{align}  \label{eq123}
    r[p+m_j] = \tilde{\omega}_j \cdot r [p]
\end{align}
where 
\begin{align*}
    \tilde{\omega}_j = \frac{(\F \rho_1) [m_j]}{(\F \rho_2) [m_j]}.
\end{align*}
From \eqref{eqn:sum_of_nonzeros_indices}, repeated application of \eqref{eq123} yields:
\begin{align} \label{eq777}
    r[p+1] = \tilde{\omega}_1^{a_1} \cdots \tilde{\omega}_\ell^{a_\ell} \cdot r[p]
         = \omega \cdot r[p],
\end{align}
where $\omega = \tilde{\omega}_1^{a_1} \cdots \tilde{\omega}_\ell^{a_\ell}$. Repeatedly applying \eqref{eq777}, we obtain $r[m] = \omega^m r[0] = \omega^m$, or equivalently:
\begin{align} \label{shift_fourier}
    (\F x_2)[m] = \omega^m \cdot (\F x_1)[m].
\end{align}
Furthermore, when $m = L$, we see:
\begin{align*}
    1 = r[0] = r[L]= \omega^L \cdot r[0] = \omega^L,
\end{align*}
i.e., $\omega$ is an $L^{th}$ root of unity. Equation \eqref{shift_fourier} then implies $x_2$ is a translation of $x_1$. Finally, \eqref{eqn:first_moment_relation} then shows that $(\F \rho_1)[m] = \omega^m (\F \rho_2)[m]$, so that $\rho_1$ is also a translation of $\rho_2$. This completes the proof.

\subsection{Proof of Lemma~\ref{prop:reshuffling}} \label{sec:proof_prop_reshuffling}

For any $0 \le i \le L-1$, we can write 
\[ (\rho \ast \theta)[i] = e_i^T C_\rho \theta , \]
with $e_i$ the unit vector with one in its $i$th entry. Consequently, equality of two distinct entries $i$ and $j$ implies 
\begin{equation} \label{eqn:cond_same_entry_dist}
(e_i - e_j)^T C_\rho \theta = 0. 
\end{equation}
However, for a random choice of $\theta$, if~\eqref{eqn:cond_same_entry_dist} holds with non-zero probability, then 
\[ (e_i - e_j)^T C_\rho = 0 , \]
or, 
\[C_\rho^T e_i = C_\rho^T e_j.\]
The latter implies that $\rho$ shifted by $i$ equals $\rho$ shifted by $j$, i.e., $\rho[k-i] = \rho[k-j]$, or 
\[ \rho[k] = \rho[k + i-j], \quad \forall k. \] 
Therefore, $\rho$ is periodic.

\subsection{Proof of Corollary~\ref{cor:stability}} \label{sec:proof_cor_stability}

Throughout the proof, $C$ will always denote a constant depending on $x$ and $\rho$ that may change value from occurrence to occurrence. Let $\hat P_x = L \diag(F \hatM^2 F^{-1})$ denote the estimated power spectrum of $x$, and $\hat p = (\hat P_x)^{-1/2}$. Because $\|\hatM^2 - M^2\|_F \le \varepsilon$, for $\varepsilon$ sufficiently small we must have $\|\hat p - p\| \le C \varepsilon$. Setting $\widehat Q = F^{-1} D_{\hat p} F$, we also have $\|\widehat Q - Q\|_F \le C \varepsilon$. Consequently, the matrix $ \widetilde M_{est}^2 = \widehat Q \hatM^2 \widehat Q^*$ is within $C\varepsilon$ of $\widetilde M^2 = Q M^2 Q^*$, i.e.\ $\|\widetilde M^2_{est} - \widetilde M^2\|_F \le C \varepsilon$.

Let $v_{est}$ denote the top eigenvector of $\widetilde M_{est}^2$, and $v$ the top eigenvector of $\widetilde M^2$. The eigenvalues of $\widetilde M^2$ are the values of $\rho$, which are distinct; let $\Delta > 0$ denote the gap between the first and second eigenvalues. We may apply Theorem 2 of \cite{yu2015useful} to the matrices $\widetilde M^2$ and $\widetilde M_{est}^2$ to say that the sine of the angle $\theta(v,v_{est})$ between $v_{est}$ and $v$ satisfies the following bound:
\begin{align}
\sin(\theta(v,v_{est})) \le 2\frac{\|\widetilde M^2 - \widetilde M_{est}^2\|}{\Delta} 
\le C \varepsilon.
\end{align}
Defining $\tilde v_{est} = \F^{-1}\left(  (\hat P_x)^{1/2} \odot \F v_{est}  \right)$ and $\tilde v = \F^{-1}\left(  (P_x)^{1/2} \odot \F v  \right)$, because $\|\hat P_x - P_x\| \le C\varepsilon$ and $F$ is unitary we also have 
\begin{math}
\sin(\theta(\tilde v,\tilde v_{est})) \le  C \varepsilon.
\end{math}
%
%
We may therefore write $\tilde v_{est} = \eta \tilde v + \tilde u$, where $\eta = \pm 1$ and $\|\tilde u\| \le C \varepsilon$. Consequently, $|\operatorname{Sum}(\tilde v_{est}) - \eta\operatorname{Sum}(\tilde v)| \le C \varepsilon$.

Furthermore, because $\|\hatM^1 - M^1\| \le \varepsilon$, $|\operatorname{Sum}(\hatM^1) - \operatorname{Sum}(M^1)| \le C \varepsilon$ too; and consequently,
\begin{align}
|\eta \operatorname{Sum}(M^1) / \operatorname{Sum}(\tilde v) 
- \operatorname{Sum}(\hatM^1) / \operatorname{Sum}(\tilde v_{est})| 
\le C \varepsilon.
\end{align}
Since $x=\operatorname{Sum}(M^1) \tilde v/ \operatorname{Sum}(\tilde v)$, by defining $\widehat{X}_{\text{Spectral}} = \operatorname{Sum}(\hatM^1)\tilde v_{est} / \operatorname{Sum}(\tilde v_{est}) $, we therefore have $\| \widehat{X}_{\text{Spectral}} -\nobreak x\| \le\nobreak C \varepsilon$, as claimed.

\subsection{Proof of Theorem~\ref{thm:exp_decay}}
\label{sec:proof_spectral l2 convergence}

Since the residuals $\hatM^1 - M^1$, $\hatM^2 - M$ and $\widehat{P}_{x} - P_{x}$ are subexponential, we can apply the Bernstein-type inequality for subexponential random variables found in~\cite{vershynin2010intro}, together with Corollary~\ref{cor:stability}, to obtain
\begin{multline}
\mathbb{P} \left[  \min_{s\in\Z_L} \| R_s\widehat{X}_{\text{Spectral}} - x \|^2 \ge t \right]\\
\le C_1\exp\left(-\frac{N}{\sigma^4} \min\left\{\frac{t}{C_2},\frac{\sqrt t}{C_3}\right\} \right)  \label{eq:bernstein},	
\end{multline}
where $C_1$, $C_2$ and $C_3$ are finite, positive constants that depend on $x$ and $\rho$. We have
\begin{align}
\nonumber \MSE \cdot \|x\|^2 &=\E \left[  \min_{s\in\Z_L} \| R_s\widehat{X}_{\text{Spectral}} - x \|^2 \right]\\
\nonumber &=\int_{0}^{\infty} \hspace{-.2cm}\mathbb{P} \left[  \min_{s\in\Z_L} \| R_s\widehat{X}_{\text{Spectral}} - x \|^2 \ge t \right]dt\\
\nonumber &\le C_1\int_{0}^{\infty} \hspace{-.2cm}\exp\left(-\frac{N}{\sigma^4} \min\left\{\frac{t}{C_2},\frac{\sqrt t}{C_3}\right\} \right) dt\\
&=\text{\small $\displaystyle C_4\frac{ \sigma^4}{N}\left[C_5+\left(C_5+2\frac{\sigma^4}{N}\right)\exp\left(-C_5\frac{N}{\sigma^4}\right)\right],$} \label{eq:bernsteinerror}
\end{align}
with $C_4=C_1C_3^2$ and $C_5=C_2/C_3^2$, thus if $N=\omega(\sigma^4)$, \eqref{eq:bernsteinerror} converges to $0$ as $n$ diverges, and $\widehat{X}_{\text{Spectral}}$ converges to the true signal in $L^2$, up to a cyclic shift.

\subsection{Proof of Proposition~\ref{prop:counterexample}} \label{sec:proof_prop_counterexample}

It is clear that, as $L>1$, $x_1 \neq x_2$. In addition, since $x_1$ is real, the construction ensures that $x_2$ is real as well.

The $\ell$ periodicity of $\rho$ means a sparsity pattern for $\F{\rho}$. Particularly, $\F{\rho}$ is zero everywhere besides 
\begin{equation} \label{eqn:sparsityRhoFourier}  
(\F{\rho})\left[kL/\ell \right] \neq 0  \quad \Longleftrightarrow \quad kL/\ell  \text{  is integer} , 
\end{equation}
for $k=0,\ldots,\ell-1$.
It is easy to verify that 
\begin{equation*}  \label{eqn:ConditionOfPeriodicityFirstMoment}
(\F{x_1})[k] (\F{\rho})[k] = (\F{x_2})[k] (\F{\rho})[k]  , \quad k=0,\ldots,L-1.
\end{equation*}
Therefore, $x_1$ and $x_2$ share the same first moment.

For the second moments, we will show the equality 
\[   C_{x_1} D_\rho C_{x_1}^T =  C_{x_2} D_\rho C_{x_2}^T  . \]
Applying the Fourier matrix, due to the realness of $\rho$, the latter is equivalent to 
\[ D_{\F{x_1}} C_{\F{\rho}} D_{\F{x_1}} =  D_{\F{x_2}} C_{\F{\rho}} D_{\F{x_2}},\] 
Similar to~\eqref{eqn:from_second_moment} and by the sparsity pattern of~\eqref{eqn:sparsityRhoFourier}, this equality should hold only if
\begin{equation*} \label{eqn:ConditionOfPeriodicitySecondMoment}
(\F{x_1})[i] \, (\F{x_1})\left[i+tL/\ell \right]^\ast   = (\F{x_2})[i] \, (\F{x_2})\left[ i+tL/\ell \right]^\ast , 
\end{equation*}
for all $t=0,\ldots,\ell$ and $i=0,\ldots,L-1$. 
By the construction~\eqref{eq:construction_example}, this equation holds true. 

\subsection{Proof of Claim~\ref{claim:L_2_period}} \label{sec:proof_L_2_period} 

Throughout the proof, we assume that each period has no repeated values. This property is guaranteed by reshuffling the measurements with random $\theta \in \Delta_L$; see Lemma~\ref{prop:reshuffling}. Additionally, we can obtain the power spectrum of $x$ from the second moment \eqref{eq:psfrommoment}, which we can then factor out as in \eqref{eq:tilde_x}. Thus, we can assume, without loss of generality, that $\vert \F x\vert [k]=1$ for all $k$. 

Observe that both $x$ and $R_{L/2}x$ are eigenvectors of $\hatM^2= C_xD_\rho  C_x^T$ (we assume exact knowledge of the moments) with the same eigenvalue. Also, $x$ and $R_{L/2} x$ are orthogonal as columns in the orthogonal matrix $C_{x}$. Then, if $u$ is an eigenvector, we can write for some scalars $\alpha,\beta\in\RL$:
\[ u = \alpha x+ \beta R_{L/2} x  , \]
and therefore,
\[  R_{L/2}u = \alpha R_{L/2} x + \beta x , \]
as $R_{L/2} = R_{L/2}^{-1}$. Then, one can verify that the inner product of $u$ and $R_{L/2}u$ is $2\alpha \beta \|x\|^2$. Since the signals are orthogonal, their inner product is zero. This means that $\alpha$ or $\beta$ must be zero. This in turn implies that $u$ was either $x$ or $R_{L/2} x$ in the first place. Therefore, $x$ is the unique eigenvector of $\hatM^2$ that is orthogonal to its translation by $L/2$. This completes the proof.

\subsection{Convex relaxation with semidefinite program } \label{subsec:SDPmethod}

In this section, we propose an additional algorithm for non-uniform MRA based on a semidefinite program (SDP) relaxation.

Since the power spectrum of the signal can be estimated from the data at sample complexity scaling as $\omega(1/\SNR^2)$ according to~\eqref{eq:powerspectrum}, we assume in this section, without loss of generality, that $\vert \F x\vert[k] =1$ for all $k$. Note, that as in Algorithm~\ref{alg:SpectralAlg}, the normalization is done on the second moment matrix, not the individual observations, in order to retain the noise statistics.

The SDP relaxation is based on considering the second moment matrix in the Fourier domain, namely,
\begin{equation} \label{eq:My_prime} 
M^2_\ast = \F\left( M^2 \right) \F^{-1} = D_{\F x} C_{\F \rho }^T D_{\F x}^\ast. 
\end{equation}
The last expression can be also written as 
\begin{equation*} 
M^2_\ast =  C_{\F \rho }^T \odot (\F x \F x{^\ast}),  
\end{equation*}
or 
\begin{equation}  \label{eq:1}
M^2_\ast \odot \overline{X}=  C_{\F \grave\rho },
\end{equation}
where $X = (\F x) (\F x)^\ast$. and $\grave\rho :=\F ^{-1}(\overline{\F {\rho}})$. 

The formulation of~\eqref{eq:1} suggests to pose the recovery problem as,
\begin{equation} \label{eq:nonconvexSDP} 
\begin{aligned}
& \min_{\tilde{\rho},\tilde{X}}
& & \norm{ \hatM^2_\ast \odot\overline{\tilde{X}}   - C_{\F\tilde{\rho}} }_{\textrm{F}}^2 \\
& \text{subject to}
& & \diag(\tilde{X}) = 1, \quad \rank (\tilde{X}) = 1,  \\
&&& \tilde{X}[1,0] = 1 ,\quad  \tilde{X} \succeq 0,  \quad \tilde{\rho}[0] = 1 , \\
&&& \tilde{\rho}[k] = \overline{\tilde{\rho}[-k]}, \, \forall k .
\end{aligned} 
\end{equation}
The constraint $\tilde{X}[1,0] = 1$ follows the assumption that 
$(\F x)[0]=(\F x)[1]=1$. While we can easily estimate $(\F x)[0]$ and therefore fix it, the assumption of fixed $(\F x)[1]=1$ is more delicate. 
Recall that the solution for the MRA problem is always up to cyclic translation. In the Fourier domain, it means that the first entry of the Fourier transform of the signal is determined up to an arbitrary modulation by $e^{2\pi \iota \ell /L}$ for some $\ell\in\mathbb{Z}$. If $L\to\infty$, this allows us to fix this coefficient arbitrarily.

Similarly to the well-known SDP relaxation of the Max-Cut problem~\cite{goemans1995improved}, the non-convex problem~\eqref{eq:nonconvexSDP} can be relaxed to a convex program by omitting the rank constraint as follows,
\begin{equation} \label{eq:opt}
\begin{aligned}
& \min_{\tilde{\rho},\tilde{X}}
& & \norm{ \hatM^2_\ast  \odot\overline{\tilde{X}}   - C_{\F\tilde{\rho}} }_{\textrm{F}}^2 \\
& \text{subject to}
& & \diag(\tilde{X}) = 1, \quad \tilde{X}[1,0] = 1, \\
&&&  \tilde{X} \succeq 0,  \quad \tilde{\rho}[0] = 1 ,\quad   \tilde{\rho}[k] = \overline{\tilde{\rho}[-k]}, \, \forall k .
\end{aligned} 
\end{equation}
This relaxation is convex and can be solved in polynomial time using off--the--shelf software, such as CVX~\cite{grant2008cvx}.

The SDP relaxation~\eqref{eq:opt} recovers the Fourier phases of the signal and the distribution exactly for $N\to\infty $ and fixed noise level, since in this regime we can estimate the first two moments arbitrarily well.

\begin{thm} \label{thm:sdp}
    Assume that $\vert \F x\vert [k]=1$ for all $k$ and that $\F\rho$ is non-vanishing. In addition, assume that $(\F x)[0] = (\F x)[1]=1$.
    Then, if $N\to\infty$ and $\sigma$ is fixed, the solution of~\eqref{eq:opt} is given by $\tilde{X} = (\F x) (\F x)^* $ and $\tilde{\rho} =  \F \grave\rho$.  
    \begin{proof}

Since $\sigma$ is fixed and $N\to\infty$, one can estimate $M^2_\ast$ as in~\eqref{eq:My_prime}  exactly.
    Then, since~\eqref{eq:opt} admits at least one solution (the underlying signal and distribution), the objective is zero at the solution and we get the relation:
\begin{equation} \label{eq:eq} 
C_{\tilde{\rho}} = M^2_\ast \odot\overline{\tilde{X}}  =  C_{\F \grave\rho}\odot(\F x\F x^*) \odot\overline{\tilde{X}}, 
\end{equation}
where we use $\grave\rho :=\F ^{-1}(\overline{\F \tilde{\rho}}) $. 
Let $u = \tilde{\rho}/\F \grave\rho$. Since $\tilde{X}\succeq 0$ we conclude that 
$C_{u}\succeq 0$ and hence $\F u\geq 0$ (the Fourier transform of $u$ is non-negative). By the constraints of~\eqref{eq:opt}, we also have $u[0]=1$. By examining the $(1,0)$th entry of~\eqref{eq:eq}, we also conclude that
\begin{equation*}
(\F x)[1]\overline{(\F x)[0]}(\F \grave\rho)[1] \overline{\tilde{X}[1,0]} = \tilde{\rho}[1] \Rightarrow u[1] = \overline{\tilde{X}[1,0]} = 1, 
\end{equation*}
where the last equality holds because of the constraints of~\eqref{eq:opt}.

Until now, we have shown that the vector $u$ satisfies $u[0]=u[1]=1$, it is conjugate-symmetric and its Fourier transform is non-negative. Therefore, by Lemma IV.2 of~\cite{bendory2017bispectrum}, we conclude that $u [n]= 1$ for all $n$, or $\tilde{\rho} = \F \grave\rho$. 
Next, we substitute $\tilde{\rho} = \F \grave\rho$ in~\eqref{eq:eq} and get 
\begin{equation*} 
1   = (\F x\F x^*) \odot\overline{\tilde{X}},
\end{equation*}
where the equality holds entry-wise.
Since all entries of $\hat{x}$ are normalized, we conclude that $\tilde{X}= (\F x) (\F x)^*$. This concludes the proof.
\end{proof}
\end{thm}

\subsection{Proof of Lemma~\ref{lem:lemma_EM}} \label{sec:proof_lemma_EM}

It is easy to check that the condition $q[\ell]> 0$ is automatically enforced whenever $w[\ell] > 0$ (otherwise the objective is $-\infty$). So the simplex constraint is equivalent to $\sum_{\ell=0}^{L-1}q[\ell] = 1$. The Lagrangian for this problem is the function:
\begin{align*}
\L(q,\nu) = \sum_{\ell=0}^{L-1} w[\ell] \log(q[\ell]) + \nu \left(1 - \sum_{\ell=0}^{L-1} q[\ell] \right),
\end{align*}
and the KKT conditions imply $q^*[\ell] = \frac{w[\ell]}{\nu^*}$.
Since $q$ is on the simplex, we conclude that $\nu^* = \sum_{\ell^\prime=0}^{L-1}w[\ell^\prime]$.

\end{document}